\documentclass[11pt,reqno]{amsart}

\usepackage[margin=1.1in]{geometry}
\usepackage[full]{textcomp}
\usepackage[osf]{newtxtext}
\usepackage{comment}

\usepackage{amssymb}
\usepackage{mathtools}
\usepackage{hyperref}
\usepackage{breakurl}
\usepackage{mhequ}
\usepackage{mhsymb}
\usepackage{booktabs}
\usepackage{tikz}
\usetikzlibrary{decorations.pathreplacing}
\usetikzlibrary{decorations.markings}
\usetikzlibrary{decorations.pathmorphing}
\usetikzlibrary{snakes}
\usepackage{mathrsfs}
\usepackage{longtable}
\usepackage{scalerel}
\usepackage{cprotect}
\usepackage{multirow}
\usepackage{colortbl}
\usepackage{makecell}
\usepackage{resizegather}
\usepackage{enumerate}

\newtheorem{lemma}{Lemma}[section]
\newtheorem{theorem}[lemma]{Theorem}
\newtheorem{proposition}[lemma]{Proposition}
\newtheorem{corollary}[lemma]{Corollary}
\newtheorem{definition}[lemma]{Definition}

\newtheorem{remark}[lemma]{Remark}

\numberwithin{equation}{section}
\def\theequation{\arabic{section}.\arabic{equation}}

\def\mD{{\mathbb D}}
\def\mE{{\mathbb E}}
\def\mF{{\mathbb F}}
\def\mG{{\mathbb G}}

\def\mM{{\mathbb M}}

\def\mS{{\mathbb S}}
\def\mT{{\mathbb T}}

\def\mV{{\mathbb V}}

\def\tr{\mathrm {tr}}
\def\no{\nonumber}
\def\cA{{\mathcal A}}

\def\cD{{\mathcal D}}

\def\cF{{\mathcal F}}

\def\cL{{\mathcal L}}

\def\ee{\boldsymbol{\epsilon}}

\newcommand{\mfg}{\mathfrak{g}}

\def\bt{\begin{theorem}}
	\def\et{\end{theorem}}
\def\bl{\begin{lemma}}
	\def\el{\end{lemma}}
\def\br{\begin{remark}}
	\def\er{\end{remark}}
\def\bx{\begin{Examples}}
	\def\ex{\end{Examples}}
\def\bd{\begin{definition}}
	\def\ed{\end{definition}}
\def\bp{\begin{proposition}}
	\def\ep{\end{proposition}}
\def\bc{\begin{corollary}}
	\def\ec{\end{corollary}}
\def\<{{\langle}}
\def\>{{\rangle}}
\def\p{\partial}

\def\dif{{\mathord{{\rm d}}}}
\def\Tr{\mathrm{Tr}}

\def\hol{\textnormal{hol}}

\def\neg{\textnormal{\tiny neg}}
\def\pos{\textnormal{\tiny pos}}

\definecolor{codegreen}{rgb}{0,0.6,0}
\definecolor{codegray}{rgb}{0.5,0.5,0.5}
\definecolor{codepurple}{rgb}{0.58,0,0.82}
\definecolor{backcolour}{rgb}{0.95,0.95,0.92}

\tikzset{
    midarrow/.style={postaction={decorate,decoration={markings,mark=at position 0.6 with {\arrow{>}}}}},
    midarrow1/.style={postaction={decorate,decoration={markings,mark=at position 0.45 with {\arrow{>}}}}},
    midarrow2/.style={postaction={decorate,decoration={markings,mark=at position 0.75 with {\arrow{>}}}}},
}

\colorlet{darkblue}{blue!90!black}
\colorlet{darkred}{red!90!black}
\colorlet{darkgreen}{green!70!black}

\title{Makeenko-Migdal equations for 2D Yang--Mills: \\ from lattice to continuum}

	\author{Hao Shen}
	\address[H. Shen]{Department of Mathematics, University of Wisconsin - Madison, USA}
	\email{pkushenhao@gmail.com}
	
	\author{Scott A. Smith}
	\address[S. A. Smith]{Academy of Mathematics and Systems Sciences, Chinese Academy of Sciences, Beijing, China
	}
	\email{ssmith@amss.ac.cn}
	
	\author{Rongchan Zhu}
	\address[R. Zhu]{Department of Mathematics, Beijing Institute of Technology, Beijing 100081, China 
	}
	\email{zhurongchan@126.com}

\begin{document}

\maketitle

\markboth{}{} 

\begin{abstract}
In this paper, we prove the convergence of 
the
 discrete Makeenko--Migdal equations for Yang--Mills model on $(\eps \Z)^{2}$ to their continuum counterparts on the plane, in an appropriate sense.
  The key step in the proof is identifying the limits of the contributions from deformations as the 
  area derivatives of the Wilson loop expectations.
\end{abstract}

\setcounter{tocdepth}{2}
\tableofcontents

\section{Introduction}
\label{sec:intro}
The goal of the paper is to prove that for the  two dimensional Yang--Mills model, the lattice master loop equations converge to the continuum ones. These equations are also known as Makeenko--Migdal equations or Dyson--Schwinger equations. We first recall the basic settings for the lattice and the continuum Yang--Mills models.

For $\eps>0$,
let $\Lambda$ be a finite subset of $(\eps\Z)^d$.
Let $E^+$ and $E^-$ be the set of positively and negatively oriented bonds of $(\eps\Z)^d$,
and  denote by $E_\Lambda^+$, $E_\Lambda^-$ the corresponding subsets
of bonds with both beginning and ending points in $\Lambda$. Define $E\eqdef E^+\cup E^-$. 
For $e\in E$, let $u(e)$ and $v(e)$ denote its starting point and ending point respectively. 
A path  in  $(\eps\Z)^d$ is defined to be a sequence of bonds $e_1e_2\cdots e_n$ with $e_i\in E$ and $v(e_i)=u(e_{i+1})$ for $1\le i \le  n-1$. It is called a closed path if $v(e_n)=u(e_1)$.  

A plaquette  is a closed path of length $4\eps$ which traces out the boundary of a square.  
Denote $\CP^+_\Lambda$ for the set of plaquettes 
$p=e_1e_2e_3e_4$ such that all the vertices of $p$ are in $\Lambda$ and
$u(e_1)$ is lexicographically the smallest vertex  and  $v(e_1)$ is the second smallest.

We fix a Lie group $G=U(N)$, and write $\mfg$ for its Lie algebra with the inner product given by
\begin{align}\label{def:inn}
	\<X,Y\>\eqdef N\Tr(XY^*),\qquad X,Y\in \mfg.
\end{align}
Throughout the paper 
we will write $\Tr$ for the usual trace and $\tr=\frac1N\Tr$.

The lattice Yang-Mills theory
on $\Lambda$  with $\beta\in\R$\footnote{Typically, one takes $\beta \geq 0$ although the measure is well-defined even for $\beta<0$.} the inverse coupling constant, is the
probability measure $\mu^\eps_\Lambda$  on the set of all collections $Q = (Q_e)_{e\in E_\Lambda^+}$ of $G$-matrices, defined as 
\begin{equation}\label{measure}
	\dif\mu^\eps_\Lambda(Q)
	:= \frac{1}{Z_{\Lambda}^{\eps} }
	\exp\biggl(-N\beta \, \Re \sum_{p\in \CP^+_\Lambda} \Tr(I-Q_p)\biggr)
	 \prod_{e\in E^+_\Lambda} \dif Q_e\, ,
\end{equation}
where $Z_{\Lambda}^{\eps} $ is the normalizing constant,  $Q_p \eqdef Q_{e_1}Q_{e_2}Q_{e_3}Q_{e_4}$ for a plaquette $p=e_1e_2e_3e_4$, and $\dif Q_e$ is the Haar measure on $G$.
Note that for $p\in \CP^+_\Lambda$ the bonds $e_{3}$ and $e_{4}$ are negatively oriented;
throughout the paper we define $Q_{e}\eqdef Q_{e^{-1}}^{-1}$ for $e \in E^{-}$, where $e^{-1}\in E^+$ denotes the bond with orientation reversed.

The measure $\mu^\eps_\Lambda$ is invariant under any gauge transformation 
$g:\Lambda\to G$
defined by 
$Q_e\mapsto g_{u(e)} Q_e g_{v(e)}^{-1}$  for all $e\in E^+_\Lambda$.
Functions of $Q$ which are invariant under all gauge transformations are called gauge invariant observables.
The gauge invariant observables which satisfy the master loop equations
are the {\it Wilson loops}. 

Given a loop
$ l  = e_1 e_2 \cdots e_n$, meaning that $l$ is a closed path without any backtracking modulo cyclic permutation equivalence, 
 the Wilson loop variable $W_l^\eps$ is defined as
\begin{equ}[e:def-Wl]
W_l^\eps = \frac1N\Tr (Q_{e_1}Q_{e_2}\cdots Q_{e_n})
=\tr (Q_{e_1}Q_{e_2}\cdots Q_{e_n}) \;.
\end{equ}
The notion of Wilson loops can be generalized to a collection of loops, see Section~\ref{sec:ext1}.

\begin{remark}\label{rem:loop}
The well-definedness of a loop (especially the ``no backtracking'' condition) is slightly more subtle; also we will talk about the {\it location} of a bond in a loop which is understood in the natural way. 
Since these subtleties are of minimal importance in this paper we refer to \cite[Section~2.1]{Cha}. Also for simplicity we will not introduce any additional notation (such as $l=[e_1 e_2 \cdots e_n]$) to distinguish a closed path and its cyclic equivalence class since it will be clear from the context.
\end{remark}

\begin{remark}\label{rem:Wilson-model}
The above model \eqref{measure} is called the Wilson model. The Wilson model can be defined more generally: for an abstract Lie group $G$ (not necessarily a matrix group), together with a unitary representation of $G$, one can replace $\Tr$ in \eqref{measure} and in the definition of Wilson loops  by  the character  $\chi$ of the representation (see e.g. \cite[Def.~8.4]{Driver89}) or general class functions on $G$. Here we just focus on the concrete case where $G$ is simply a  group of $N\times N$ matrices naturally representing on $\C^N$,
so that  the character is just $\Tr$.
\end{remark}

We will investigate the scaling limit problem in $d=2$, namely the limit of the loop equations 
on $(\eps \Z)^2$ as $\eps \to 0$.
We consider the well-known scaling (e.g. \cite[Section~3]{Chatterjee18}):  
\begin{equ}[e:scale-beta]
\beta=\eps^{-2} \;. 
\end{equ}

To state the lattice master loop equation, we need some additional notation,
including the set of deformations $\mathbb{D}^\pm(l)$ and splittings $\mS^\pm (l)$
of a given loop $l$.  In words, the set $\mathbb{D}^\pm(l)$ corresponds to all possible loops obtained by adjoining a plaquette to some bond of $l$, where the bond is either removed (negative deformation) or repeated (positive deformation) in the process.  It can be expressed as a union of the smaller sets $\mathbb{D}_e^\pm(l)$ containing only the deformations along a specific bond $e$ in $l$.  The set $\mS^\pm (l)$ consists rather of \textit{pairs} of loops obtained from splitting a single loop along a repeated bond \footnote{In particular, this set is empty for simple loops, where each bond appears only once.}, and the set $\mS_e^\pm (l)$ is understood in a similar way.  We will recall the precise definitions in Section~\ref{sec:tools}. 

Fixing a bond $e$ in the loop $l$, the (single location) master loop equation for $G=U(N)$ is given by
\begin{equs}[eq:wl]
\E	W_l^\eps 
& =\frac{1}{2\eps^2} \!\!\! \sum_{l'\in \mathbb{D}^-_e(l)}\!\!\E W_{l'}^\eps
	-\frac{1}{2\eps^2} \!\!\! \sum_{l'\in \mathbb{D}^+_e(l)} \!\!\E W_{l'}^\eps
	+\!\!\!\sum_{l'\in \mS^{-}_e(l)}\!\!\E W_{l'}^\eps
	-\!\!\! \sum_{l'\in \mS^{+}_e(l)}\!\!\E W_{l'}^\eps.
\end{equs}
We note that if $e$ appears several times in $l$, the deformations $\mathbb{D}_e^\pm(l)$, splittings $\mS_e^\pm (l)$ also depend on the location of $e$,
and this location is also fixed here,  see Section~\ref{sec:notation}.

Here \eqref{eq:wl} is the version of 
the equation derived in \cite[Theorem~5.7]{CPS2023}
where we replaced $\beta$ by $\eps^{-2}$ as in \eqref{e:scale-beta}. 
It was first derived by \cite{Cha} by Stein's method (and then by \cite{SSZloop} using the Langevin dynamics and It\^o formula and by \cite{OmarRon} using the integration by parts formula) where $e$ is summed instead of fixed, which is slightly weaker.

We will focus on the case $d=2$ and consider the limit of the master loop equations in the double limit where first $\Lambda$ approaches $(\eps\Z)^2$, then $\eps \to 0$. 
Note that for $\eps$ fixed, the joint law of Wilson loops does not depend on how $\Lambda$ approaches 
 $(\eps\Z)^2$ (see e.g.  \cite[Theorems~7.2 and 7.4]{Driver89}),
and the above master loop equations hold on $(\eps\Z)^2$.
For the remainder of the paper, we will focus on the second limit and work on the entire plane and  the infinite volume limit $\mu^\eps$ instead of $\mu^\eps_\Lambda$.


In continuum, for $\beta \geq 0$ the Yang--Mills measure is formally defined by  
\begin{equ}[e:cont-YM]
\exp\Big(-\frac{\beta}{2}\sum_{i<j}\int_{\R^d} |F_A^{ij}(x)|^{2} \dif x\Big) \dif A , 
\end{equ}
where  $A$ is a $\mfg$-valued 1-form called a connection and $F_A$ is its curvature two form given by 
$F_A^{ij} = \partial_i A_j - \partial_j A_i + [A_i,A_j]$, $i,j=1,\dots,d$. 

In the setting where $A$ is smooth, the {\it Wilson loops} are defined as $\tr\,\hol (A,l)$,
namely the (normalized) trace of the holonomy of $A$ around the loop $l:[0,1]\to \R^d$.
Recall that the holonomy is defined by $\hol (A,l)\eqdef h(1)\in G$ where 
$h$ solves the parallel translation ODE $\dif h(s)=h(s)\langle A(l(s)),\dif l(s)\rangle$.
In $d=2$, it is well-known that even the Gaussian free field lives in the space of distributions $\mathcal C^\alpha$ for $\alpha<0$, and there is not any classical meaning for holonomies of {\it generic} distributions $A$. 
However, the notion of  Wilson loops was rigorously defined for 2D Yang--Mills, simultaneously by \cite{MR1015789,Driver89}, in which
the authors use an axial gauge fixing to represent the continuum Yang--Mills measure on the plane as a Gaussian measure, and define  the parallel translation of a connection $A$ along a loop $l$ by the solution to a stochastic differential equation. 
This also 
allows for rigorous computations of expectations of Wilson loops in continuum,
where the heat kernel on the Lie group plays a central role.
Driver \cite{Driver89} also gave a formula  for the joint law of a general class of  self-intersecting Wilson loops in the plane. 
In the work \cite{MR2667871} and \cite{Levy11}, L\'evy 
takes Driver's  formula as the definition of the Yang--Mills measure on a graph and then uses the consistency of this measure under subdivision to construct a continuous theory. More recently \cite[Section~3]{CCHS2d} constructed a 
space of connections (which can be embedded into $\mathcal C^\alpha$ for $\alpha<0$)
where every connection in this space has a deterministic notion of holonomies and corresponding Wilson loops,
and there is a random field $A$ (see  \cite{Chevyrev19YM,Chevyrev2023}) taking values in this space,
whose  holonomies coincide with the previous literature i.e. \cite{MR2667871}.
In this paper we will take these equivalent definitions of Wilson loops for 2D Yang--Mills theory,
and we write $W_l = \tr(\hol (A,l))$,
but in terms of calculations 
the main tool for us is Driver's formula \cite{Driver89} and
we review it in Section~\ref{sec:Dri}. 

The Wilson loops in continuum satisfy the continuum master loop equations.
These equations were proposed by physicists \cite{MM1979}  in a heuristic way, and then rigorously formulated and 
derived by \cite{Levy11}, with various proofs and simplifications in the later works  in \cite{MR3554890,Driver17,MR3982691,PPSY2023}. 
Let us first recall the formulation in \cite{Driver17}.
Suppose that $l$ is a loop in $\R^2$ with a {\it simple} crossing
 at point $v$. 
(We refer to \cite[below (1.2)]{Driver17} for the precise meaning of simple crossing,
in particular this means that $l$ passes through $v$ exactly twice, each time transversely.) 
 We parametrize $l$ on $[0,1]$ with $l(0)=l(1)=v$ and there exists a unique $s_0\in (0,1)$ with $l(s_0)=v$. We label the outgoing edges $e_1,\dots,e_4$ at $v$ in cyclic order, and  the loop $l$ starts from $e_1$ and ends through $e_3^{-1}$, 
see Figure~\ref{fig:DHKcase}  for example. (We also refer to Section \ref{sec:Dri} for precise definition of edges in graphs.) 
 We then denote the faces $F_1,\dots,F_4$ adjacent to $v$ in cyclic order such that $e_1$ lies between $F_4$ and $F_1$, $e_2$ lies between $F_1$ and $F_2$, etc. 
 We set $t_i=|F_i|$, the area of the face $F_i$. Then, for $G=U(N)$, \cite[Theorem~1.1 eq.(1.3)]{Driver17} states that 
\begin{equ}\label{eq:wlc}
	\big(\p_{t_1}-\p_{t_2}+\p_{t_3}-\p_{t_4}\big)
	\E W_l 
	=\E \big( W_{l_1} W_{l_2}\big)
\end{equ}
where $l_1$ and $l_2$ are the restrictions of $l$ to $[0,s_0]$ and $[s_0,1]$, respectively.  Note that the loop $l$ starts from $v$ at time $0$ and proceeds across $e_{1}$, then visits various unspecified edges until passing through $e_{4}^{-1}$ and reaching $v$ again at time $s_{0}$ (forming $l_1$), then proceeds from $v$ across $e_{2}$, following some additional unspecifed edges until crossing $e_{3}^{-1}$ and arriving back at $v$ at time $1$ (forming $l_2$), see \eqref{e:genl} for the precise formulation. Interestingly, equation \eqref{eq:wlc} only depends on the local topological behavior around the vertex $v$ and the area derivatives, not on the whole structure of the loop.

\begin{figure}[h] 
  \centering
\begin{tikzpicture}[scale=1.5]
\filldraw [black] (0,0) circle (1pt); \node at (0, -0.2) {$v$};
\draw[thick,->] (0,0) -- (1,1) node[midway,right] {$e_1$};
\draw[thick,->] (0,0) -- (1,-1) node[midway,right] {$e_4$};
\draw[thick,->] (0,0) -- (-1,1) node[midway,left] {$e_2$};
\draw[thick,->] (0,0) -- (-1,-1) node[midway,left] {$e_3$};
\node at (0, 0.7) {$F_1$};\node at (0, -0.7) {$F_3$};
\node at (0.7,0) {$F_4$};\node at (-0.7,0) {$F_2$};
\end{tikzpicture}
\qquad\qquad
\includegraphics[scale=0.6]{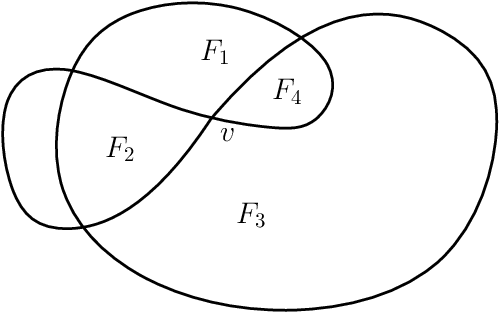}
  \caption{An example of a  loop  and a neighborhood around a vertex $v$.}  
  \label{fig:DHKcase}
\end{figure}
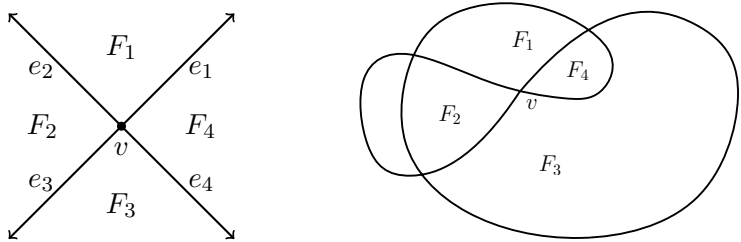

Both the lattice and the continuum loop equations \eqref{eq:wl} 
and \eqref{eq:wlc} carry significant conceptual and practical information about the
 lattice or continuum Yang--Mills measure.  Most existing applications of \eqref{eq:wlc} are in the limit $N \to \infty$, known as large $N$ problems, where the RHS of the equation factorizes into a product of expectations and the utility of \eqref{eq:wlc} becomes particularly apparent.  Rigorous work in the continuum begins with \cite{Levy11} on the plane and remains an ongoing area of investigation for compact surfaces, the sphere case being studied in \cite{hall2018large,dahlqvist2020yang} and more recently, progress on general surfaces in \cite{dahlqvist2023large, Dah2022II}.  Similarly, the master loop equation in the discrete setting is used to study large $N$ problems on the lattice in \cite{Cha,MR3861073,borga2024surface}. They are also closely related with surface summations (c.f.  \cite{Cha,ChatterjeeJafar,CPS2023}). We also mention that Dyson--Schwinger equations are useful for large $N$ problems of other models (see e.g. \cite{SZZ23}). 
  
{\it Heuristically}, the lattice loop equations \eqref{eq:wl} 
and the continuum  loop equations \eqref{eq:wlc}  have some obvious similarities.
For instance $\{l_1, l_2\}$ in  \eqref{eq:wlc} is a splitting of $l$, and the area derivatives 
in  \eqref{eq:wlc} are reminiscent of deformations 
 in a certain infinitesimal sense. 
 Some proofs of \eqref{eq:wl} and \eqref{eq:wlc} also share similar philosophies. For instance, \cite{Driver17} in continuum and 
 \cite{Cha} on the lattice both use integration by parts;
 \cite{MR3982691}  in continuum and \cite{SSZloop} on lattice both apply 
It\^o formula, to the SDE of parallel translations in the former paper, and  to the  SDE of Langevin dynamic in the latter paper.
In fact, the early physics paper  \cite{MM1979} also describes their heuristics in a mix of discrete and continuum settings.

However, despite the conceptual  similarities, 
the  loop equations \eqref{eq:wl} and \eqref{eq:wlc} have far-apart mathematical forms (see  remarks along this line in the end of \cite[Section~1]{Driver17} and  \cite[Section~1]{MR3982691}).
In particular, the lattice loop equation \eqref{eq:wl} is parametrized by the choice of a local bond $e$, whereas the continuum loop equation only involves the variation of a ``global'' quantity which is the area.
Moreover, we  emphasize that while a negative deformation appearing in the lattice loop equation is very intuitively related to ``discrete area derivatives'', it is not obvious for  positive deformations  (see \eqref{eq:def} below). 
The natural questions we are interested in are thus (see Fig.~\ref{fig:DHKcase}):
\begin{enumerate}
\item Suppose we have a sequence of lattice loops approximating $l$ in a proper sense, then, what happens in the limit as $\eps\searrow 0$, if we consider the lattice loop equation \eqref{eq:wl} where we deform at a bond $e$ inside the (lattice approximations of) edges $e_1,\cdots,e_4$?
\item Suppose in the above  approximations, $v$ is replaced by a lattice bond $e$, then what happens in the limit as $\eps\searrow 0$ for the lattice loop equation \eqref{eq:wl} with respect to this bond $e$?
\item How, or in what sense, do the loop equations \eqref{eq:wl} converge to \eqref{eq:wlc}? Which bond(s) does one need to choose (or how to combine them) in \eqref{eq:wl} for this convergence to occur? What happens for each term in \eqref{eq:wl} and what kind of cancellations occur during this limiting procedure?
\end{enumerate}
In this paper we address these questions.
We will also include the cases $G\in \{SU(N), SO(N)\}$, and extend the convergence result to a string of loops. 
The proof of this convergence is more subtle than one might first expect.
For instance, we will need to find certain cancellations  in  \eqref{eq:wl}; and,
 in the discrete equation \eqref{eq:wl} one needs to pick a bond $e$,
and we will actually  combine  \eqref{eq:wl} with several choices of 
the bond $e$ in a suitable way in order to see the limiting equation
 \eqref{eq:wlc}. 
 
Finally, although this was not the main motivation for this paper, we note that our convergence results  in particular provide {\it one more} proof of the master loop equation in continuum, in addition to the earlier proofs \cite{Levy11,MR3554890,Driver17,MR3982691,PPSY2023}. 
 
\subsection{Main results}  

We state the main result of this paper, with more precise formulation postponed in the later sections.

\bt\label{th:main} 
Suppose $d=2$ and $G=U(N)$. Let $l$ be a loop in $\R^2$ with a simple crossing at the point $v$ as described above \eqref{eq:wlc}. 

Then for a general class of 
lattice approximations $\{l^\eps\}_\eps$ of $l$,  
as well as a general rule of selecting bonds in $l^\eps$, 
 suitable linear combinations
of the discrete master loop equations \eqref{eq:wl} for $l^\eps$ with these selected bonds converge
 to the continuum master loop equation \eqref{eq:wlc} as $\eps\to0$. 
\et

The more precise version of Theorem~\ref{th:main} is stated in Theorem~\ref{th:g1} and Corollary~\ref{cor:lin-comb}.
In particular, 
the lattice approximation $\{l^\eps\}_\eps$ of $l$ is defined in Section \ref{sec:Dri} and Section \ref{sec:gen}.   
Regarding the selection of bonds in the above statement, 
we will consider a rule described as in Definition~\ref{def:rule} to prove our convergence result, and then extend it to more general ways of selecting the bonds as described in Corollary~\ref{cor:lin-comb}.
The phrase ``suitable linear combinations'' is an important point in the theorem.
We note that a short soft argument for Theorem \ref{th:main}  can be given if one appeals to the corresponding continuum result \eqref{eq:wlc}.
We prefer to avoid using the continuum result, as this forces us to carry out a more refined analysis of the deformation terms which we believe is insightful in its own right. For a more technical discussion on this point, we refer the reader to the last paragraph at the end of Section \ref{sec:gen}.

\br \label{rem:conv-eq}
Throughout the paper,  when we say that  ``the equation $(A_\eps)$ converges  to equation $(A)$ as $\eps\to 0$'', we  mean that as we take $\eps\to0$ on both sides of the equation $(A_\eps)$, the resulting limit is equation $(A)$. Here we allow moving terms to either side of the limiting equation. 
\er 

We  extend  Theorem \ref{th:main} to a collection of loops in Section~\ref{sec:ext}.
In this case, we also obtain the splitting term as in the RHS of \eqref{eq:wlc}, which arises from 
splitting terms $\mS(l)$ in the lattice loop equations. 
See \eqref{ms:s}. 
  If we consider a collection of loops $l_1,\dots,l_n$,
where $l_1, l_2$ intersect at a point $v$,  the merger term may also appear. 
   More precisely, we can also find a suitable lattice approximation $\{l_1^\eps,\dots,l_n^\eps\}$  such that suitable linear combinations of the discrete master loop equation \eqref{eq:wl} converge to
\begin{align*}
	\Big(\frac{\p}{\p t_1}-\frac{\p}{\p t_2}+\frac{\p}{\p t_3}-\frac{\p}{\p t_4}\Big)
	\E \Big(\prod_{i=1}^n W_{l_i}\Big)
	=
	\frac1{N^2} 
	\E \Big( W_{l_{12}}  \prod_{i=3}^n W_{l_i}\Big)
\end{align*}
 where the new loop $l_{12}$  is the merger of $l_1$ and $l_2$, see Theorem \ref{th:g2} below for the precise statement.

We can also extend the result to $G=SU(N)$ and $G=SO(N)$. For the case of $SO(N)$, a new action called twisting will appear (see Section \ref{sec:e2} for more details).

\medskip

Let us preview some main ingredients of our proof.
The key step in the proof is to choose appropriate bonds $e$ and identify the limit of the following terms arising from deformations:
\begin{align}\label{dis:def}
	\frac1{2\eps^2}\sum_{l'\in \mathbb{D}^-_e(l)}\!\!\E W_{l'}^\eps
	-\frac{1}{2\eps^2} \!\!\! \sum_{l'\in \mathbb{D}^+_e(l)} \!\!\E W_{l'}^\eps
\end{align}
from \eqref{eq:wl}, interpreting them as the appropriate area derivative of the expectation of the Wilson loop. To achieve this, we express the deformation term \eqref{dis:def} in terms of integrals over the Lie group $G$ with respect to the ``Wilson action'' (see \eqref{ac:wi} below) using Driver's formula (see \eqref{dri:dis} below). We then convert these integrals on $G$ into corresponding Gaussian integrals on the Lie algebra $\mfg$ by dividing the integral into `small field' region and the `large field' region -- see Section~\ref{sec:Gaussian}.
 This allows us to identify the limit of \eqref{dis:def} as the appropriate gradient of the heat kernel, which yields the area derivatives by applying the integration by parts formula on the Lie group. 
Another important tool we exploit is the Peter--Weyl spectral theory for $L^2$ class functions on Lie groups,
 allowing us to 
decompose the heat kernel  and Wilson action  in terms of the irreducible characters, which we review in Section~\ref{sec:Lie}. 
Finally, the suitable choice of bonds as mentioned in Theorem~\ref{th:main} is important to derive \eqref{eq:wlc}: it is natural to approximate the crossing vertex $v$ by  a small edge $e^\eps$ and apply \eqref{eq:wl} for a bond in $e^\eps$ to obtain a splitting term;
however, in doing so, certain correction terms arise from integration by parts.
To cancel these, we further select  bonds from edges adjacent to $e^\eps$, yielding the continuum loop equation \eqref{eq:wlc}.
In all these calculations, fixing axial gauge in a suitable way turns out to be helpful.

\medskip

We conclude by mentioning some questions for possible future studies. 
Based on ideas in \cite{Driver17}, the paper 
\cite{MR3631396} extended the proof of master loop equations from the plane to compact surfaces, and it would be interesting to derive a version of discrete  loop equations on surfaces and show their continuum limits.  We also remark that the lattice  loop equations hold in {\it all} dimensions, and it would not be hard to derive them for lattice Yang--Mills--Higgs models (the  Langevin dynamics were derived in \cite{SZZ2024Higgs}) as well, but the continuum analogues  of such loop equations in $d\ge 3$ or with Higgs are far from being understood.  Our present paper relies on Driver's formula in $d=2$, which is unavailable for these problems, and a natural question is whether our results can be proved in a way that relies less heavily on the exact integrability properties of the model.  This would be interesting already for the case of simple loops as in Theorem \ref{th:1}, as this connects naturally to understanding the relationship between the master loop equations and Wilson's area law, particularly in the continuum limit.  
Finally, due to the close relation between the loop equations and surfaces summations  (c.f.  \cite{Cha,ChatterjeeJafar,CPS2023})
it would be great if our  analysis of the limiting behavior  of each term in the lattice loop equation (and certain cancelations among these terms) would shed some light on the continuum limit problems of random surfaces (see such open problems in \cite[Section~7]{CPS2023}, for instance  Problems 5, 7, 10, 13) at least in 2D. 

\medskip


{\bf Structure of the paper.}
This paper is organized as follows. In Section \ref{sec:tools} we review notations and basic tools from literature. In Section \ref{sec:sim} we explain how the discrete loop equations converge to their continuum counterparts when the loop is simple.
Section \ref{sec:com} studies general loops
 as stated in Theorem \ref{th:main}, whose proof is given in  Section \ref{sec:gen}. 
Sections \ref{sec:Gaussian} and \ref{sec:4.1} provide 
some general approximation results on integrals against Wilson action over Lie groups.
 In Section \ref{sec:ext1} we extend the result to loop sequences,
and  in Section \ref{sec:e2} we  also extend to the cases of $G=SO(N)$ and $G=SU(N)$.

\medskip


{\bf Acknowledgments.}
H.S. gratefully acknowledges financial support from an NSF grant (CAREER DMS-2044415), and a Simons Fellowship from the Simons Foundation.  S.S. and R.Z. are grateful to the financial supports by National Key R\&D Program of China (No. 2022YFA1006300).
R.Z. is grateful to the financial supports of the NSFC (No. 12426205, 12271030), and BIT Science and Technology Innovation Program Project 2022CX01001 and the financial supports  by the Deutsche Forschungsgemeinschaft (DFG, German Research Foundation) – Project-ID 317210226--SFB 1283.

\section{Notation and some basic tools}
\label{sec:tools}

\subsection{Notation}
\label{sec:notation}

{\it Wilson Action and its $k$-fold convolution.} Define the action \footnote{Although we used the notation $Q$ to denote a configuration in the definition of \eqref{measure}, here and occasionally in other places of the text we will also use $Q$ to denote a single element of the Lie group.  The distinction should be clear from the context.} 
\begin{equation}\label{ac:wi}
	G\ni Q 
	 \mapsto S^{\eps}(Q)=\frac{1}{Z^{\eps}}e^{-\eps^{-2}N\text{Re Tr}(I-Q)}, \quad Z^{\eps}=\int e^{-\eps^{-2}N\text{Re Tr}(I-Q)}\dif Q
\end{equation}
which can be viewed as the transition probability from the identity $I$ to $Q$ of a random walk on $G$.
Following \cite[Section 8]{Driver89}  we will call $S^\eps$  the ``Wilson action'' (although in some other references ``Wilson action'' just means the exponent in $S^\eps$).
We recall that for two functions $f,g : G \mapsto \R$, their convolution is defined by 
\begin{equation}
f*g(a)=\int f(b)g(b^{-1}a)\dif b \nonumber, \qquad a\in G,\label{eq:defconv}
\end{equation} 
whenever the above integral is finite.
We write $S^\eps_k$ for its $k$-fold convolution of the action with itself.  

Recall that a function $f$ on $G$ is called a {\it class function} if 
it only depends on the conjugation class, namely
$f(aba^{-1})=f(b)$ for all $a,b\in G$.
We will often use the fact that 
$\Tr$ and $S^\eps$ are class functions on $G$,
in particular $S^\eps (ab) =S^\eps (ba)$.
We will often use $S^\eps (a)=S^\eps (a^{-1})$.

{\it Operations on loops in the lattice.} In the master loop equations, besides single loops we will often 
need a {\it string}  $s=(l_{1},\dots,l_{m})$, which means a collection of loops (also called a loop sequence in \cite{Cha}). 
Now we define the loop operations, including splitting, merger, and deformation. 
In these definitions we will write $a,b,c,d$ for paths and $e$ for a bond. We also recall  that (see Remark~\ref{rem:loop}) the loops are defined modulo cyclic permutations, e.g. $l=aebec$ and $l=becae$ are the same loop.
We refer to \cite[Fig.~4 -- Fig.~11]{Cha} or  \cite[Fig.~2 -- Fig.~9]{Jafar} for pictures
of these operations, and we will also provide more pictures along the proofs in Sections~\ref{sec:sim}  and \ref{sec:com} 
(c.f. Fig.~\ref{fig:deform-e}, \ref{fig:deform-ee1}).

{\it Splitting.} Given a loop $l$ of the form $l=aebec$ (where $e$ is a bond appearing twice at locations $x$ and $y$), the {\it positive splitting} of $l$ is a pair of loops 
\begin{equ}[e:pos-split]
 l_1 \eqdef aec \;,\qquad  l_2 \eqdef be\;.
\end{equ}
For $l=aebe^{-1}c$ (where $e$ and $e^{-1}$ appear at locations $x$ and $y$ respectively), the {\it  negative splitting} of $l$ is a pair of loops
\[
l_1 \eqdef ac \;,\qquad  l_2 \eqdef b\;.
\]
We say that a string $s'$ is obtained from splitting $s$ if exactly two loops in $s'$ arise from splitting a single loop in $s$.  
We write  $\mS^{+}_e((x,y);s)$ (resp. $\mS^{-}_e((x,y);s)$) for the set of strings
obtained from positive (resp. negative) splitting of $s$ with respect to the bond $e$ at locations $(x,y)$. In fact, since the loops are defined modulo cyclic equivalence by Remark~\ref{rem:loop}, once we fix $e$ and the locations $x,y$,  then each of
$\mS^{+}_e((x,y);s)$  and $\mS^{-}_e((x,y);s)$ only has one possible string $s'$.

{\it Merger.}
For two loops $l=aeb$ and $l'=ced$,
where $e$ appears at location $x$ in $l$ and at location $y$ in $l'$,
the {\it positive and negative mergers} of $l$ and $l'$
at locations $x,y$ are the loops
\begin{equ}[e:merger1]
l\oplus_{x,y} l' = aedceb \;,\qquad l\ominus_{x,y} l'=ac^{-1}d^{-1}b\;.
\end{equ}
For $l=aeb$ and $l'=ce^{-1}d$, 
where $e$ appears at location $x$ in $l$ and $e^{-1}$ at location $y$ in $l'$,
the {\it positive and negative mergers} of $l$ and $l'$
at locations $x,y$ are the loops
\begin{equ}[e:merger2]
l\oplus_{x,y} l' = aec^{-1}d^{-1}eb\;,
\qquad
l\ominus_{x,y} l'=adcb\;.
\end{equ}
 We say that a string $s'$ is obtained from merging $s$ if exactly one component of $s'$ arises from merging two loops in $s$.  
 The sets $\mM^{+}_e((x,y);s)$ and $\mM^{-}_e((x,y);s)$ denote all the strings obtained from either positive mergers or negative mergers of $s$ with respect to $e$ at locations $(x,y)$.  
 
 Furthermore, we define two more sets $\mM^{+}_{U,e}((x,y);s) \subset \mM^{+}((x,y);s)$ and $\mM^{-}_{U,e}((x,y);s) \subset \mM^{-}((x,y);s)$; the first consists of positive mergers resulting from \eqref{e:merger1}, namely a bond $e$ appearing in both of the two merged loops; the second consists of negative mergers 
resulting from \eqref{e:merger2}, i.e. where a bond $e$ occurs in one loop and $e^{-1}$ in the other. 

The notion of merger will be only used in Section~\ref{sec:ext} in its general form,
and in this section we only need it to define deformations as follows.

{\it Deformation.} 
For a loop $l$ where the bond $e$ occurs at location $x$,  and a plaquette $p$ where  
$e$ or $e^{-1}$   occurs at location $y$ (in this case $y$ is the unique such location), 
we write 
\[
l \oplus_x p \qquad \mbox{and} \qquad l\ominus_x p
\]
for the positive or negative deformations which means the positive or negative mergers 
of $l$ and $p$ at locations $x$ and $y$. \footnote{Spelled out more explicitly, for $l=aeb$ where the edge $e$ starts at location $x$ and a plaquette $p=ec$, positive and negative deformation at $x$ map $l$ to $apeb$ and $ac^{-1}b$ respectively; adjoining a plaquette while either repeating or removing the edge $e$.}

	 We say that a string $s'$ is obtained from deformations of $s$ if exactly one component of $s'$ arises from deformation of one loop in $s$.  
	The sets $\mD^{+}_e(x;s)$ and $\mD^{-}_e(x;s)$ consist of all strings obtained from positive or negative deformations of $s$ with respect to $e$ at location $x$, respectively.

\medskip

We emphasize that these sets 
depend on the bond $e$ as well as its {\it locations},
so they are  the same as \cite{CPS2023}
but smaller than  \cite{Cha,Jafar}. 
For the rest of the paper, for simplicity of notation we will just write 
\begin{equ}[e:ignore-loc]
\mS^{\pm}_e(s) = \mS^{\pm}_e((x,y);s)\;,
\quad
\mM^{\pm}_e(s) = \mM^{\pm}_e((x,y);s)\;,
\quad
\mD^{\pm}_e(s)  = \mD^{\pm}_e(x;s)\;,
\end{equ}
implicitly keeping in mind that we fix the locations. Also when $s$ is a single loop $l$,
we will write 
$\mS^{\pm}_e(l)$, 
$\mD^{\pm}_e(l)$.

\subsection{Analysis on Lie groups}
\label{sec:Lie}

Let $G$ be a compact Lie group with a given bi-invariant metric.
Let $\mfg$ be its Lie algebra, and $d(\mfg)$ be its dimension. 
Let $\{L_j\}_{j=1}^{d(\mfg)}$ be an orthonormal basis for $\mathfrak{g}$.
We write the derivative of a function $f$ on $G$ in the direction $X\in \mfg$ as 
\begin{equ}[e:CL]
\mathcal L_X f(a) 
\eqdef \frac{\dif}{\dif t} \Big|_{t=0} f(e^{tX} a) 
=\<\nabla_{a}f (a), X a\>\;,
\qquad
\mathcal L_j \eqdef \mathcal L_{L_j}
\end{equ}
for $a\in G$, where $Xa$ denotes the translation of $X$ to 
the tangent space of $G$ at $a$ 
via right multiplication,
and $\nabla_a f$ denotes the gradient of $f$ (and we write a subscript $a$
to emphasize that  it is as a function of $a$).
		
 Recall that a finite dimensional representation $\tau$ is a homomorphism $g \in G \mapsto \tau(g) \in GL(V_{\tau})$, where $V_{\tau}$ is a (complex) vector space  of dimension $d_\tau=\dim V_\tau$, and $GL(V_{\tau})$ is the general linear group of automorphisms of $V_{\tau}$.
 For any finite dimensional representation $\tau$ of a Lie group $G$, the corresponding character is defined by  $\chi_\tau(g)=\Tr(\tau(g) )$.  Although the mapping $g \mapsto \chi_{\tau}(g)$ is generally non-linear, it linearizes on the Lie algebra $\mathfrak{g}$ as follows.  We write the representation on the Lie algebra as $\tau(A):=\frac{\dif}{\dif t}|_{t=0}\tau(e^{t A})$ for $A\in \mfg$, which is linear w.r.t. $A$. 
Recall that there always exists a suitable inner product on $V_\tau$ such that $\tau$ is equivalent with a unitary representation, i.e. $\tau(g^{-1})=\tau(g)^{-1}$, and in this case $\tau(A)$ is skew-Hermitian for all $A\in\mfg$. For the rest of the paper, we will just assume that $\tau$ is unitary. (This is sufficient for our purpose of applying spectral theory on Lie groups, c.f.  \cite[Section~3.1]{Folland}.) 
We also recall that
the Casimir operator $C_\tau$ on $V_\tau$ is given by 
\begin{align}\label{e:Ctau}
	C_\tau=\sum_{j=1}^{d(\mathfrak{g})}(\tau(L_j))^2 \;.
\end{align}
If $\tau$ is irreducible then by Schur's lemma $C_\tau=c_\tau I$ 
where $c_\tau\le 0$ is called the Casimir constant and $I$ is the identity matrix on $V_{\tau}$, so that $\text{Tr}C_{\tau}=c_{\tau}d_{\tau}$.
We have the probabilistic characterization of the Casimir:
with a standard Gaussian measure proportional to $e^{-\frac12 |A|^2}$ on $\mfg$,
one has 
\begin{equ}[e:Cas]
\E [(\tau (A))^2] = C_\tau\;.
\end{equ}

We will always assume that the Haar measure on $G$ is normalized
so that $G$ has unit volume. Recall that the characters of all the 
irreducible representations form an orthonormal basis 
for the Hilbert space of  square integrable class functions. 
We use $p_t:G\to \R^+,t\in\R^+$,  to denote the heat kernel of $\frac12\Delta$ with Laplace-Beltrami operator $\Delta$ relative to the metric \eqref{def:inn},  i.e. the unique positive solution of the heat equation $(\p_t-\frac12\Delta)p=0$ with initial condition $p(t,Q)\dif Q\to \delta_{I}$ as $t\to0$.
It is important for us to have the following spectral decomposition
\begin{equ}[e:spec-S]
S^\eps =\sum_\tau d_\tau a_\tau(\eps) \chi_\tau\;,
\qquad
p_t = \sum_\tau d_\tau e^{\frac12 t c_\tau} \chi_\tau\;,
\end{equ}
where $\tau$ is over all the irreducible representations of $G$, $d_\tau$ is the dimension of $\tau$, and $a_\tau(\eps)\in \R^+$. By \cite[(8.2), (8.3)]{Driver89} (or \cite[Appendix~A]{BS83}) 
\begin{equ}[e:approx-a]
\Big| a_\tau(\eps) - e^{\frac1{2} c_\tau \eps^2} \Big| \lesssim c_\tau^2 \eps^4\;.
\end{equ}
Here and in the sequel, we use the notation $a\lesssim b$ if there exists a constant $c>0$ such that $a\leq cb$.
The $k$-fold convolution is 
simply given by $S^\eps_k =\sum_\tau d_\tau a_\tau(\eps)^k \chi_\tau$. 

By \cite[Theorem 8.8]{Driver89} (or \cite[Appendix~A]{BS83}), we have that for $t>0$ and $t(\eps)$ satisfying $t(\eps)/\eps^2\in\Z^+$, $|t(\eps)-t|\lesssim \eps$ 
\begin{equation}\label{eq:a-cons}
S^\eps_{t(\eps)/\eps^2}
\to p_{t} \text{ uniformly },\qquad  \mbox{as }\eps\to0.
\end{equation}

\begin{remark}\label{rem}
We will often work with the {\it standard representation} 
\footnote{It is also called ``fundamental representation'' in e.g. \cite{BS83}.}
(following the terminology of \cite[Section~4.2]{Hall15})
$\tau$ of $G=U(N)$ on $\C^N$
where $\tau:U(N)\to U(N)$ is just the identity map.
For instance this is how we set up our model \eqref{rem:Wilson-model}, see also Remark~\ref{rem:Wilson-model}.
In this case we have $d_\tau=N$ and $c_\tau = -1$ 
\footnote{See also \cite[(2.4)]{SSZloop} where $c_\tau = -N$ shows up. 
Note that our convention for the  inner product on $\mfg$ differs from 
\cite{SSZloop} by a factor $N$ and thus    $c_\tau = -N$ therein.}.
\end{remark}

%
%

\subsection{Some earlier results of Driver}\label{sec:Dri}


We first recall {\it Driver's formula} for Wilson loop expectations on the plane
 in the continuum setting.  
 To this end, we define an {\it edge} to be a continuous map $e:[0,1]\to\R^2$, which is assumed to be injective except possibly that $e(0)=e(1)$. \footnote{Allowing for $e(0)=e(1)$ can be useful, for instance a simple loop can be considered as only a single edge.} 
 The inverse of $e$, denoted by $e^{-1}$, 
 is the edge traced backwards: $e^{-1}(s)=e(1-s)$.
 A {\it graph}  is a finite set of edges (and their inverses), that meet only at their endpoints. The vertices of a graph are the endpoints of its edges. The faces of a graph are the connected components of the complement in $\R^2$ of the union of its edges. A graph is then described as a triple $\mG=(\mV,\mE,\mF)$ consisting of a set $\mV$ of vertices, a set $\mE$ of edges and a set $\mF$ of faces. Note that $\mG$ is determined by the set of edges $\mE$. We choose an {\it orientation} of $\mG$, i.e. a subset $\mE^+\subset \mE$ containing exactly one element in each pair $\{e,e^{-1}\}$. We then refer to the edges in $\mE^+$ as the positively oriented edges.

 We associate to each positively oriented  edge $e\in \mE$ an {\it edge variable} $Q_e\in G$, and correspondingly associate $Q_e^{-1}$ to $e^{-1}$. We view the edge variable as the {\it parallel transport} of a connection along the edge, which is the solution to a stochastic differential equation in the sense of \cite{MR1015789,Driver89} as mentioned below \eqref{e:cont-YM}.

A discrete gauge transformation is a map $g:\mV\to G$.  It induces a transformation $\Psi_g$ of the edge variables given by 
$$\Psi_g(Q_e)=g(v_2)^{-1}Q_eg(v_1),$$
where $v_1=e(0), v_2=e(1)$. For a function $f$ of the edge variables, we say that $f$ is {\emph{gauge invariant}} if $f\circ \Psi_g=f$ for every discrete gauge transformation $g$. 

 {\it Driver's formula} in \cite[Theorem 6.4]{Driver89} then says that for $\mu$ being the Yang--Mills measure on the plane, the expectation of any gauge-invariant function $f$ of the parallel transport along the edges of $\mG$ (denoted as $\E[f(P|_\mE)]$ therein) can be computed via the following integral: 
\begin{align}\label{eq:Driver-c}
	\mu(f)=\int_{G^{\mE^+}} f(Q)  \prod_{F\in \mF}p_{|F|}(h_{\p F}(Q))\dif Q,
\end{align}
where $|F|$ is the area of the face $F$, 
$\p F$ denotes a path that goes once around the face $F$ in the positive direction,
and $h_{\p F}$ denotes the holonomy around $F$, that is, the product of edge variables  going around $\p F$, i.e. for $\p F=e_1\dots e_k$, $h_{\p F}=Q_{e_1}\dots Q_{e_k}$.  Here $\dif Q$ denotes the product of normalized Haar measures in all the edge variables.  
We also refer to \cite{Sengupta97,Levy03} for discussion on this integral formula.

Given a loop  $l$ in $\R^2$, i.e. a path, which is a continuous map  $l:[0,1]\to \R^2$, with $l(0)=l(1)$,  
assuming that all the self-crossing points of $l$ are simple, then 
it naturally defines a graph $\mG$.
Note that since in a graph the edges must meet only at their endpoints,
for each simple self-crossing $v$ of $l$, 
and each edge $e$ of the graph $\mG$ induced by $l$,
either $v$ does not belong to $e$, or $v$ is an endpoint of $e$.
The orientations of these edges of $\mG$ are given by the  orientation of $l$.
We define  the Wilson loop $W_l\eqdef \tr Q_l = \frac1N \Tr Q_l$  with $Q_l$ being the holonomy along $l$. It is easy to check that $W_l$ is a gauge invariant function, so \eqref{eq:Driver-c} gives a useful formula to calculate the expectations of the Wilson loops.  

\begin{remark}\label{rem:finer}
Obviously, we can break an edge $e$ in a graph into several concatenating edges.
This will not change the Wilson loop expectations, because holonomies are multiplicative along these smaller edges and the faces remain the same.
\end{remark}

Similarly we  also introduce graphs $\mG^\eps=(\mV^\eps,\mE^\eps,\mF^\eps)$
on $(\eps\Z)^2$.
More precisely, we write $B(\eps)$ for the infinite graph in $\R^2$ 
whose edges are the bonds of $(\eps \Z)^2$ (we view these bonds as line segments in $\R^2$). $B(\eps)$ has the topology induced from $\R^2$.
An edge  \footnote{Note the change in terminology in comparison to \cite{SSZloop}. A single edge in the present paper corresponds to a sequence of edges in \cite{SSZloop}, which is why we use the terminology bond to distinguish the two.} $e\in \mE^\eps$ is a continuous map $: [0,1]\to B(\eps)$,
which is assumed to be injective except possibly that $e(0)=e(1)$.
By definition, we can write  such an edge as $e=e_1\dots e_n$ where $e_1,\dots ,e_n$ are bonds of $(\eps \Z)^2$.
Similarly as in the continuum setting, the vertices of a graph are the endpoints of its edges, and the faces of a graph are the connected components of the complement in $\R^2$ of the union of its edges.  As in the continuum, we assume the graph consists of edges which only meet at their endpoints.

For each $e\in \mE^\eps$, we also introduce an edge variable $Q_e:\mE^\eps\to G$, and correspondingly associate $Q_e^{-1}$ to the inverse of $e$. More precisely,
recalling the field 
$Q = (Q_e)_{e\in E_\Lambda^+}$  defined around
\eqref{measure},
 for $e=e_1\dots e_n$ 
 where $e_1,\dots ,e_n$ are bonds of $(\eps \Z)^2$,
we set  $Q_e\eqdef Q_{e_1}\dots Q_{e_n}$. 
Gauge invariant functions can be introduced in exactly the same way as the continuous setting above, and it is clearly consistent with the notion of gauge invariance mentioned below \eqref{measure}.
 In particular, a Wilson loop $W_l^\eps$ defined in \eqref{e:def-Wl} is a gauge invariant function.
Here, the lattice loop $l$ can be again viewed as a lattice graph $\mG^\eps$. More precisely,
recall that $l$ does not have backtracking, but it may pass through some bonds of $(\eps \Z)^2$ more than once; in other words if we view $l\subset B(\eps)\subset \R^2$ as a continuum loop it will have non-simple self-crossings. 
 Whenever  $l$ has the form $l=aebe$ or $l=aebe^{-1}$ where $a,b,e$ are sequences of bonds, we assume that the starting and the ending vertices of $e$ belong to $\mV^\eps$,
and $e,e^{-1}$ belong to $\mE^\eps$.
 
We will frequently use the following formula (\cite[Theorem~7.5]{Driver89}) for  a gauge invariant function $f$ on edge variables
which we will refer to as the {\it discrete Driver's formula}:
\begin{align}\label{dri:dis}
\E f
=\mu^\eps(f)
\eqdef
\int_{G^{(\mE^\eps)^+}} f(Q) 
\prod_{F^\eps\in \mF^\eps} 
S^\eps_{\frac{| F^\eps|}{\eps^2}} (h_{\p F^\eps}(Q))\dif Q,
\end{align}
where we follow the same notation as in \eqref{eq:Driver-c},
with $|F^\eps|\in \eps^2\Z^+$ being the area of $F^\eps$ as a set in $\R^2$,
(so $| F^\eps| / \eps^2$ is the number of  plaquettes enclosed by $F^\eps$).

\medskip

Furthermore, it will be convenient for us to fix an {\it axial gauge} in the calculations of the expectation of Wilson loop. More precisely,  for any gauge invariant function $f$ we can choose  a tree $T$ (i.e. a subgraph which does not include any closed path), such that 
\begin{align}\label{dri:dis1}
	\E f
	=\mu^\eps(f)=
	\int_{G^{(\mE^\eps)^+}} f(Q) \prod_{F^\eps\in \mF^\eps} S^\eps_{\frac{| F^\eps|}{\eps^2}} (h_{\p F^\eps}(Q))\dif_T Q,
\end{align}
where $\dif_T Q=\prod_{e\in T} \delta_{\{Q_e=I\}}\prod_{e\notin T}\dif Q_e$. 

\medskip

Suppose that $\mG=(\mV,\mE,\mF)$ is a graph on $\R^2$. 
As in \cite[Definition~8.1]{Driver89}, a lattice approximation to $\mG$ is a collection of graphs 
$\{\mG^\eps=(\mV^\eps,\mE^\eps,\mF^\eps)\}_{\eps_0>\eps>0}$ on $(\eps\Z)^2$ with 
maps
\begin{equ}[e:ij]
i_\eps:\;
\mE\to \mE^\eps\quad  e \mapsto e^\eps\;, 
\qquad \mbox{and} \qquad
j_\eps:\,
\mF\to \mF^\eps \quad F\mapsto F^\eps
\end{equ}
satisfying the following conditions:
\begin{itemize}
	\item $j_\eps$ is a bijection, and $i_\eps$ is an injection. Moreover, the edges in $\mE^\eps\backslash  i_\eps(\mE)$ do not meet each other and all have lengths of $O(\eps)$. 
	\item The area $|F\backslash F^\eps|+|F^\eps\backslash F|$ is of  $O(\eps)$. 
	\item 
	For every $F\in \mF$, one has $i_\eps(\p F)=\p F^\eps \cap i_\eps(\mE)$, 
	where $F^\eps=j_\eps(F)$. 
\end{itemize}
When $i_\eps(e)=e^\eps$ we will sometimes say that the edge $e^\eps$ approximates $e$.

\br\label{re:app}
Note that we do not require $i_\eps$ to be surjective as in \cite[Definition~8.1]{Driver89}. 
The reason that we may have more edges in $\mE^\eps$ (rather than requiring a one-to-one correspondence between $\mE$ and $\mE^\eps$)
is that 
in Section \ref{sec:gen}  we will approximate a crossing vertex $v\in \R^2$
 by a ``tiny'' edge $e^\eps$ with $|e^\eps|=O(\eps)$,
 and exploit the splitting that occurs at (a bond of) $e^\eps$. 
 \er
 
It is easy to find  
that given a graph $\mG=(\mV,\mE,\mF)$, there exists a lattice approximation $\{\mG^\eps=(\mV^\eps,\mE^\eps,\mF^\eps)\}_{\eps_0>\eps>0}$ (c.f. \cite[Lemma 8.2]{Driver89}).  
\begin{lemma}\label{lem:Driver}
For any gauge invariant function $f$ of the edge variables on the graph $\mG$, one has
\begin{equ}[e:lattice-approx]
	\lim_{\eps\to0}\mu^\eps(f\circ i_\eps^{-1})=\mu(f),
\end{equ}
where $i_\eps^{-1}:G^{\mE^\eps}\to G^{\mE}$ is the pull-back map.
\end{lemma} 

\begin{proof}
This is essentially \cite[Theorem 8.10]{Driver89}, with the only difference here being that $i_\eps$ is not required to be surjective.
However, for the proof of \cite[Theorem 8.10]{Driver89}, which relies on \eqref{dri:dis1}, it suffices for $i_\eps:\mE\to \mE^\eps$ to be surjective outside a 
tree in $\mG^\eps$.
We can always find a tree $T$ in $\mG$ such that 
 $T(\eps)\eqdef (\mE^\eps\backslash i_\eps(\mE)) \cup i_\eps(T)$ is also a tree. 
 (For instance if $T$ is empty, $\mE^\eps\backslash i_\eps(\mE)$ is a disjoint collection of edges which is a tree.)
 We then fix an axial gauge 
 by  choosing the edge variables on $T(\eps)$ to be identity.
Since $f$ is a gauge invariant function
\begin{align*}
	\mu^\eps(f\circ i_\eps^{-1})&\stackrel{\eqref{dri:dis1}}{=}
	\int_{G^{(\mE^\eps)^+}} f(i_\eps^{-1}(Q)) \prod_{F^\eps\in \mF^\eps} S^\eps_{\frac{| F^\eps|}{\eps^2}} (h_{\p F^\eps}(Q))\dif_{T(\eps)} Q
	\\&=
	\int_{G^{i_\eps(\mE^+)}} f(i_\eps^{-1}(Q)) \prod_{F^\eps\in \mF^\eps} S^\eps_{\frac{| F^\eps|}{\eps^2}} (h_{\p F^\eps}(Q))\dif_{i_\eps(T)} Q
	\\&=
	\int_{G^{\mE^+}} f(Q) \prod_{F\in \mF} S^\eps_{\frac{| j_\eps(F)|}{\eps^2}} (h_{\p F}(Q))\dif_{T} Q
	\\	&\stackrel{\eps\to 0}{\longrightarrow} \int_{G^{\mE^+}} f(Q)  \prod_{F\in \mF}p_{|F|}(h_{\p F}(Q))\dif_T Q = \mu(f),
\end{align*}
where we used $Q\big|_{T(\eps)}=I$ in the second step,
and change of variable and 
$(i_\eps^{-1}Q)_{\p F}=Q_{\p F_\eps}$ 
(which follows from $Q\big|_{\mE^\eps\backslash i_\eps(\mE)}=I$)
 in the third step,
and the fourth line follows from \eqref{eq:a-cons} and the requirement that $|F^\eps|$ and $|F|$ differ by $O(\eps)$. The last step again uses gauge invariance of $f$.
\end{proof}
 
 \br\label{re:ax} 
 When we apply Lemma \ref{lem:Driver} in Section \ref{sec:sim} and Section \ref{sec:gen} below, we will   fix an axial gauge with different choices of the tree $T$ in $\mG$, 
and we will ensure that  $T(\eps)\eqdef (\mE^\eps\backslash i_\eps(\mE)) \cup i_\eps(T)$ is also a tree in $\mG^\eps$ as required in 
 the proof of Lemma \ref{lem:Driver}. 
 \er

We say that the set of lattice loops $\{l^\eps\}_{\eps>0}$ is a lattice approximation of the continuum loop $l$ if  $\{\mG^\eps\}_{\eps>0}$ is lattice approximation of $\mG$,
where $\mG$ and $\mG^\eps$ are the graphs induced by the loops $l$ and $l^\eps$ respectively.  
For simplicity we will always assume that $l$ is smooth, so that a lattice approximation $\{l^\eps\}_{\eps>0}$ exists. 
By \eqref{e:lattice-approx}, for any lattice approximation $\{l^\eps\}_{\eps>0}$  of a  loop $l$, 
\begin{equ}[e:conv-loop]
\lim_{\eps \to 0 }\E W_{l^\eps} = \E W_l \;.
\end{equ}
Furthermore, for a sequence of loops $l_1,\cdots,l_m$, 
and any lattice approximation $l_1^\eps,\cdots,l_m^\eps$, we have 
\begin{equ}\label{e:con:loops}
	\lim_{\eps \to 0}\E \Big(\prod_{j=1}^m W_{l_j^\eps}\Big) = \E  \Big(\prod_{j=1}^m W_{l_j}\Big) \;.
\end{equ}

\br \label{rem:eps}
Throughout the paper, to simplify the notation, when we write expressions for lattice Wilson loop expectations
such as \eqref{eq:wl-si}, \eqref{eqc:WLd}, \eqref{eq:loop:g}, \eqref{eq:loop:g1}
(and already in e.g. \eqref{eq:wl}),
we often omit the $\eps$ in the notation for loops, edges, and edge variables 
(whose dependence on $\eps$ is clear from the context),
and only write a  superscript for the Wilson loop, e.g. $W_l^\eps= W_{l^\eps}$
to remind the dependence on $\eps$.
\er

\section{Gaussian approximation lemma}
\label{sec:Gaussian}


Given a smooth function $f: G\mapsto \R$, we are interested in the asymptotics as $\eps \to 0$ of the integral
\begin{equation}
	\int f(Q)S^{\eps}(Q)\dif Q. \nonumber
\end{equation}
We will argue that
\begin{equation} \label{eee1}
	\int f(Q)S^{\eps}(Q)\dif Q=f(I)+\frac{1}{2}\eps^{2} \Delta f(I)+O(\eps^{4}) . \
\end{equation}
We are inspired in this section by \cite{BS83} (see also \cite{Gawedzki1982}), which establishes \eqref{e:approx-a}.  Notice that for the functions $Q \mapsto \frac{1}{d_{\tau}}\chi_{\tau}(Q)$ , the asymptotics \eqref{eee1} is easily implied by \eqref{e:approx-a} upon Taylor expansion of the exponential, taking into account that $\Delta \chi_{\tau}(I)=c_{\tau}d_{\tau}$, which is not surprising since $\chi_{\tau}$ are precisely the eigenfunctions of the Laplace-Beltrami operator and $c_{\tau}$ the corresponding eigenvalues.

In the remainder of the article, we will use a few other choices of the function $f$, so it is useful to have a general result in this direction.

\medskip

To quantify the $O(\eps^{4})$ error term in \eqref{eee1}, we define a norm to measure $f$.  Let $\delta_{0}$ be such that the map $A \in \mathfrak{g} \mapsto e^{A} \in Q$ is invertible for $|A|< \delta_{0}$, where $|A|$ denotes the Hilbert-Schmidt norm defined through the inner product \eqref{def:inn}.
We then define
\begin{equs}[e:norm-f]
	\|f\|&\eqdef \sup_{Q \in G }|f(Q)|  \\
	&+\sup_{0<|A|<\delta_{0}} \bigg ( |A|^{-2}|\mathcal{L}_{A}^{2}f(I) |+\sup_{t \in [-1,1]}|A|^{-4} \bigg |\frac{\dif^{4}}{\dif t^{4}}f(e^{tA}) \bigg | \bigg ) \;.
\end{equs}

\begin{lemma}\label{lem:ga}
	Let $f: G\mapsto \R$ be a smooth function such that $f(I)=0$.  There exist constants $C=C(\mathfrak{g}),\eps_{0}=\eps_{0}(\mathfrak{g})$ such that for all $\eps<\eps_{0}$ 
\begin{equation}\label{eq:ap-f}
	\bigg | \int f(Q)S^{\eps}(Q)\dif Q-\frac{1}{2}\eps^{2} \Delta f(I)  \bigg | \leq  C \|f\| \eps^{4} .
\end{equation} 
\end{lemma}

\begin{proof}  Note that for $\eps$ sufficiently small, $S^{\eps}$ concentrates near $I$, so it's natural to apply Laplace's method. 
	Define $\tilde{K}_{\delta}$  to be a neighborhood of $0$ in $\mathfrak{g}$ for some $\delta<\delta_{0}$ to be selected below. For $A\in \tilde{K}_{\delta}$, $|A|\leq \delta$.  Let ${K}_{\delta}$ be the image of $\tilde K_{\delta}$ under the exponential map. Since $S^{\eps}(Q)=S^\eps(Q^{-1})$, we may symmetrize $f$ to find
	\begin{equation}
		\int f(Q)S^{\eps}(Q)\dif Q=\frac{1}{2}\int \big (  f(Q)+f(Q^{-1}) \big )S^{\eps}(Q)\dif Q\;.
	\end{equation}
We divide the integral into the `small field' region $K_{\delta}$ and the `large field' region $K_{\delta}^{c}$.  In the large field region it holds $\text{Re Tr}(I-Q) \geq c$ for some $c=c(\delta)$ and hence
	\begin{equation}
		Z^{\eps}\bigg | \int_{K_{\delta}^{c}} \big (  f(Q)+f(Q^{-1}) \big ) S^{\eps}(Q)\dif Q \bigg | \leq 2 \sup_{Q \in G} |f(Q)| e^{-cN\eps^{-2}} \leq  2\|f\| e^{-cN\eps^{-2}}, \nonumber
	\end{equation}
	where $Z^\eps$ is the normalization defined by \eqref{ac:wi}.
	To analyze the small field region, we write
	\begin{equation}
		\int_{K_{\delta}} \big ( f(Q)+f(Q^{-1}) \big ) S^{\eps}(Q)\dif Q=\int_{\tilde{K}_{\delta} } \big (f(e^{A})+f(e^{-A}) \big ) J(A)S^{\eps}(e^A)\dif A \nonumber\;,
	\end{equation}
	where $J(A)$ denotes the Jacobian from changing variables to the Lie algebra.  We now Taylor expand the action and claim that  
	$$-\eps^{-2}N\Re\text{Tr}(I-e^{A})=-\frac1{2\eps^2}|A|^2+Y(A)\;,$$
	with $|Y(A)|\leq \frac1{N\eps^24!}|A|^4$.  Indeed, consider the function $t \mapsto g(t):=-\text{Re Tr}(I-e^{tA})$ and note that $g^{(k)}(t)=\text{Re Tr} A^{k}e^{tA}$  for all $k \in \N$.  In particular, since $A^{*}=-A$, it holds $g(0)=g^{(k)}(0)=0$ for all odd $k$.  Furthermore, $g^{(2)}(0)=-\text{Re}\text{Tr}AA^{*}=-N^{-1}|A|^{2}$ and $|g^{(4)}(t) | \leq N^{-2}|A|^{4}$ by the spectral theorem.
	 Hence, the claim follows by Taylor expansion.   In addition, we have
	\begin{equs}[e:GaussH]
		S^{\eps}(e^A)=\frac1{Z^\eps}e^{-\frac{1}{2\eps^{2} }|A|^{2}}\big (1+H(A) \big ),
	\end{equs}
	with 
	\begin{equs}[bd:H]|H(A)|=|e^{Y(A)}-1|=\Big|Y(A)\int_0^1e^{rY(A)}\dif r\Big| \leq \eps^{-2}N^{-1} |A|^{4}e^{\frac{1}{2\eps^{2}}\delta^{2} |A|^{2} },\end{equs} for $A \in \tilde{K}_{\delta}$.
	Inserting the above, we find that
	\begin{align}
		&\int_{K_{\delta}} \big [  f(Q)+f(Q^{-1}) \big ]S^{\eps}(Q)\dif Q \nonumber \\
		&=\frac1{Z^\eps}\int_{\tilde{K}_{\delta}}\big [ f(e^{A})+f(e^{-A}) \big ](1+H(A))J(A)e^{-\frac{1}{2\eps^{2} }|A|^{2}} \dif A  \nonumber\\
		&= \frac1{Z^\eps}\int_{\tilde{K}_{\delta}}\big [ f(e^{A})+f(e^{-A}) \big ]e^{-\frac{1}{2\eps^{2} }|A|^{2}}\dif A +\frac1{Z^\eps}R, \label{e:GaussR}
	\end{align}
	where
	\begin{equ}
		R=N_{0} \int_{\tilde{K}_{\delta} } \Big ( f(e^{A})+f(e^{-A}) \Big ) \Big ( H(A)+ (J(A)-1)(1+H(A)) \Big )\dif \nu_{\eps}(A)
	\end{equ}
	and
	\begin{equ}[e:def-N01]
		\dif\nu_\eps=
		\frac{1}{\int e^{-\frac{|A|^2}{2\eps^2}}\dif A}
		e^{-\frac{|A|^2}{2\eps^2}}\dif A\;,
		\qquad
		N_0=
		\int e^{-\frac{1}{2\eps^2}|A|^2}\dif A \;.
	\end{equ}
	
	To estimate the remainder, we need a few auxilliary estimates.  One estimate, which we quote directly from \cite[(A.31)]{BS83} is that for some $C=C(\mathfrak{g})$,  it holds 
	\begin{align}\label{bd:JA}
		|J(A)-1| \leq C |A|^{2},
		\end{align}
	 for $A \in \tilde{K}_{\delta}$. This also follows from the fact that the Haar measure is invariant under the transformation $A\to -A$. We also need the following estimate 
	\begin{equation}
		\big |f(e^{A})+f(e^{-A})-\mathcal{L}_{A}^{2}f(I) \big | \leq 2 \|f\||A|^{4} \label{EE50},
	\end{equation}
	together with the simple corollary
	\begin{equation}
		\big |f(e^{A})+f(e^{-A}) \big | \leq C \|f\||A|^{2} \label{EE51},
	\end{equation}
	both of which hold for $A \in \tilde{K}_{\delta}$.  To argue \eqref{EE50}, we apply Taylor's theorem to the function $t \mapsto f(e^{tA})$ to obtain the bound
	\begin{align}
		\big | f(e^{A})-\mathcal{L}_{A} f(I)-\frac{1}{2} \mathcal{L}_{A}^{2} f(I)-\frac{1}{3!} \mathcal{L}_{A}^{3}f(I) \big | \leq \|f\| |A|^{4} \nonumber, 
	\end{align}
	where we used the centering $f(I)=0$.  The same estimate also holds with $A \mapsto -A$, and combining these two bounds, taking into account the linearity of $A \mapsto \mathcal{L}_{A}f(I)$, the odd order derivatives cancel and we obtain \eqref{EE50}. The bound \eqref{EE51} follows from \eqref{EE50} by the triangle inequality and the definition of $\|f\|$.
	
	\medskip
	We now apply \eqref{EE51} and \eqref{bd:H}, \eqref{bd:JA} to obtain 
	\begin{align}
		|R|& \lesssim N_{0} \|f\|  \int_{\tilde{K}_{\delta} } |A|^{2} \big [ \eps^{-2}|A|^{4}+ |A|^{2}(1+\eps^{-2}|A|^{4} ) \big ]e^{\frac{1}{2\eps^{2}}\delta^{2} |A|^{2}}\dif \nu_{\eps}(A) \nonumber \\
		&\lesssim   \|f\| \int_{\tilde{K}_{\delta} } \big ( \eps^{-2}|A|^{6}+|A|^{4} \big ) e^{-\frac{1-\delta^2}{2\eps^{2}} |A|^{2}}\dif A \nonumber \\
		&\lesssim N_0\|f\| \eps^4,  \nonumber
	\end{align}
	where we choose $\delta^2<\text{min}(1/2,\delta_0^2)$.  In the last line, we change to a Gaussian with density proportional to $e^{-\frac{1}{4\eps^2}|A|^2}$ and used that $N_{0}^{-1} \int e^{-\frac{1}{4\eps^2}|A|^2}\dif A \lesssim 1$. 
	
	Finally, concerning the first term on the RHS of \eqref{e:GaussR}, we need the error bound 
	\begin{equation}
		\bigg | \int_{\tilde{K}_{\delta}}\big [f(e^{A})+f(e^{-A}) \big]\dif \nu_{\eps}(A) -\int_{\mathfrak{g}}\mathcal{L}^{2}_{A}f(I)\dif \nu_{\eps}(A) \bigg | \lesssim \eps^{4}\|f\|,\label{EE52}
	\end{equation}
	which follows from \eqref{EE50} and the following tail estimate for the Gaussian measure $\nu^{\eps}$.
	$$\Big|\int_{\mathfrak{g}\backslash \tilde K_\delta}\mathcal{L}^{2}_{A}f(I)\dif \nu_{\eps}(A)\Big|\lesssim\|f\| e^{-\frac{\delta^2}{4\eps^2}} \Big|\int_{\mathfrak{g}\backslash \tilde K_\delta}|A|^2e^{-\frac1{4\eps^2}|A|^2}\dif A\Big|\lesssim \eps^{4}\|f\|.$$
	Here we used that $|A|\geq\delta$ for $A\in \mathfrak{g}\backslash \tilde K_\delta$. 
	
Note also that (recalling the notation \eqref{e:CL})
\begin{equ}
\int_{\mathfrak{g}}  \mathcal{L}_{A}^{2}f(I) \dif\nu_{\eps}(A)
=\sum_{i,j}\frac1{N_0}\mathcal{L}_i\mathcal{L}_j f(I)\int_{\R^{d(\mathfrak{g})}} s_{i}s_{j}  e^{-\frac{1}{2 \eps^{2}}|s|^{2} }\dif s 
=\Delta f(I) \eps^{2} \;.
\end{equ}
The proof is completed by combining the above with the bound $|\frac{Z^{\eps} }{N_{0}}-1 | \lesssim \eps^{2}$, which follows from a similar line of argument. 
Indeed, using \eqref{bd:JA} and \eqref{bd:H}
\begin{align}
Z^{\eps}=\int e^{-\eps^{-2}N\text{ReTr}(I-Q)}\dif Q&=O(e^{-c\eps^{2}}) +\int_{\tilde{K}_{\delta}} (1+H(A))J(A)e^{ -\frac{1}{2 \eps^{2}}|A|^{2}}\dif A \nonumber. \\
	&= \int_{\mfg}e^{ -\frac{1}{2 \eps^{2}}|A|^{2}}\dif A \Big (1+O(\eps^{2}) \Big ), \nonumber 
\end{align}
where the last line is argued similar to the bound for $R$.  To conclude, choose $\eps$ sufficiently small that $Z^{\eps} \geq \frac{1}{2}N_{0}$ and $|\frac{N_{0}}{Z^{\eps}}-1| \lesssim \eps^{2}$.  The large field contribution is then bounded by
\begin{equation}
\bigg | \int_{K_{\delta}^{c}} \big (  f(Q)+f(Q^{-1}) \big ) S^{\eps}(Q)\dif Q \bigg | \lesssim \frac{\|f\|}{N_{0}}e^{-cN\eps^{-2}} \lesssim \|f\|\eps^{4},
\end{equation}
for $\eps$ sufficiently small.  For the small field contribution, we use \eqref{e:GaussR}, \eqref{EE52} and the above estimates  to write
\begin{equation}
\frac{1}{2}\int_{K_{\delta}} \big [  f(Q)+f(Q^{-1}) \big ]S^{\eps}(Q)\dif Q=\frac{N_{0}}{Z^{\eps}}\frac{1}{2}\eps^{2}\Delta f(I)+\frac{N_{0}}{Z^{\eps}}O(\eps^{4}\|f\|) \nonumber,  
\end{equation}
then use that $|\Delta f(I)| \lesssim \|f\|$. 
\end{proof}

 \section{Simple loops}\label{sec:sim}

 We start with  simple loops, which are loops without any self-intersection.
In this case the proof is easier, but demonstrates some of our ideas in the general cases. 
In particular we aim to recover the loop equation in \cite[Section~2.2, Theorem~2.3]{Driver17}, which was first derived in \cite[Proposition 6.4]{Levy11}. 

For a simple loop $l$ on $\R^2$, which encloses an area $t$, we can easily find a lattice approximation $\{l^\eps\}$, consisting of simple loops on $(\eps \Z)^2$,  
and each $l^\eps$ encloses area $t_\eps$.
Obviously $t_\eps/\eps^2$ is a positive integer. 
 By the definition of lattice approximation we know that $|t_\eps-t|\lesssim \eps$.
 We then have the following discrete master loop equation  for the loop $l^\eps$, i.e. \eqref{eq:wl}   can be written as 
\begin{align}\label{eq:wl-si}
\E	W_l^\eps=	\frac{1}{2\eps^2} \!\!\! \sum_{l'\in \mathbb{D}^-_e(l) }\!\!\E W_{l'}^\eps
	-\frac{1}{2\eps^2} \!\!\! \sum_{l'\in \mathbb{D}^+_e(l)} \!\!\E W_{l'}^\eps
\end{align}
where $e=e_\eps$ with $|e_\eps|=\eps$ 
and we recall our convention Remark~\ref{rem:eps}.


\begin{theorem}\label{th:1} 
As $\eps\to0$,  one has
\begin{equ}[e:th:1]
\frac{1}{2\eps^2} \!\!\! \sum_{l'\in \mathbb{D}^-_e(l)}\!\!\E W_{l'}^\eps
-\frac{1}{2\eps^2} \!\!\! \sum_{l'\in \mathbb{D}^+_e(l)} \!\!\E W_{l'}^\eps
\to -2\frac{\dif }{\dif t}\E W_l\;.
\end{equ}
In particular,
the discrete master loop equation \eqref{eq:wl-si} converges to 
\begin{equ}\label{eq:mqs}
\frac{\dif }{\dif t} \E W_l=-\frac12  \E W_l\;.
\end{equ} 
\end{theorem}

This recovers the master loop equation for simple loops in \cite[Theorem~2.3]{Driver17}. 

Note that while a simple application of Driver's formula  will directly yield \eqref{eq:mqs} ($\E W_l$ is just a function of a single variable $t$ for a simple loop $l$), the point of Theorem~\ref{th:1} is to demonstrate 
how various terms in  \eqref{eq:wl-si} combine  (or cancel) to give the area derivative.
To this end, we examine the terms  \eqref{eq:wl-si}  more carefully as follows.
Since $d=2$, there is a natural notion 
of {\it inner and outer} deformations for a simple loop $l$, and we write 
 $l_{i,-}$, $l_{i,+}$, $l_{o,-}$, $l_{o,+}$ for 
inner negative, inner positive, outer negative and outer positive deformations respectively.  With the above notations, we then have $ \mathbb{D}^-_e(l)=\{l_{i,-}, l_{o,-}\}$ and $\mathbb{D}^+_e(l)=\{l_{i,+}, l_{o,+}\}$.   
\[
\begin{tikzpicture}[scale=.9]
\draw[thick] (0,0) -- (0,2) -- (1,2) -- (1,1.5);
\draw[thick,midarrow]   (1,1.5) to (1.5,1.5) ;
\draw[thick] (1.5,1.5) -- (1.5,2) -- (2.5,2) -- (2.5,0)  -- (0,0);
\node at (1.3,-0.5) {$l_{i,-}$};
\end{tikzpicture}
\qquad
\begin{tikzpicture}[scale=.9]
\draw[thick] (0,0) -- (0,2) -- (1,2) -- (1,2.5);
\draw[thick,midarrow]   (1,2.5) to (1.5,2.5) ;
\draw[thick] (1.5,2.5) -- (1.5,2) -- (2.5,2) -- (2.5,0)  -- (0,0);
\node at (1.3,-0.5) {$l_{o,-}$};
\end{tikzpicture}
\qquad
\begin{tikzpicture}[scale=.9]
\draw[thick,red] (0,0) -- (0,2) node[midway,left] {$e_1$} -- (1,2);
\draw[thick,midarrow1] (1,2) -- (1.5,2) node[midway,above] {$e$};
\draw[thick] (1.5,2) -- (1.5,1.5);
\draw[thick,midarrow] (1.5,1.5) -- (1,1.5) node[midway,below] {$e_2$};
\draw[thick] (1,1.5) -- (1,2);
\draw[thick,midarrow2]  (1,2) to (1.5,2);
\draw[thick,red] (1.5,2) -- (2.5,2) -- (2.5,0);  \draw[thick,midarrow,red] (2.5,0) -- (0,0);
\node at (1.3,-0.5) {$l_{i,+}$};
\end{tikzpicture}
\quad
\begin{tikzpicture}[scale=.9]
\draw[thick,red] (0,0) -- (0,2) node[midway,left] {$e_1$} -- (1,2);
\draw[thick,midarrow1] (1,2) -- (1.5,2) node[midway,below] {$e$};
\draw[thick] (1.5,2) -- (1.5,2.5);
\draw[thick,midarrow] (1.5,2.5) -- (1,2.5) node[midway,above] {$e_3$};
\draw[thick] (1,2.5) -- (1,2);
\draw[thick,midarrow2]  (1,2) to (1.5,2);
\draw[thick,red] (1.5,2) -- (2.5,2) -- (2.5,0);  \draw[thick,midarrow,red] (2.5,0) -- (0,0);
\node at (1.3,-0.5) {$l_{o,+}$};
\end{tikzpicture}
\]

 \begin{equs}[eq:def]
 	\E W_{l_{i,-}}^\eps 
	& =\int \tr(Q)S^\eps_{\frac{t_\eps-\eps^2}{\eps^2}}(Q)\dif Q\;, 
\\
 	\E W_{l_{o,-}}^\eps
	& =\int \tr(Q)S^\eps_{\frac{t_\eps+\eps^2}{\eps^2}}(Q)\dif Q\;, 
\\
 	\E W_{l_{i,+}}^\eps
	&=\int \tr(Q_1Q_eQ_2Q_e)S^\eps_{\frac{t_\eps-\eps^2}{\eps^2}}(Q_{1}Q_2^{-1})S^\eps(Q_eQ_2)\dif Q_1\dif Q_2\dif Q_e\;,
\\
 	\E W_{l_{o,+}}^\eps
	&=\int \tr(Q_1Q_eQ_3Q_e)S^\eps_{\frac{t_\eps}{\eps^2}}(Q_{1}Q_e)S^\eps(Q_eQ_3)\dif Q_1\dif Q_3\dif Q_e\;.
 \end{equs}
Indeed, this follows from \eqref{dri:dis} by taking different choices of the graph.  For the negative deformations, we consider the graph built from a single edge which traverses the whole loop. 
 For the positive deformations, we take a graph consisting of three edges, as indicated in the above figure.  For example, to calculate $\E W_{l_{i,+}}^\eps$ we have three edges $e$, $e_{1}$, $e_{2}$ (where $e_1$ is colored red for clarity)
and corresponding edge variables $Q_{e}$, $Q_{1}$, $Q_{2}$.  There are two bounded faces $F_{1}$ and $F_{2}$ with $|F_{1}|=t_\eps-\eps^2$, $\partial F_{1}=e_{1}e_{2}^{-1}$,  and $|F_{2}|=\eps^{2}$, $\partial F_{2}=ee_{2}$, so that \eqref{dri:dis} yields the identity above. Similarly, to calculate $\E W_{l_{o,+}}^\eps$, we have edges $e$, $e_{1}$, $e_{3}$ which generate two bounded faces $\tilde{F}_{1},\tilde{F}_{2}$ with $|\tilde{F}_{1}|=t_{\eps}$, $\partial \tilde{F}_{1}=e_{1}e$ and $|\tilde{F}_{2}|=\eps^2$, $\partial \tilde{F}_{2}=e e_{3}$. 
We remark that 
each negative deformation here yields a function which only depends on a single area variable and obviously  approximates $\E W_l$, this is not the case for the positive deformations.

\medskip 

We now illustrate the role of axial gauge fixing.  Namely, in the above expressions for $\E W_{l_{i,+}}^\eps$ and $\E W_{l_{o,+}}^\eps$ we can reduce from three integration variables to two and effectively `fix' $Q_{e}=I$ by applying \eqref{dri:dis1} with $T=e$ to obtain
\begin{align}
\E W_{l_{i,+}}^\eps
	&=\int \tr(Q_1Q_2)S^\eps_{\frac{t_\eps-\eps^2}{\eps^2}}(Q_{1}Q_2^{-1})S^\eps(Q_2)\dif Q_1\dif Q_2, \label{eq:innerPos}\\
 \E W_{l_{o,+}}^\eps
	&=\int \tr(Q_1Q_3)S^\eps_{\frac{t_\eps}{\eps^2}}(Q_{1})S^\eps(Q_3) \dif Q_1\dif Q_3.  \label{eq:outerPos}
\end{align}
We now turn to the analysis of the deformations in \eqref{eq:wl-si}, which we may write as
\begin{align}
&\frac{1}{2\eps^2} \!\!\! \sum_{l'\in \mathbb{D}^-_e(l) }\!\!\E W_{l'}^\eps
	-\frac{1}{2\eps^2} \!\!\! \sum_{l'\in \mathbb{D}^+_e(l)} \!\!\E W_{l'}^\eps \nonumber \\
&=\frac{1}{2\eps^{2}} \E\big (W_{l_{o,-}}^\eps-W_{l_{o,+}}^\eps \big ) + \frac{1}{2\eps^{2}} \E\big (W_{l_{i,-}}^\eps-W_{l_{i,+}}^\eps \big ) \nonumber.
\end{align}
We first argue that the outer deformations cancel.
\begin{lemma}\label{lem:out} 
It holds that 
 $$\E W_{l_{o,+}}^\eps=\E W_{l_{o,-}}^\eps.$$
\end{lemma}
\begin{proof}
The proof follows from changing variables in $Q_{1}$ then using the definition of convolution in $Q_{3}$.  Indeed, by translation invariance of the Haar measure, for all $Q_{3}$ it holds
\begin{equation}
\int \tr(Q_1Q_3)S^\eps_{\frac{t_\eps}{\eps^2}}(Q_{1})\dif Q_1=\int \tr(Q_1)S^\eps_{\frac{t_\eps}{\eps^2}}(Q_{1}Q_{3}^{-1} )\dif Q_1 \nonumber,
\end{equation}
and for all $Q_{1}$ it holds, c.f. \eqref{eq:defconv}
\begin{equation}
\int  S^\eps_{\frac{t_\eps}{\eps^2}} (Q_{1}Q_{3}^{-1} )S^\eps(Q_3)\dif Q_3=S^{\eps}_{\frac{t_\eps+\eps^{2} }{\eps^2}}(Q_{1}) \nonumber,  
\end{equation}
so combining the two with \eqref{eq:outerPos} gives the claim.
\end{proof}

We proceed to the inner deformations. 
\begin{proof}[Proof of Theorem~\ref{th:1}] 
We change variables in $Q_{1}$ to have
\begin{align*}
\E W_{l_{i,+} }^\eps
	&=\int \tr(Q_1Q_2^{2} )S^\eps_{\frac{t_\eps-\eps^2}{\eps^2}}(Q_{1})S^\eps(Q_2)\dif Q_1\dif Q_2.
\end{align*}
Hence, 
\begin{align*}
	\E W_{l_i,-}^\eps-	\E W_{l_i,+}^\eps
	&=\int \tr\Big(Q_1(I-Q_2^{2}) \Big)S^\eps_{\frac{t_\eps-\eps^2}{\eps^2}}(Q_{1})S^\eps(Q_2)\dif Q_1\dif Q_2.
\end{align*}

We can now apply the Gaussian approximation Lemma \ref{lem:ga} in the $Q_{2}$ integral to the function $f(Q)= \text{tr}\big(Q_{1}(I-Q^{2})\big)$.
Indeed, noting that
\begin{equation*}
	\mathcal{L}_{j}f(Q)=-2\text{tr}(Q_1L_{j}Q^{2}), \quad \mathcal{L}_{j}^{2}f(Q)=-4\text{tr}(Q_1L_{j}^{2}Q^{2}),
\end{equation*}
we obtain
 \begin{align*}
\E W_{l_i,-}^\eps-	\E W_{l_i,+}^\eps
	=&-2\eps^2\sum_j\int \tr(Q_1L_j^2 )S^\eps_{\frac{t_\eps-\eps^2}{\eps^2}}(Q_{1})\dif Q_1+O(\eps^4)
	\\=&-2\eps^2\int \Delta\tr(Q_1 )S^\eps_{\frac{t_\eps-\eps^2}{\eps^2}}(Q_{1})\dif Q_1+O(\eps^4),
\end{align*}
where we use $\sup_{Q_1}\|f\|<\infty$ by an elementary calculation and the compactness of Lie group $G$. 
Hence, we use \eqref{eq:a-cons} to obtain
\begin{align*}
	&\lim_{\eps \to 0} \frac{1}{2\eps^{2}} \E\big (W_{l_{i,-}}^\eps-W_{l_{i,+}}^\eps \big ) \nonumber \\
	&=-\int \Delta \text{tr}(Q) p_{t}(Q )  \dif Q
	=-2\frac{\dif}{\dif t}\int \text{tr}(Q) p_{t}(Q )  \dif Q \nonumber,
\end{align*}
where we used  integration by parts, and finally $\partial_{t}p_{t}=\frac{1}{2}\Delta p_t$ . 
\end{proof}

We make some remarks on the above proof before moving on to the general case. 
Note that passing $\eps\to 0$ on both sides of \eqref{eq:wl-si} together with \eqref{eq:mqs} would yield \eqref{e:th:1} (so in particular \eqref{eq:wl-si}, \eqref{e:th:1}, \eqref{eq:mqs} are ``consistent''). Of course, it is not our purpose to prove \eqref{e:th:1} using  \eqref{eq:mqs}.
More importantly, our goal here is to analyze the limiting behavior of {\it each} deformation term in the lattice loop equation.
In fact, even the existence of limit for each positive deformation term is not a priori guaranteed by the general result \eqref{e:conv-loop}, or  even if we take for granted that $\E W_{l_{i,+}},\E W_{l_{o,+}}$ both converge to $\E W_{l}$, it is not clear how that helps us pass to the limit in $\frac{1}{2\eps^2} \! \sum_{l'\in \mathbb{D}^-_e(l) }\E W_{l'}^\eps
-\frac{1}{2\eps^2} \! \sum_{l'\in \mathbb{D}^+_e(l)} \E W_{l'}^\eps$, due to the factor of $\eps^{-2}$. 
Our analysis above 
 shows that two of the deformation terms cancel and the other two of them yields a Laplacian $\Delta$ that is turned into an area derivative, 
which is stronger than convergence of the sum of the four terms.

 \section{General loops}\label{sec:com}
 In this section we consider a loop  given by a closed curve with simple crossings 
 and let $v$ be a crossing. We prove that suitable linear combination of  the discrete master loop equations converge to the Makeenko-Migdal equation derived in \cite[Proposition 6.22]{Levy11} and later in \cite[Theorem 1.1]{Driver17}.   In Section \ref{sec:4.1} we first prove two useful convergence lemmas. In Section \ref{sec:gen} we prove Theorem \ref{th:main}.  In Section \ref{sec:deg} we consider some degenerate case. 
 
 \subsection{Useful convergence lemmas}\label{sec:4.1}
 We first prove two useful lemmas. Recall our notation \eqref{e:CL} for $L_j$ and $\cL_j$.

 \bl\label{lem:J} 
 For $a_1, a_2\in G$ and every representation $\tau$, 
 \begin{equs}[eq:trq-1]
 	&\frac{1}{2\eps^{2}}\int \tr\Big((Q-Q^{-1})a_1 \Big)\chi_{\tau}(Qa_2)S^\eps(Q) \dif Q
 \\&	=\sum_{j=1}^{d(\mfg)}\cL_j\tr(a_1 )\cL_j\chi_{\tau}(  a_2)+O\Big((|c_\tau|^{2}+1)d_\tau\eps^2\Big).
 \end{equs}
The proportional constant in the error term is independent of
$\tau$, $a_1$ and $a_2$.
 \el

 \begin{proof} We also apply Lemma \ref{lem:ga}. 
 	Let us now consider  the function 
 	\begin{equation}
 		f(Q)\eqdef  f_{1}(Q)f_{2}(Q), \qquad f_{1}(Q)\eqdef\text{tr}( (Q-Q^{-1})a_{1} ), \quad f_{2}(Q)\eqdef\text{Tr}[\tau(Q)\tau(a_{2}) ] \nonumber.
 	\end{equation}
 	By direct calculation, we find that for $X\in \mfg$ 
 	\begin{align*}
 		\mathcal{L}_{X}f_{1}(Q)&=\text{tr} \Big [ (XQ-Q^{-1}X^*)a_{1} \Big ], 
 		\\  
		 \mathcal{L}_{X}^2f_{1} (Q)&=\text{tr} \Big [ (X^2Q- Q^{-1}(X^*)^2)a_{1} \Big ] \nonumber,
 	\end{align*}
 	and in particular we notice that $f_{1}(I)= \mathcal{L}_{X}^2f_{1}(I)=0$  and $\mathcal{L}_{X}f_{1}(I)=2\text{tr}[Xa_{1}] $.  Hence, by Leibniz rule we find that
 	\begin{align*}
 		 \mathcal{L}_{X}^2 f (I)
 		&=   \mathcal{L}_{X}^2 f_{1} (I) f_{2}(I)+  f_{1} (I) \mathcal{L}_{X}^2f_{2}(I)+ 2\mathcal{L}_{X}f_{1}(I)\mathcal{L}_{X}f_{2}(I) \nonumber \\
 		&= 2\mathcal{L}_{X}f_{1}(I)\mathcal{L}_{X}f_{2}(I) \nonumber.
 	\end{align*}
 	Noting that $\mathcal{L}_{X}f_{2}(Q)=\text{Tr}[\tau(X)\tau(Q)\tau(a_2) ]$, we find that
 	\begin{equation}
 		\frac{1}{2}\Delta f(I)= 2\sum_{j=1}^{d(\mfg)}\text{tr}[L_ja_{1}]\Tr\big(\tau(L_j)\tau(a_{2}) \big)=2\sum_{j=1}^{d(\mfg)}\cL_j\tr(a_1 )\cL_j\chi_{\tau}(  a_2) \nonumber. 
 	\end{equation}
 	It remains to estimate $\|f\|$ given by \eqref{e:norm-f}.  
 	We have 
 	\begin{align}\label{bd:fQ}
 		|f(Q)|\leq 2|f_2(Q)|\leq2 \chi_\tau(I)=2d_\tau\;.
 	\end{align}
 	For  $A\in\mfg$
 	$$\cL_A^2 f(I)=4\text{tr}[ Aa_{1}]\Tr\big(\tau(A)\tau(a_2)\big)\;,$$ 
 	we use \eqref{e:Ctau} to  have
 	\begin{equs}[bd:LA2]
 		|\cL_A^2 f(I)|&\leq\frac4N |A|\Big|\sum_j\<A,L_j\>\Tr(\tau(L_j)\tau(a_2))\Big|
 		\\&\leq\frac4N |A|\Big(\sum_j|\<A,L_j\>|^2\Big)^{1/2}\Big|\sum_j\Tr\big(\tau(L_j)^2\big)\Tr(\tau(a_2a_2^*))\Big|^{1/2}
 		\\&\leq\frac4N |A|^2|c_\tau|^{1/2}d_\tau.
 	\end{equs}
 	 It remains to consider  $\frac{\dif^{4}}{\dif t^{4}}f(e^{tA})$ for $t\in [-1,1]$. 
 	 Define  $g(t):=f(e^{At}), g_{i}(t)=f_{i}(e^{At}), i=1,2$. We use $g^{(k)}_i(t)$ to denote $k$-th derivative of $g_i$.  Let us estimate $g^{(4)}(t)$ for $t \in [-1,1]$ as required by the above norm. Recall that $\tau(L_j^*)=\tau(L_j)^*=-\tau(L_j)$.  We have
 	\begin{align}
 		&|g_{2}^{(4)}(t)|=|\text{Tr}[\tau(A)^{4}\tau(e^{At})\tau(a_2)]| \nonumber \\
 		&=\Big|\sum_{\substack{i_k=1\\ k=1,\dots,4}}^{d(\mfg)} \prod_{k=1}^4\langle A, L_{i_k} \rangle \text{Tr}\Big( \prod_{k=1}^4\tau(L_{i_k})\tau(e^{At})\tau(a_2) \Big) \Big| \nonumber \\
		& \leq \sum_{i_{1},i_{2} }\langle A,L_{i_{1}} \rangle \langle A,L_{i_{2}} \rangle \text{Tr}\Big(\tau (L_{i_1})^{2} \tau (L_{i_2})^{2} \Big)^{\frac{1}{2} }  \nonumber \\
		&\quad  \times \sum_{i_{3},i_{4}} \langle A,L_{i_{3}} \rangle \langle A,L_{i_{4}} \rangle \text{Tr}\Big(\tau (L_{i_3})^{2} \tau (L_{i_4})^{2}\Big)^{\frac{1}{2} }  \nonumber \\   
 		& \leq |A|^{4} \bigg ( \sum_{\substack{i_k=1\\ k=1,\dots,4}}^{d(\mfg)}\text{Tr}\Big(\tau (L_{i_1})^{2} \tau (L_{i_2})^{2} \Big)  \text{Tr}\Big(\tau (L_{i_3})^{2} \tau (L_{i_4})^{2} \Big ) \bigg)^ {\frac{1}{2}} \nonumber \\
 		&\leq |A|^{4}c_{\tau}^{2}d_{\tau} \nonumber.  
 	\end{align}
 	 Note that in the first inequality we used the following H\"{o}lder's inequality for the trace
 	$$\Tr(B_1B_2)\leq \Tr(B_1B_1^*)^{1/2}\Tr(B_2B_2^*)^{1/2},$$ with $B_1=\tau (L_{i_1})\tau (L_{i_2})$ and $B_2=\tau (L_{i_3}) \tau (L_{i_4})\tau(e^{At})\tau(a_2)$.  In the second inequality, we used H\"{o}lder's inequality for each double summation.  In the third, we used  the definition of $C_{\tau}$ from \eqref{e:Ctau}.
 	
 	  Similarly, we have
 	\begin{align}
 		&|g_1^{(1)}(t)g_{2}^{(3)}(t)|=|\text{tr}  [ (Ae^{At}-e^{-At}A^*)a_{1}  ]\text{Tr}[\tau(A)^{3}\tau(e^{At})\tau(a_2)]| \nonumber \\
 		&\leq\frac2N|A|\Big|\sum_{i,j,k} \langle A, L_{i} \rangle \<A,L_j\> \langle A, L_{k} \rangle \text{Tr}[ \tau(L_{i})\tau(L_j)\tau(L_{k})\tau(e^{At}) \tau(a_2)\ ] \Big| \nonumber \\
 		& \leq\frac2N|A|^{4}\Big|\sum_{i,j,k} \text{Tr}[\tau (L_{i})^{2} \tau (L_{j})^{2} ]   \text{Tr}[\tau (L_{k})^{2} ]\Big|^{\frac{1}{2}} \leq\frac2N |A|^{4}|c_{\tau}|^{3/2}d_{\tau} , \nonumber 
 	\end{align}
 	and 
 	\begin{align}
 		&|g_1^{(2)}(t)g_{2}^{(2)}(t)|=\Big|\text{tr}\Big( (A^2e^{At}-e^{-At}(A^*)^2)a_{1}  \Big)\text{Tr}[\tau(A)^{2}\tau(e^{At})\tau(a_2)]\Big| \nonumber \\
 		&\leq\frac2{N^2}|A|^2\Big|\sum_{i,j} \langle A, L_{i} \rangle \<A,L_j\>  \text{Tr}[ \tau(L_{i})\tau(L_j)\tau(e^{At})\tau(a_2) \ ] \Big| \nonumber \\
 		& \leq\frac2{N^2}|A|^{4}\Big(\sum_{i,j} \text{Tr}[\tau (L_{i})^{2}  ]   \text{Tr}[\tau (L_{j})^{2} ]\Big)^{\frac{1}{2}} 
 		\leq\frac2{N^2} |A|^{4}|c_{\tau}|d_{\tau} \nonumber.  
 	\end{align}
 	We further have
 		\begin{align}
 		&|g_1^{(3)}(t)g_{2}^{(1)}(t)|=\Big|\text{tr}\Big( (A^3e^{At}-e^{-At}(A^*)^3)a_{1}  \Big)\text{Tr}\Big(\tau(A)\tau(e^{At})\tau(a_2)\Big)\Big| \nonumber,
 	\end{align}
 	and
 		\begin{align}
 		&|g_1^{(4)}(t)g_{2}(t)|=\Big|\text{tr}\Big( (A^4e^{At}-e^{-At}(A^*)^4)a_{1}  \Big)\text{Tr}\Big(\tau(e^{At})\tau(a_2)\Big)\Big| \nonumber,
 	\end{align}
 	We write $A=\sum_{j=1}^{d(\mfg)}\<A,L_j\>L_j$ as above and obtain
 		\begin{align*}
 		&|g_1^{(3)}(t)g_{2}^{(1)}(t)| \leq  2|A|^{4}|c_{\tau}|^{1/2}d_{\tau}\;,
 	\end{align*}
 	and
 	\begin{align}
 		&|g_1^{(4)}(t)g_{2}(t)| \leq 2|A|^{4}d_\tau \nonumber.
 	\end{align}
Summarizing  the above calculation we obtain
$$ \sup_{t \in [-1,1]}|A|^{-4} \bigg |\frac{\dif^{4}}{\dif t^{4}}f(e^{tA}) \bigg |\leq O\Big((|c_\tau|^{2}+1)d_\tau\Big),$$
which combined with \eqref{bd:fQ} and \eqref{bd:LA2} implies  the result  by Lemma \ref{lem:ga}. 
 \end{proof}

  \bl\label{lem:J1} Suppose that $\frac{t(\eps)}{\eps^2}\in \Z^+$,  $|t(\eps)-t|\lesssim \eps$. It holds that for $a_1, a_2\in G$ 
 \begin{equ}[eq:trq-2]
 	\lim_{\eps\to 0}\frac1{2\eps^2} \int  \tr\Big((Q-Q^{-1})a_1 \Big)S^\eps_{\frac{t(\eps)}{\eps^2}}(Qa_2)S^\eps(Q) \dif Q
 		= \sum_{j=1}^{d(\mfg)}\cL_j\tr(a_1 )\cL_jp_{t}( a_2)
 \end{equ}
 where  
  the convergence holds uniformly in $a_1, a_2$.
 \el 
 \begin{proof}
 	We use \eqref{e:spec-S} to write the both sides of \eqref{eq:trq-2} as 
 	\begin{align*}
 	\sum_\tau d_\tau a_\tau(\eps)^{\frac{t(\eps)}{\eps^2}}J_\tau^\eps,\qquad \sum_\tau d_\tau e^{\frac{c_\tau t}{2}}J_\tau
 	\end{align*}
 	with  the sum running over all the irreducible representations of $G$ and 
 	$$ 	J_\tau^\eps\eqdef \frac1{2\eps^2}\int \tr\Big((Q-Q^{-1})a_1 \Big)\chi_\tau(Qa_2)S^\eps(Q) \dif Q,$$
 	and 
 	$$ 	J_\tau\eqdef \sum_{j=1}^{d(\mfg)}\cL_j\tr(a_1 )\cL_j\chi_\tau( a_2).$$
 	We need to prove that for $\eps\to0$,
 	\begin{align*}
 		\cA\eqdef \Big|\sum_\tau d_\tau \Big(a_\tau(\eps)^{\frac{t(\eps)}{\eps^2}}J_\tau^\eps- e^{\frac{c_\tau t}{2}}J_\tau\Big)\Big|\to 0.
 	\end{align*}
 		We decompose the sum as 
 		\begin{align*}
 			\cA &\leq \Big|\sum_{\tau\in \cF} d_\tau \Big(a_\tau(\eps)^{\frac{t(\eps)}{\eps^2}}J_\tau^\eps- e^{\frac{c_\tau t}{2}}J_\tau\Big)\Big|
 			\\&\quad+\Big|\sum_{\tau\notin \cF} d_\tau a_\tau(\eps)^{\frac{t(\eps)}{\eps^2}}J_\tau^\eps\Big|+\Big|\sum_{\tau\notin \cF} d_\tau e^{\frac{c_\tau t}{2}}J_\tau\Big|
 			\\&=\cA_1+\cA_2+\cA_3,
 		\end{align*}
 		with $\cF$ being a suitable finite set. 
 		
 		Using Lemma \ref{lem:J} and \eqref{e:approx-a} we know that for any finite set $\cF$, $\cA_1\to0$, as $\eps\to0$. It remains to  choose suitable $\cF$ such that $\cA_2$ and $\cA_3$ close to zero. Since $\cA_3$ is the remainder of the  spectral decomposition of the smooth function
 		 $$\sum_{j=1}^{d(\mfg)}\cL_j\tr(a_1 )\cL_jp_{t}(  a_2)=\sum_\tau d_\tau e^{\frac{c_\tau t}{2}}J_\tau\;,$$
 	we can easily find $\cF$ such that $\cA_3$ is small enough. 
 	
 	In the following we consider $\cA_2$. Using Lemma \ref{lem:J} and \eqref{e:Ctau} we have
 	\begin{align*}
 				|J_\tau^\eps|&\lesssim  \Big(\sum_{j}|\tr(L_ja_1 )|^2\Big)^{1/2}\Big(\sum_{j}|\cL_j\chi_{\tau}( a_2)|^2\Big)^{1/2}+(|c_\tau|^{2}+1)d_\tau\eps^2
 				\\&\lesssim\Big|\sum_{j}\Tr\Big(\tau(L_j)^2 \Big)\chi_{\tau} (a_2a_2^*)\Big|^{1/2}+(|c_\tau|^{2}+1)d_\tau\eps^2
 				\\&\lesssim d_\tau(|c_\tau|^{1/2}+(1+|c_\tau|^{2})\eps^2)\;,
 	\end{align*}
 	where the proportional constant is independent of $\tau$. 
 	Hence, we obtain
 	\begin{align*}
 		\cA_2\lesssim \Big|\sum_{\tau\notin \cF} d_\tau^2 a_\tau(\eps)^{\frac{t(\eps)}{\eps^2}}(|c_\tau|^{1/2}+(1+|c_\tau|^{2})\eps^2)\Big|\;.
 	\end{align*}
 	In \cite[Appendix A]{BS83} it is proved that  for any $\kappa>0$ there exists a finite set $\cF$ such that
 	\begin{align*}
 		\sum_{\tau\notin\cF}d_\tau^2a_{\tau}(\eps)^{\frac{t}{\eps^2}}<\kappa.
 	\end{align*}
 	We then follow the calculation there to prove that there exists a finite set $\cF$ and $\eps_0>0$ such that for $\eps\leq \eps_0$
 	 \begin{align}\label{eq:A2}
 	 	\cA_2<\kappa.
 	 	\end{align}
 	 	We put the proof of \eqref{eq:A2} in appendix. Combining the above calculation, we obtain that $\cA\to0$, $\eps\to0$. Since $G$ is compact, we know that the convergence is uniform in $a_1, a_2$.  The result follows. 
 		\end{proof}
 
 \subsection{Convergence of master loop equations}\label{sec:gen}
 
 In this section we consider general loops of the following form:
\begin{equ}[e:genl]
l=e_1Ae_4^{-1}e_2Be_3^{-1},
\end{equ}
 where $A$ and $B$ are sequences of edges not belonging to $\{e_1,e_2,e_3,e_4\}$. 
 As shown in Fig.~\ref{fig:DHKcase} (which is essentially \cite[Fig. 3]{Driver17}), we consider the crossing point denoted by $v$, at which  $e_1,e_2,e_3,e_4$ meet,
 and we write 
\begin{equ}[e:genl12]
l_1=e_1 Ae_4^{-1},\qquad l_2=e_2Be_3^{-1}.
\end{equ}

We consider 
 $$W_l=\tr(a_{1}\alpha a_{4}^{-1}a_{2}\beta a_{3}^{-1}),$$
where $a_i=Q_{e_i}$ and $\alpha=Q_A$, $\beta=Q_B$,
 and
  $$W_{l_1}=\tr(a_{1}\alpha a_{4}^{-1})\;,
  \qquad W_{l_2}=\tr(a_{2}\beta a_{3}^{-1}) \;.$$

{\bf Notation.} 
Below we will often have products 
of the form $\prod_{i=1}^4 c_{i,i+1}$ where $i+1$ is understood as $1$ when
$i=4$.

\medskip

  We then write the boundary of each face as 
 $$\partial F_1=e_1 A_1e_2^{-1},\quad \p F_2=e_2A_2e_3^{-1}, \quad \p F_3=e_3A_3 e_4^{-1},\quad \p F_4=e_4 A_4 e_1^{-1},$$
 where $A_1,\cdots,A_4$ are certain sets of edges (which might be empty, see, e.g. $F_4$ in Fig.~\ref{fig:DHKcase}).
 Set $t_i=|F_i|$, $i=1,2,3, 4$, and write $\alpha_i=Q_{A_i}$, which is understood as identity
 if $A_i$ is empty.
 By Driver's formula \eqref{eq:Driver-c} we have 
 \begin{equ}[eqc:WL]
 	\E W_{l}=\int \tr(a_1\alpha a_4^{-1}  a_2\beta a_3^{-1})
	\prod_{i=1}^4 p_{t_i}(a_i\alpha_i a_{i+1}^{-1})
	\,\nu(\mathbf{b}) \, \dif \mathbf{a} \, \dif \mathbf{b},
 \end{equ}
where $\mathbf{a}=(a_1,a_2,a_3,a_4)\in G^4$, and $\mathbf{b}$ denotes the edge variables 
for the edges in $A,B$.
Here $\nu(\mathbf{b})$ denotes the product of heat kernels arising from the faces other than $F_1,\cdots,F_4$, which depends only on the $\mathbf{b}$ variables;
 for example in the right picture of Fig.~\ref{fig:DHKcase} it contains a heat kernel arising from the face to the left of $F_2$.
 
In Section~\ref{sec:Dri} we defined general lattice approximation in terms of the maps.  Here we impose an additional condition that the crossing point $v$ should be ``approximated'' by an edge, as follows. 
Let $l$ be
given by the general form \eqref{e:genl}, then
we say that $\{l^\eps\}_{\eps>0}$ is a lattice approximation of $l$
{\it with respect to $v$} if   $\{l^\eps\}_{\eps>0}$ is a lattice approximation of  the following form
\begin{equ}[e:genloop]
 l^\eps=e^\eps e_1^\eps A^\eps (e_4^{\eps})^{-1}e^\eps e_2^\eps B^\eps (e_3^{\eps})^{-1},
 \end{equ}
where, recalling the injection $i_\eps$ in \eqref{e:ij},
\[
i_\eps : (e_1,e_2,e_3,e_4,A,B) \mapsto (e_1^\eps,e_2^\eps,e_3^\eps,e_4^\eps,A^\eps,B^\eps)
\]
and $e^\eps \notin \mbox{Image}(i_\eps)$ (recall Remark \ref{re:app}),
with $|e^\eps|=O(\eps)$.
Here $e^\eps$ can be viewed as the lattice approximation to the crossing point $v$. 
Since the
  edges in $A,B$ do not belong to $\{e_1,e_2,e_3,e_4\}$, and $i_\eps$ is injective,
the
  edges in $A^\eps,B^\eps$ also do not belong to  $\{e_1^\eps,e_2^\eps,e_3^\eps,e_4^\eps\}$.

Remark that besides \eqref{e:genloop} one may also consider lattice approximations where  an edge $e^\eps$ and  its inverse $(e^\eps)^{-1}$ appear, see Remark~\ref{re:ee}. 

Below we take any such  lattice approximation $\{l^\eps\}_{\eps>0}$ with respect to $v$.
 We then have  lattice approximations $\{l^\eps_1\}_{\eps>0}$ and $\{l^\eps_2\}_{\eps>0}$ to $l_1$ and $l_2$ given by 
\begin{equ}[e:genloop12]
l^\eps_1=e^\eps e_1^\eps A^\eps (e_4^{\eps})^{-1},\qquad l_2^\eps=e^\eps e_2^\eps B^\eps (e_3^{\eps})^{-1}.
\end{equ}
By \eqref{e:pos-split}, $l^\eps_1$, $l^\eps_2$ are positive splitting of $l^\eps$.
Here  it is clear that $\{l^\eps_k\}_{\eps>0}$  is a lattice approximation of $l_k$ for each $k=1,2$. Indeed, $l^\eps_k$ consists of a subset of the edges of the graph given by $l^\eps$, 
and correspondingly $l_k$ consists of a subset of the edges of the graph given by $l$, 
so there is a natural injection $i_\eps^{(k)}$ as required in \eqref{e:ij} from the edges in $l_k$ to the edges in $l^\eps_k$, which is just the restriction of $i_\eps$.
Also, since each face enclosed by $l_k^\eps$ is a union of  (typically more than one) faces enclosed
by $l^\eps$ (recalling Fig.~\ref{fig:DHKcase} for example),
and similarly for $l_k$ and $l$,
 there is also a natural bijection $j_\eps^{(k)}$ as required in \eqref{e:ij}
from the faces  enclosed by $l_k$ to the  faces  enclosed by $l^\eps_k$,
which is just applying $j_\eps$ to a union of faces.
 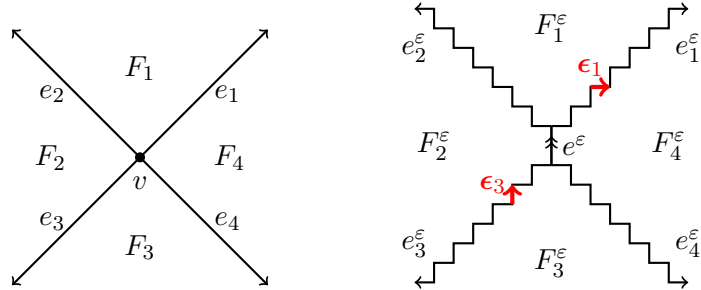
\begin{figure}[h] 
  \centering
  \begin{tikzpicture}[scale=1.7]
\filldraw [black] (0,0) circle (1pt); \node at (0, -0.2) {$v$};
\draw[thick,->] (0,0) -- (1,1) node[midway,right] {$e_1$};
\draw[thick,->] (0,0) -- (1,-1) node[midway,right] {$e_4$};
\draw[thick,->] (0,0) -- (-1,1) node[midway,left] {$e_2$};
\draw[thick,->] (0,0) -- (-1,-1) node[midway,left] {$e_3$};
\node at (0, 0.7) {$F_1$};\node at (0, -0.7) {$F_3$};
\node at (0.7,0) {$F_4$};\node at (-0.7,0) {$F_2$};
\end{tikzpicture}
\qquad
\qquad
\begin{tikzpicture}[scale=2.6]
\draw[thick,midarrow] (0,-0.1) -- (0,0.1)   node[midway,right] {$e^\eps$};
\draw[thick,midarrow2] (0,-0.1) -- (0,0.1);
\draw[thick,->] (0,0.1)
\foreach \i in {1, 2, ..., 13} {-- ++({(0.5-0.5*(-1)^\i)*0.1}, {(0.5+0.5*(-1)^\i)*0.1})};
\node at (0.7,0.5) {$e_1^\eps$};
\draw[ultra thick,red,->] (0.2,0.3) -- (0.3,0.3); 
\node[red] at (0.2,0.4) {$\ee_1$};
\draw[thick,->] (0,-0.1)
\foreach \i in {1, 2, ..., 13} {-- ++({(0.5-0.5*(-1)^\i)*0.1}, {-(0.5+0.5*(-1)^\i)*0.1})};
\node at (0.7,-0.5) {$e_4^\eps$};
\draw[thick,->] (0,0.1)
\foreach \i in {1, 2, ..., 13} {-- ++({-(0.5-0.5*(-1)^\i)*0.1}, {(0.5+0.5*(-1)^\i)*0.1})};
\node at (-0.7,0.5) {$e_2^\eps$};
\draw[thick,->] (0,-0.1)
\foreach \i in {1, 2, ..., 13} {-- ++({-(0.5-0.5*(-1)^\i)*0.1}, {-(0.5+0.5*(-1)^\i)*0.1})};
\node at (-0.7,-0.5) {$e_3^\eps$};
\draw[ultra thick,red,<-] (-0.2,-0.2) -- (-0.2,-0.3); 
\node[red] at (-0.3,-0.2) {$\ee_3$};
\node at (0, 0.6) {$F_1^\eps$};\node at (0, -0.6) {$F_3^\eps$};
\node at (0.6,0) {$F_4^\eps$};\node at (-0.6,0) {$F_2^\eps$};
\end{tikzpicture}
\caption{The left picture illustrates part of a loop $l$ in a neighborhood of the crossing point $v$
as in Fig.~\ref{fig:DHKcase}.
The right picture is a lattice approximation of the loop $l$ with respect to  $v$. The edge $e^{\eps}$ can be viewed as the ``lattice approximation'' to  $v$. The picture also shows the bonds $\ee_1$, $\ee_3$ in red in \eqref{e:triple} (where $\ee$ that is not drawn here is a bond in $e^\eps$) which are used later in Lemmas \ref{lem:ms1} \ref{lem:ms2}.} \label{fig1eps}
\end{figure}

 We also have
 $$W_{l}^\eps=\tr(Q_{e^\eps}a_1^\eps\alpha^\eps (a_4^\eps)^{-1}  Q_{e^\eps} a_2^\eps\beta^\eps (a_3^\eps)^{-1}),$$
 and
 $$W_{l_1}^\eps=\tr(Q_{e^\eps}a_1^\eps\alpha^\eps (a_4^\eps)^{-1}),\qquad W_{l_2}^\eps=\tr(Q_{e^\eps} a_2^\eps\beta^\eps (a_3^\eps)^{-1}),$$
 where we write $a_i^\eps=Q_{e_i^\eps}$ and $\alpha^\eps=Q_{A^\eps}, \beta^\eps=Q_{B^\eps}$. 
  We then write the boundary of each face as 
 \begin{align*}
 	\partial F_1^\eps&=  e_1^\eps A_1^\eps (e_2^\eps)^{-1},
	\quad 
	\p F_2^\eps=e^\eps e_2^\eps A_2^\eps (e_3^\eps)^{-1}, 
 	\\ 
	\p F_3^\eps&=e_3^\eps A_3^\eps  (e_4^\eps)^{-1},
	\quad 
	\p F_4^\eps=e^\eps e_4^\eps A_4^\eps (e_1^\eps)^{-1} 
 \end{align*}
 for suitable sets of edges $A^\eps_1,\cdots,A^\eps_4$.
 Set $t_i(\eps)=|F_i^\eps|$, $i=1,2,3, 4$. Note that $t_i(\eps)/\eps^2$ is a positive integer.

We write $\alpha_i=Q_{A_i^\eps}$. By Remark \ref{re:ax} we choose 
$Q_{\ee}=I$ for all the bonds $\ee$ in $e^\eps$ 
in an axial gauge and apply Driver's formula \eqref{dri:dis1} to get, 
 \begin{equ}[eqc:WLd]
 	\E W_{l}^\eps=\int \tr(a_1\alpha a_4^{-1}  a_2\beta a_3^{-1})
\prod_{i=1}^4 S^\eps_{\frac{t_i(\eps)}{\eps^2}}(a_i\alpha_i a_{i+1}^{-1})
	\, \nu^\eps(\mathbf{b}) \, \dif \mathbf{a} \, \dif \mathbf{b} 
 \end{equ}
 where, similarly as in \eqref{eqc:WL}, 
 $\mathbf{a}=(a_1,a_2,a_3,a_4)\in G^4$, 
 $\mathbf{b}$ denotes the edge variables 
for the edges in $A^\eps$ and $B^\eps$,
  and $\nu^\eps(\mathbf{b})$ is the product of Wilson actions $S^\eps$ 
  arising from the faces other than $F_1^\eps,\cdots, F_4^\eps$, which depends only on the 
  $\mathbf{b}$ variables.
Here we recall Remark~\ref{rem:eps} that 
we hide the dependence on $\eps$ when writing the edge variables $Q_e, a_i, \alpha, \beta$ and $\alpha_i$.

\medskip

With the above lattice loops $\{l^\eps\}_{\eps>0}$, we now describe a general rule to choose the bonds in the discrete  loop equation \eqref{eq:wl}. Recall that a bond always has length $\eps$. To distinguish the two occurrences of $e$ in $l^{\eps}$, 
we introduce a placeholder $\underline{e}^\eps$ which just means  $\underline{e}^\eps = e^\eps$,
and rewrite  \eqref{e:genloop} as 
\begin{equ}[e:genloop-bar]
 l^\eps=e^\eps e_1^\eps A^\eps (e_4^{\eps})^{-1}\underline{e}^\eps e_2^\eps B^\eps (e_3^{\eps})^{-1}.
 \end{equ}

\begin{definition}\label{def:rule}
Let  $\{l^\eps\}_{\eps>0}$ be a lattice approximation of $l$
 with respect to $v$ as above.
We say that $(\ee,\ee',\ee'')$ is a {\it compatible triple of bonds} if either 
\begin{equ}[e:triple]
(\ee,\ee',\ee'') = (\ee,\ee_1,\ee_3)\qquad \mbox{where }\quad
 \ee\in e^\eps, \;\; \ee_1 \in e^\eps_1, \;\; \ee_3\in (e^\eps_3)^{-1}\;,
\end{equ}
or
\begin{equ}[e:triple1]
(\ee,\ee',\ee'') = (\ee,\ee_2,\ee_4)\qquad \mbox{where }\quad
 \ee\in \underline{e}^\eps, \;\; \ee_2 \in e^\eps_2, \;\; \ee_4\in (e^\eps_4)^{-1}\;.
\end{equ}
\end{definition}

Note that in the above definition, $\ee,\ee',\ee''$ always belong to three consecutive edges in $l^\eps$.

\begin{remark}\label{rem:simplest}
The simplest situation is that $e^\eps$ in $l^\eps$ is a bond, so $|e^\eps|=|\ee|=\eps$,
and $\ee_i$ ($i=1,3$) are the bonds adjacent to $\ee$.  This choice is sufficient to prove the convergence to the continuum loop equation. However we formulate our result in a  more general way, allowing a larger class of lattice approximations and more general rule of selecting the triple of bonds.
\end{remark}

We now turn to the proof of our main result. 
We will only consider compatible triples of type \eqref{e:triple}, since the proof for \eqref{e:triple1} is the same.
\footnote{In fact the Wilson loop expectation  does not change if we start $l^\eps$ from $\underline{e}^\eps$ and end it at $(e_4^\eps)^{-1}$. In this case $e_2^\eps$ plays the same role as $e_1^\eps$ and $e_4^\eps$ plays the same role as $e_3^\eps$ in the proof below.}
Below let $(\ee,\ee_1,\ee_3)$ be a compatible triple of bonds as in \eqref{e:triple}.
 In a first step (Proposition~\ref{prop:g}), we consider the impact of deformations along the bond $\ee$ in the edge $e^\eps$.  Due to its central location within the loop, area derivatives of all four faces arise in the continuum limit.  However, there are also correction terms arising from integration by parts.
To cancel these, we consider in a second step (Lemmas~\ref{lem:ms1} + \ref{lem:ms2}) the deformations along the bonds $\ee_i$ in the edges $e_{i}^\eps$, $i=1,3$, respectively.  In this case, only the area derivatives from the two faces sharing the edge $e_{i}$ appear.  In particular, deformations along $\ee_{1}$ produce $\partial_{t_{1}}, \partial_{t_{4}}$, while deformations along $\ee_{3}$ produce $\partial_{t_{2}},\partial_{t_{3}}$.  These also come with corrections in the continuum limit, which then turn out to cancel the correction terms from the loop equation at $\ee$ in $e^\eps$, yielding the continuum loop equation.
 
\medskip
 
Consider the bond  $\ee$   in the edge $e^\eps$ in the center of Fig.~\ref{fig1eps}.
We take this bond $\ee$ as the fixed bond in the discrete master loop equation \eqref{eq:wl}.
 By the definition of splitting, we have 
 \begin{equs}[eq:loop:g]
 	\E	W_l^\eps & =\frac{1}{2\eps^2} \!\!\! \sum_{l'\in \mathbb{D}^-_{\ee}(l)}\!\!\E W_{l'}^\eps
 	-\frac{1}{2\eps^2} \!\!\! \sum_{l'\in \mathbb{D}^+_{\ee}(l)} \!\!\E W_{l'}^\eps
 	-\E W_{l_1}^\eps W_{l_2}^\eps\;.
 \end{equs}

 We define the following quantity which will play an interesting role:
  \begin{equs}\label{def:Im}
 	I_m
 	\eqdef \int &
 	\<\nabla_{b_m} \tr(a_1\alpha a_4^{-1}  a_2\beta a_3^{-1}),
 	\nabla_{b_m}p_{t_m}(a_m\alpha_m a_{m+1}^{-1})\>
 	\\
 	&\prod_{k\in \{1,2,3,4\}\backslash\{m\}} \!\!\!\!
 	p_{t_k}(a_k\alpha_k a_{k+1}^{-1})
 	\,\nu(\mathbf{b})\,\dif \mathbf{a} \,\dif \mathbf{b}, 
 \end{equs}  
for $m=1,\dots,4$,  with $(b_1,b_2,b_3,b_4)\eqdef (a_1,a_3,a_3,a_1)$,
and  $k+1$ is understood as $1$ when $k=4$ as the convention set above.

 \bp\label{prop:g} 
 The discrete master loop equation \eqref{eq:loop:g} converges to 
 \begin{align}\label{e:loop:1}
 	&2(\p_{t_1}-\p_{t_2}+\p_{t_3}-\p_{t_4}) \E W_l+I_1+I_3
 	=\E W_{l_1}W_{l_2}+\E W_l,
 \end{align}
 where $I_1$, $I_3$ are as in \eqref{def:Im}. 
 \ep 
 
 \begin{proof} Using \eqref{e:conv-loop} and \eqref{e:con:loops} and Remark \ref{re:ax}  we know  that $\E	W_l^\eps$ on the LHS of  \eqref{eq:loop:g} and $\E W_{l_1}^\eps W_{l_2}^\eps$ on the RHS of  \eqref{eq:loop:g} converge to $\E W_l$ and $\E W_{l_1}W_{l_2}$, respectively. 
 It remains to prove that the  deformation terms   in \eqref{eq:loop:g} converge to the LHS of \eqref{e:loop:1}. There are $4$ deformed loops which we analyze below.
 	
Consider the loop 
\[
l_{F_4,-}^{\ee} =e'_{\neg} e_1^\eps A^\eps (e_4^\eps)^{-1}e^\eps e_2^\eps B^\eps (e_3^\eps)^{-1}
\]
obtained from the negative deformation of $l$ in \eqref{e:genloop} in the face $F_4$,
 where   $e'_{\neg}$ is an edge obtained by replacing 
the bond $\ee$ in $e^\eps$ by $\ee'$ as shown in Fig.~\ref{fig:deform-e}, where $|\ee'|=3\eps$.

 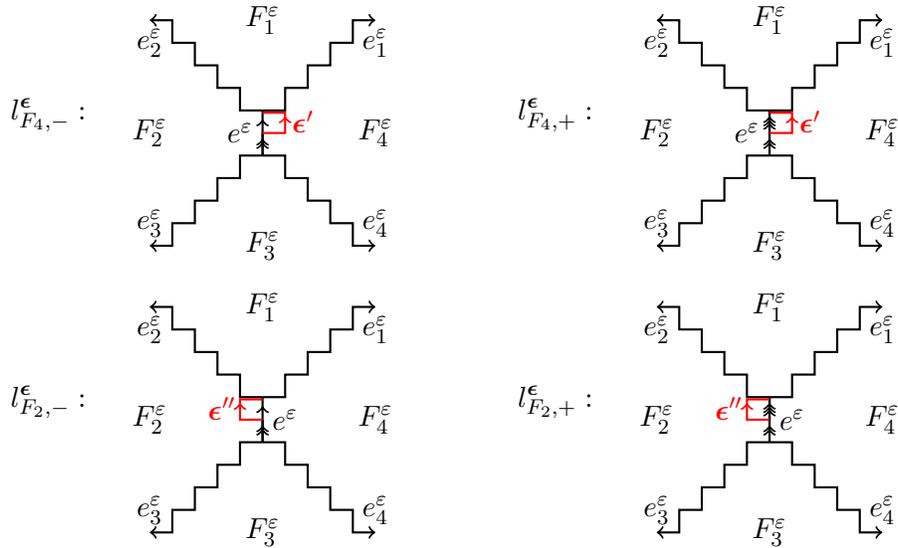
\begin{figure}[h] 
   \centering
\[
l_{F_4,-}^{\ee} :\quad
\begin{tikzpicture}[scale=3,baseline=5]
\draw[thick,midarrow] (0,-0.1) -- (0,0); \draw[thick,midarrow2] (0,-0.1) -- (0,0); 
\draw[thick,midarrow] (0,0) -- (0,0.1) ; \node at (-0.1,0) {$e^\eps$};
\draw[thick,->] (0,0.1)
\foreach \i in {1, 2, ..., 9} {-- ++({(0.5-0.5*(-1)^\i)*0.1}, {(0.5+0.5*(-1)^\i)*0.1})};
\node at (0.5,0.4) {$e_1^\eps$};
\draw[thick,->] (0,-0.1)
\foreach \i in {1, 2, ..., 9} {-- ++({(0.5-0.5*(-1)^\i)*0.1}, {-(0.5+0.5*(-1)^\i)*0.1})};
\node at (0.5,-0.4) {$e_4^\eps$};
\draw[thick,->] (0,0.1)
\foreach \i in {1, 2, ..., 9} {-- ++({-(0.5-0.5*(-1)^\i)*0.1}, {(0.5+0.5*(-1)^\i)*0.1})};
\node at (-0.5,0.4) {$e_2^\eps$};
\draw[thick,->] (0,-0.1)
\foreach \i in {1, 2, ..., 9} {-- ++({-(0.5-0.5*(-1)^\i)*0.1}, {-(0.5+0.5*(-1)^\i)*0.1})};
\node at (-0.5,-0.4) {$e_3^\eps$};
\draw[thick,red,midarrow] (0,0) -- (0.1,0) -- (0.1,0.09) -- (0,0.09);
\node[red] at (0.18,0.05) {$\ee'$};
\node at (0, 0.5) {$F_1^\eps$};\node at (0, -0.5) {$F_3^\eps$};
\node at (0.5,0) {$F_4^\eps$};\node at (-0.5,0) {$F_2^\eps$};
\end{tikzpicture}
\qquad\qquad
l_{F_4,+}^{\ee} :\quad
\begin{tikzpicture}[scale=3,baseline=5]
\draw[thick,midarrow] (0,-0.1) -- (0,0); \draw[thick,midarrow2] (0,-0.1) -- (0,0); 
\draw[thick,midarrow] (0,0) -- (0,0.1); \draw[thick,midarrow1] (0,0) -- (0,0.1); \draw[thick,midarrow2] (0,0) -- (0,0.1); 
\node at (-0.1,0) {$e^\eps$};
\draw[thick,->] (0,0.1)
\foreach \i in {1, 2, ..., 9} {-- ++({(0.5-0.5*(-1)^\i)*0.1}, {(0.5+0.5*(-1)^\i)*0.1})};
\node at (0.5,0.4) {$e_1^\eps$};
\draw[thick,->] (0,-0.1)
\foreach \i in {1, 2, ..., 9} {-- ++({(0.5-0.5*(-1)^\i)*0.1}, {-(0.5+0.5*(-1)^\i)*0.1})};
\node at (0.5,-0.4) {$e_4^\eps$};
\draw[thick,->] (0,0.1)
\foreach \i in {1, 2, ..., 9} {-- ++({-(0.5-0.5*(-1)^\i)*0.1}, {(0.5+0.5*(-1)^\i)*0.1})};
\node at (-0.5,0.4) {$e_2^\eps$};
\draw[thick,->] (0,-0.1)
\foreach \i in {1, 2, ..., 9} {-- ++({-(0.5-0.5*(-1)^\i)*0.1}, {-(0.5+0.5*(-1)^\i)*0.1})};
\node at (-0.5,-0.4) {$e_3^\eps$};
\draw[thick,red,midarrow] (0,0) -- (0.1,0) -- (0.1,0.09) -- (0,0.09);
\node[red] at (0.18,0.05) {$\ee'$};
\node at (0, 0.5) {$F_1^\eps$};\node at (0, -0.5) {$F_3^\eps$};
\node at (0.5,0) {$F_4^\eps$};\node at (-0.5,0) {$F_2^\eps$};
\end{tikzpicture}
\]
\[
l_{F_2,-}^{\ee} :\quad
\begin{tikzpicture}[scale=3,baseline=5]
\draw[thick,midarrow] (0,-0.1) -- (0,0); \draw[thick,midarrow2] (0,-0.1) -- (0,0); 
\draw[thick,midarrow] (0,0) -- (0,0.1) ; \node at (0.1,0) {$e^\eps$};
\draw[thick,->] (0,0.1)
\foreach \i in {1, 2, ..., 9} {-- ++({(0.5-0.5*(-1)^\i)*0.1}, {(0.5+0.5*(-1)^\i)*0.1})};
\node at (0.5,0.4) {$e_1^\eps$};
\draw[thick,->] (0,-0.1)
\foreach \i in {1, 2, ..., 9} {-- ++({(0.5-0.5*(-1)^\i)*0.1}, {-(0.5+0.5*(-1)^\i)*0.1})};
\node at (0.5,-0.4) {$e_4^\eps$};
\draw[thick,->] (0,0.1)
\foreach \i in {1, 2, ..., 9} {-- ++({-(0.5-0.5*(-1)^\i)*0.1}, {(0.5+0.5*(-1)^\i)*0.1})};
\node at (-0.5,0.4) {$e_2^\eps$};
\draw[thick,->] (0,-0.1)
\foreach \i in {1, 2, ..., 9} {-- ++({-(0.5-0.5*(-1)^\i)*0.1}, {-(0.5+0.5*(-1)^\i)*0.1})};
\node at (-0.5,-0.4) {$e_3^\eps$};
\draw[thick,red,midarrow] (0,0) -- (-0.1,0) -- (-0.1,0.09) -- (0,0.09);
\node[red] at (-0.18,0.05) {$\ee''$};
\node at (0, 0.5) {$F_1^\eps$};\node at (0, -0.5) {$F_3^\eps$};
\node at (0.5,0) {$F_4^\eps$};\node at (-0.5,0) {$F_2^\eps$};
\end{tikzpicture}
\qquad\qquad
l_{F_2,+}^{\ee} :\quad
\begin{tikzpicture}[scale=3,baseline=5]
\draw[thick,midarrow] (0,-0.1) -- (0,0); \draw[thick,midarrow2] (0,-0.1) -- (0,0); 
\draw[thick,midarrow] (0,0) -- (0,0.1); \draw[thick,midarrow1] (0,0) -- (0,0.1); \draw[thick,midarrow2] (0,0) -- (0,0.1); 
\node at (0.1,0) {$e^\eps$};
\draw[thick,->] (0,0.1)
\foreach \i in {1, 2, ..., 9} {-- ++({(0.5-0.5*(-1)^\i)*0.1}, {(0.5+0.5*(-1)^\i)*0.1})};
\node at (0.5,0.4) {$e_1^\eps$};
\draw[thick,->] (0,-0.1)
\foreach \i in {1, 2, ..., 9} {-- ++({(0.5-0.5*(-1)^\i)*0.1}, {-(0.5+0.5*(-1)^\i)*0.1})};
\node at (0.5,-0.4) {$e_4^\eps$};
\draw[thick,->] (0,0.1)
\foreach \i in {1, 2, ..., 9} {-- ++({-(0.5-0.5*(-1)^\i)*0.1}, {(0.5+0.5*(-1)^\i)*0.1})};
\node at (-0.5,0.4) {$e_2^\eps$};
\draw[thick,->] (0,-0.1)
\foreach \i in {1, 2, ..., 9} {-- ++({-(0.5-0.5*(-1)^\i)*0.1}, {-(0.5+0.5*(-1)^\i)*0.1})};
\node at (-0.5,-0.4) {$e_3^\eps$};
\draw[thick,red,midarrow] (0,0) -- (-0.1,0) -- (-0.1,0.09) -- (0,0.09);
\node[red] at (-0.18,0.05) {$\ee''$};
\node at (0, 0.5) {$F_1^\eps$};\node at (0, -0.5) {$F_3^\eps$};
\node at (0.5,0) {$F_4^\eps$};\node at (-0.5,0) {$F_2^\eps$};
\end{tikzpicture}
\]

\caption{Loops obtained by deforming a bond $\ee$ of $e^\eps$.
Here $|e^\eps|=2\eps$  and $|\ee|=\eps$ ($e^\eps$ has two bonds where $\ee$ is the upper one). The number of arrows on an edge or bond indicates how many times it shows up in the loop.}
\label{fig:deform-e}
\end{figure}
 	
By Driver's formula \eqref{dri:dis1}  
(recall that every bond variable on $e^\eps$ is identity)
 \begin{align*}
 \E W_{l_{F_4,-}^{\ee}}^\eps
 =\int &\tr(Q_{\ee'}a_1\alpha a_4^{-1}  a_2\beta a_3^{-1})
\prod_{i=1,2,3} S^\eps_{\frac{t_i(\eps)}{\eps^2}}(a_i\alpha_i a_{i+1}^{-1})
 \\
 &\times S^\eps_{\frac{t_4(\eps)-\eps^2}{\eps^2}}(Q_{\ee'}a_1\alpha_4^{-1} a_4^{-1})
 S^\eps(Q_{\ee'}^{-1})
\,\nu^\eps(\mathbf{b}) \,\dif \mathbf{a}\,\dif \mathbf{b}\,\dif Q_{\ee'}\;.
 \end{align*}
%
 Similarly for the loop 
 \[
 l_{F_4,+}^{\ee}=e'_{\pos}e_1^\eps A^\eps (e_4^\eps)^{-1}e^\eps e_2^\eps B^\eps (e_3^\eps)^{-1}
 \]
  from the positive deformation 
 of $l$ in \eqref{e:genloop} in the face $F_4$, where $e'_{\pos}$ is an edge obtained by replacing $\ee$ in $e^\eps $ by $\ee(\ee')^{-1}\ee$ with $\ee'$  the same as above,  we have 
 \begin{align*}
 \E W_{l_{F_4,+}^{\ee}}^\eps
 =\int &\tr(Q_{\ee'}^{-1}a_1\alpha a_4^{-1}  a_2\beta a_3^{-1})
\prod_{i=1,2,3} S^\eps_{\frac{t_i(\eps)}{\eps^2}}(a_i\alpha_i a_{i+1}^{-1})
 \\
 &\times S^\eps_{\frac{t_4(\eps)-\eps^2}{\eps^2}}(Q_{\ee'}a_1\alpha_4^{-1} a_4^{-1})S^\eps(Q_{\ee'}^{-1})
\,\nu^\eps(\mathbf{b}) \,\dif \mathbf{a}\,\dif \mathbf{b}\,\dif Q_{\ee'} \;.
 \end{align*}
We then have
 \begin{equs}\label{EE2}
 {}&\E W_{l_{F_4,-}^{\ee}}^\eps-\E W_{l_{F_4,+}^{\ee}}^\eps
 \\
 =\int &\tr\Big((Q_{\ee'}-Q_{\ee'}^{-1})a_1\alpha a_4^{-1}  a_2\beta a_3^{-1}\Big)
\prod_{i=1,2,3} S^\eps_{\frac{t_i(\eps)}{\eps^2}}(a_i\alpha_i a_{i+1}^{-1})
\\
&\times S^\eps_{\frac{t_4(\eps)-\eps^2}{\eps^2}}(Q_{\ee'}a_1\alpha_4^{-1} a_4^{-1})
S^\eps(Q_{\ee'}^{-1})
\,\nu^\eps(\mathbf{b})\,\dif \mathbf{a}\,\dif \mathbf{b}\,\dif Q_{\ee'} \;.
\end{equs}
 	
Using Lemma~\ref{lem:J1} for the integration in $Q_{\ee'}$, then applying \eqref{eq:a-cons} to replace $S^\eps_{\frac{t_i(\eps)}{\eps^2}}(\cdot)$ with $p_{\frac{t_{i}}N}(\cdot)$ for $i=1,2,3$, we find 
 \begin{equs}
 		{}&\lim_{\eps\to 0}\frac1{2\eps^2}\Big(\E W_{l_{F_4,-}^{\ee}}^\eps-\E W_{l_{F_4,+}^{\ee}}^\eps\Big)
 		\\
		&= \sum_{j=1}^{d(\mfg)}
		\int \cL_j\tr(a_1\alpha a_4^{-1}  a_2\beta a_3^{-1})
		\cL_j p_{t_4}( a_1\alpha_4^{-1} a_4^{-1})
\!\!\prod_{i=1,2,3}\!\! p_{t_i}(a_i\alpha_i a_{i+1}^{-1})
		\, \nu(\mathbf{b})\,\dif \mathbf{a}\,\dif \mathbf{b}
\\
&= I_4\;, \label{e:F4diff}
 	\end{equs}
where $I_4$ is as in \eqref{def:Im} and we use $p_t$ is a class function.

 We then use integration by parts in $a_{1}$ in the following form: for smooth functions $f, g$ and $h$ on $G$ 
 \begin{align}
 	&\int \<\nabla_{a_1} f,\nabla_{a_1} g\> h \, \dif a_1\no
 	\\
	=&-	\int  f \Delta_{a_1}g h \,\dif a_1
	-\int \<\nabla_{a_1} h,\nabla_{a_1} g\>f \dif a_1\label{eq:IBPa}
 	\\
	=& -\int  f \Delta_{a_1}g h \,\dif a_1+\int \Delta_{a_1} h g f \,\dif a_1
	+\int \<\nabla_{a_1} f,\nabla_{a_1} h\> g\,\dif a_1.\no
 \end{align}
 Substituting this formula \eqref{eq:IBPa}  
 into $I_4$ obtained in \eqref{e:F4diff} with 
 $$
 f=\tr(a_1\alpha a_4^{-1}  a_2\beta a_3^{-1}),\quad 
 g=p_{t_4}(  a_4\alpha_4 a_1^{-1}),\quad
 h=p_{t_1}(a_1\alpha_1 a_2^{-1}),
 $$
 we obtain
 \begin{equs}[e:F4d]
\eqref{e:F4diff}
 &=\int \tr(a_1\alpha a_4^{-1}  a_2\beta a_3^{-1})\prod_{i=2,3}\!\! p_{t_i}(a_i\alpha_i a_{i+1}^{-1})
 	\\&\qquad\times	\Big(\Delta_{a_1}p_{t_1}(a_1\alpha_1 a_2^{-1})p_{t_4}(  a_4\alpha_4 a_1^{-1})
 	\\&\qquad-p_{t_1}(a_1\alpha_1 a_2^{-1})\Delta_{a_1}p_{t_4}(  a_4\alpha_4 a_1^{-1})\Big)
 	\,\nu(\mathbf{b})\,\dif \mathbf{a}\,\dif \mathbf{b}
 	\; +\; I_1
 	\\&= 2(\p_{t_1}-\p_{t_4})\E W_l+I_1\,,
 \end{equs}
 where we used that $p_t$ is the heat kernel of $\frac12\Delta$ in the last step.

Next, consider 
 the loop   obtained from the negative deformation in the face $F_2$:
 \[
 l_{F_2,-}^{\ee}= e''_{\neg} e_1^\eps A^\eps (e_4^\eps)^{-1}e^\eps e_2^\eps B^\eps (e_3^\eps)^{-1}
 \]
where the edge  $e''_{\neg}$ is obtained by replacing the bond $\ee$ in $e^\eps $ by $\ee''$ as shown in Fig.~\ref{fig:deform-e}, with $|\ee''|=3\eps$. 
Also, consider the loop   obtained from the positive deformation in the face $F_2$:
 \[
 l_{F_2,+}^{\ee}= e''_{\pos} e_1^\eps A^\eps (e_4^\eps)^{-1}e^\eps e_2^\eps B^\eps(e_3^\eps)^{-1}
 \]
where the edge $e''_{\pos} $ is obtained by replacing the bond  $\ee$ in $e^\eps$ by $\ee(\ee'')^{-1}\ee$.

 We then have, similarly as above,
 \begin{equs}\label{EE3}
 	&\frac1{2\eps^2}\Big(\E W_{l_{F_2,-}^{\ee}}^\eps-\E W_{l_{F_2,+}^{\ee}}^\eps\Big)
 	\\&=\frac1{2\eps^2} \!\!\int \tr\Big((Q_{\ee''}-Q_{\ee''}^{-1})a_1\alpha a_4^{-1}  a_2\beta a_3^{-1}\Big)S^\eps_{\frac{t_2(\eps)-\eps^2}{\eps^2}}(Q_{\ee''}a_2\alpha_2 a_3^{-1})
 	\\&\qquad\times  \prod_{i=1,3,4}S^\eps_{\frac{t_i(\eps)}{\eps^2}}(a_i\alpha_i a_{i+1}^{-1})S^\eps(Q_{\ee''}^{-1})
	\,\nu^\eps(\mathbf{b})\,\dif \mathbf{a}\,\dif \mathbf{b}\,\dif Q_{\ee''}\;.
 \end{equs}
 We apply Lemma \ref{lem:J1} for the integral in $Q_{\ee''}$, then \eqref{eq:a-cons}, to obtain
 \begin{equs}
 	{}&\lim_{\eps\to 0}\frac1{2\eps^2}\Big(\E W_{l_{F_2,-}^{\ee}}^\eps-\E W_{l_{F_2,+}^{\ee}}^\eps\Big)
 	\\&= \sum_j\int \cL_j\tr(a_1\alpha a_4^{-1}  a_2\beta a_3^{-1})
	\cL_jp_{t_2}(a_2\alpha_2 a_3^{-1})
\!\!\prod_{i=1,3,4}\!\! p_{t_i}(a_i\alpha_i a_{i+1}^{-1})
	 \,\nu(\mathbf{b})\,\dif \mathbf{a}\,\dif \mathbf{b}
	 \\
	 &=I_2\,, \label{e:F2diff}
 \end{equs}
 where $I_2$ is as in \eqref{def:Im}. 
In the last step we  use the fact that $\mathcal L_{X,b} f(ba^{-1}) = -\mathcal L_{X,a} f(ba^{-1})$ for class function $f$ and $X\in\mfg, a,b\in G$, which is why we obtain two gradients in $a_3$ 
as required in the definition \eqref{def:Im} of $I_2$.

 Hence, we use integration by parts \eqref{eq:IBPa} 
 with 
 $$f=\tr(a_1\alpha a_4^{-1}  a_2\beta a_3^{-1}),\quad 
 g=p_{t_2}(a_2\alpha_2 a_3^{-1}),\quad
 h=p_{t_3}(a_3\alpha_3 a_4^{-1}),$$
 to obtain 
 \begin{equs}[e:F2d]
	\eqref{e:F2diff} 
	&=\int 
	\tr(a_1\alpha a_4^{-1}  a_2\beta a_3^{-1})
	\prod_{i=1,4}p_{t_i}(a_i\alpha_i a_{i+1}^{-1})
	\\&\qquad\times	\Big(p_{t_2}(a_2\alpha_2 a_3^{-1})
\Delta_{a_3}	p_{t_3}(a_3\alpha_3 a_4^{-1})
	\\&\qquad-\Delta_{a_3}p_{t_2}(a_2\alpha_2 a_3^{-1})
		p_{t_3}(a_3\alpha_3 a_4^{-1})\Big)
	\,\nu(\mathbf{b})\,\dif \mathbf{a}\,\dif \mathbf{b}
	\; + \;I_3
	\\&= 2(\p_{t_3}-\p_{t_2})\E W_l+I_3\;.
\end{equs}
 Combining  \eqref{eq:loop:g}, \eqref{e:F4d}, \eqref{e:F2d} the result \eqref{e:loop:1} follows. 
 \end{proof}

Now consider the bonds $\ee_i$  in the edge $e_i^\eps$ 
where $i=1,3$. See Fig.~\ref{fig1eps}. 
 We then choose the fixed bond $\ee_1$ in the discrete master loop equation \eqref{eq:wl} 
 so that we  have 
 \begin{equs}[eq:loop:g1]
 \E	W_l^\eps & =\frac{1}{2\eps^2} \!\!\! \sum_{l'\in \mathbb{D}^-_{\ee_1}(l)}\!\!\E W_{l'}^\eps
 -\frac{1}{2\eps^2} \!\!\! \sum_{l'\in \mathbb{D}^+_{\ee_1}(l)} \!\!\E W_{l'}^\eps\;.
 \end{equs}
 
 \bl\label{lem:ms1} 
 The discrete master loop equation \eqref{eq:loop:g1} converges to 
 \begin{equ}\label{loop:g:2}
 2(\p_{t_1}-\p_{t_4})\E W_l+2I_1=\E W_l. 
 \end{equ}
 \el 
 \begin{proof}
 We  consider the loop 
 \[
 l_{F_4,-}^{\ee_1}=e^\eps e_{1,\neg}'A^\eps (e_4^\eps)^{-1}e^\eps e_2^\eps B^\eps (e_3^\eps)^{-1}
 \]
  obtained from the negative deformation of \eqref{e:genloop} at $\ee_1$ in the face $F_4$, where $e_{1,\neg}'$ is given by replacing $\ee_1^\eps$ in $e_1^\eps$ by $\ee_1'$ as shown in Fig.~\ref{fig:deform-ee1}, which consists $3$ bonds. As stated in  Remark \ref{re:ax}, we select the edge variables to be $I$ on the tree, determined by the path extending from $e^\eps$ to $\ee_1$, exclusive of the bond $\ee_1$. Additionally, for the sake of brevity, we overload the notation $a_1$ and write $a_1=Q_{e_1\backslash \ee_1}$. We have 
 \begin{align*}
 	\E W_{l_{F_4,-}^{\ee_1}}^\eps=\int &\tr(Q_{\ee_1'}a_1\alpha a_4^{-1}  a_2\beta a_3^{-1})S^\eps_{\frac{t_1(\eps)+\eps^2}{\eps^2}}(Q_{\ee_1'}a_1\alpha_1 a_2^{-1})\prod_{i=2}^3S^\eps_{\frac{t_i(\eps)}{\eps^2}}(a_i\alpha_i a_{i+1}^{-1})
 	\\&\times S^\eps_{\frac{t_4(\eps)-\eps^2}{\eps^2}}(Q_{\ee_1'}a_1\alpha_4^{-1} a_4^{-1})
	\,\nu^\eps(\mathbf{b})\,\dif \mathbf{a}\,\dif \mathbf{b}\,\dif Q_{\ee_1'}.
 \end{align*}
 
 \begin{figure}[h] 
\[
l_{F_4,-}^{\ee_1} :
\qquad
\begin{tikzpicture}[scale=3.5,baseline=5]
\draw[thick,midarrow,darkgreen] (0,-0.1) -- (0,0.1);
\draw[thick,midarrow1,darkgreen] (0,-0.1) -- (0,0.1);
\node at (-0.1,0) {$e^\eps$};
\draw[thick,darkgreen] (0,0.1)
\foreach \i in {1, 2, ..., 4} {-- ++({(0.5-0.5*(-1)^\i)*0.1}, {(0.5+0.5*(-1)^\i)*0.1})};
\draw[dotted,thick] (0.2,0.3) -- (0.3,0.3);
\draw[thick,->] (0.3,0.3)
\foreach \i in {1, 2, ..., 4} {-- ++({(0.5+0.5*(-1)^\i)*0.1}, {(0.5-0.5*(-1)^\i)*0.1})};
\node at (0.5,0.4) {$e_1^\eps$};
\draw[thick,->] (0,-0.1)
\foreach \i in {1, 2, ..., 9} {-- ++({(0.5-0.5*(-1)^\i)*0.1}, {-(0.5+0.5*(-1)^\i)*0.1})};
\node at (0.5,-0.4) {$e_4^\eps$};
\draw[thick,->] (0,0.1)
\foreach \i in {1, 2, ..., 9} {-- ++({-(0.5-0.5*(-1)^\i)*0.1}, {(0.5+0.5*(-1)^\i)*0.1})};
\node at (-0.5,0.4) {$e_2^\eps$};
\draw[thick,->] (0,-0.1)
\foreach \i in {1, 2, ..., 9} {-- ++({-(0.5-0.5*(-1)^\i)*0.1}, {-(0.5+0.5*(-1)^\i)*0.1})};
\node at (-0.5,-0.4) {$e_3^\eps$};
\draw[thick,red,midarrow] (0.21,0.3) -- (0.21,0.2) -- (0.3,0.2) -- (0.3,0.3);
\node[red] at (0.28,0.12) {$\ee_1'$};
\node at (0, 0.5) {$F_1^\eps$};\node at (0, -0.5) {$F_3^\eps$};
\node at (0.5,0) {$F_4^\eps$};\node at (-0.5,0) {$F_2^\eps$};
\end{tikzpicture}
\]
\caption{Picture for $l_{F_4,-}^{\ee_1}$. The green part is the tree used in the axial gauge fixing. The dashed line is where $\ee_1$ was and it is replaced by $\ee_1'$ (obviously a backtracking is formed but we just erase it as usual). 
The variable $a_1$ is the holonomy along the black part of $e_1^\eps$.
The pictures for $l_{F_4,+}^{\ee_1}$, $l_{F_1,-}^{\ee_1}$ and $l_{F_1,+}^{\ee_1}$ are similar and we do not draw all of them here.}
\label{fig:deform-ee1}
\end{figure}
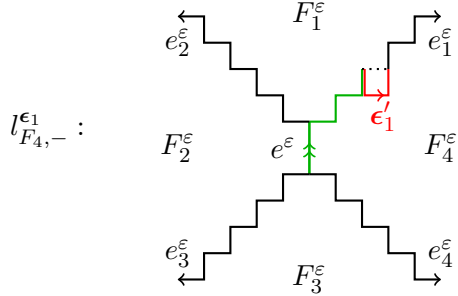

 For the loop 
 \[
 l_{F_4,+}^{\ee_1}=e^\eps e_{1,\pos}' A^\eps (e_4^\eps)^{-1}e^\eps e_2^\eps B^\eps (e_3^\eps)^{-1}
 \]
  obtained from the positive deformation  of \eqref{e:genloop} at $\ee_1$ in the face $F_4$, where $e_{1,\pos}'$ is given by replacing $\ee_1$ in $e_1^\eps$ by $\ee_1(\ee_1')^{-1}\ee_1$ with $\ee_1'$ the same as above, we write
 \begin{align*}
 	\E W_{l_{F_4,+}^{\ee_1}}^\eps=\int &\tr(Q_{\ee_1}Q_{\ee_1'}^{-1}Q_{\ee_1}a_1\alpha a_4^{-1} a_2\beta a_3^{-1})S^\eps_{\frac{t_1(\eps)}{\eps^2}}(Q_{\ee_1}a_1\alpha_1 a_2^{-1})
 	\\&\times \prod_{i=2}^3S^\eps_{\frac{t_i(\eps)}{\eps^2}}(a_i\alpha_i a_{i+1}^{-1})S^\eps_{\frac{t_4(\eps)-\eps^2}{\eps^2}}(Q_{\ee_1'}a_1\alpha_4^{-1} a_4^{-1})
 	\\&\times S^\eps(Q_{\ee_1}^{-1}Q_{\ee_1'})
	\,\nu^\eps(\mathbf{b})\,\dif \mathbf{a}\,\dif \mathbf{b}\,\dif Q_{\ee_1'}\dif Q_{\ee_1}.
 \end{align*}
Since 
 $$S^\eps_{\frac{t_1(\eps)+\eps^2}{\eps^2}}(Q_{\ee_1'}a_1\alpha_1 a_2^{-1})=\int S^\eps_{\frac{t_1(\eps)}{\eps^2}}(Q_{\ee_1} a_1\alpha_1 a_2^{-1})S^\eps(Q_{\ee_1}^{-1}Q_{\ee_1'})\dif Q_{\ee_1},$$
we have
 \begin{equs}[EE1]
 &\frac1{2\eps^2}(\E W_{l_{F_4,-}^{\ee_1}}^\eps	-\E W_{l_{F_4,+}^{\ee_1}}^\eps)
 \\=\frac1{2\eps^2}\int &\tr\Big((Q_{\ee_1'}Q_{\ee_1}^{-1}-Q_{\ee_1}Q_{\ee_1'}^{-1})Q_{\ee_1}a_1\alpha a_4^{-1} a_2\beta a_3^{-1}\Big)
 \\&\times S^\eps_{\frac{t_1(\eps)}{\eps^2}}(Q_{\ee_1}a_1\alpha_1 a_2^{-1})
  \prod_{i=2}^3S^\eps_{\frac{t_i(\eps)}{\eps^2}}(a_i\alpha_i a_{i+1}^{-1})
\\&\times  S^\eps_{\frac{t_4(\eps)-\eps^2}{\eps^2}}(Q_{\ee_1'}a_1\alpha_4^{-1} a_4^{-1})
 S^\eps(Q_{\ee_1}^{-1}Q_{\ee_1'})
 \,\nu^\eps(\mathbf{b})\,\dif \mathbf{a}\,\dif \mathbf{b}\,\dif Q_{\ee_1'}\dif Q_{\ee_1}.
 \end{equs}
 Changing variables in $Q_{\ee_1'}$ so that $Q=Q_{\ee_1'}Q_{\ee_1}^{-1}$ and changing $Q_{\ee_1}a_1$ to $a_1$, we find that \eqref{EE1} equals \eqref{EE2} 
 from the proof of Proposition \ref{prop:g}. 
 Hence by \eqref{e:F4d} we have
 \begin{equ}[e:F41d]
 \lim_{\eps\to 0}
 	\frac1{2\eps^2}\Big(\E W_{l_{F_4,-}^{\ee_1}}^\eps-\E W_{l_{F_4,+}^{\ee_1}}^\eps\Big)
	=
	 2(\p_{t_1}-\p_{t_4})\E W_l+I_1\;.
 \end{equ}

We also consider
 \[
 l_{F_1,-}^{\ee_1}=e^\eps e_{1,\neg}'' A^\eps (e_4^\eps)^{-1}e^\eps e_2^\eps B^\eps (e_3^\eps)^{-1}
 \]
  from the negative deformation  of \eqref{e:genloop} at $\ee_1$ in the face $F_1$,  
  where $e_{1,\neg}''$ is given by replacing $\ee_1$ in $e_1^\eps$ by  $\ee_1''$, which consists $3$ bonds. 
We also have the loop 
 \[
 l_{F_1,+}^{\ee_1}=e^\eps e_{1,\pos}''A^\eps (e_4^\eps)^{-1}e^\eps e_2^\eps B^\eps (e_3^\eps)^{-1}
 \]
  from the positive deformation  of \eqref{e:genloop} at $\ee_1$ in the face $F_1$, where $e_{1,\pos}''$  is given by replacing $\ee_1$ in $e_1^\eps$ by  $\ee_1(\ee_1'')^{-1}\ee_1$, which consists $3$ bonds. 
So we have
 \begin{align*}
 	&\E W_{l_{F_1,-}^{\ee_1}}^\eps-\E W_{l_{F_1,+}^{\ee_1}}^\eps
 	\\=&\int\tr\Big((Q_{\ee_1''}Q_{\ee_1}^{-1}-Q_{\ee_1}Q_{\ee_1''}^{-1})Q_{\ee_1}a_1\alpha a_4^{-1}  a_2\beta a_3^{-1}\Big)S^\eps_{\frac{t_1(\eps)-\eps^2}{\eps^2}}(Q_{\ee_1''}a_1\alpha_1 a_2^{-1})
 	\\&\times \prod_{i=2}^3S_{\frac{t_i(\eps)}{\eps^2}}(a_i\alpha_i a_{i+1}^{-1})S_{\frac{t_4(\eps)}{\eps^2}}(a_4\alpha_4 (Q_{\ee_1}a_{1})^{-1})
 	\\&\times S^\eps(Q_{\ee_1}^{-1}Q_{\ee_1''})
 	 \,\nu^\eps(\mathbf{b})\,\dif \mathbf{a}\,\dif \mathbf{b}\,\dif Q_{\ee_1''}\dif Q_{\ee_1}\;.
 \end{align*}
 We then use a change of variable $Q_{\ee_1''}Q_{\ee_1}^{-1}\to Q$ and $Q_{\ee_1}a_1\to a_1$ and $S^\eps(Q)=S^\eps(Q^{-1})$ to have
 \begin{align*}
&\E W_{l_{F_1,-}^{\ee_1}}^\eps-\E W_{l_{F_1,+}^{\ee_1}}^\eps
 	\\=&\int\tr\Big((Q-Q^{-1})a_1\alpha a_4^{-1}  a_2\beta a_3^{-1}\Big)
	S^\eps_{\frac{t_1(\eps)-\eps^2}{\eps^2}}(Qa_1\alpha_1 a_2^{-1})
 	\\&\times 
\prod_{i=2}^4S_{\frac{t_i(\eps)}{\eps^2}}(a_i\alpha_i a_{i+1}^{-1})
	S^\eps(Q)\,\nu^\eps(\mathbf{b})\,\dif \mathbf{a}\,\dif \mathbf{b}\,\dif Q.
 \end{align*}
 By Lemma \ref{lem:J1} and \eqref{eq:a-cons} as before we have
 \begin{equs}[e:F11d]
& \lim_{\eps\to 0} 
 	\frac1{2\eps^2}\Big(\E W_{l_{F_1,-}^{\ee_1}}^\eps-\E W_{l_{F_1,+}^{\ee_1}}^\eps\Big)
 	\\ 
	&= \sum_j\int\cL_j\tr(a_1\alpha a_4^{-1}  a_2\beta a_3^{-1})\cL_jp_{t_1}(a_1\alpha_1 a_2^{-1})
 		\\&\qquad\times \prod_{i=2}^4p_{t_i}(a_i\alpha_i a_{i+1}^{-1})\,
		\nu(\mathbf{b})\,\dif \mathbf{a}\,\dif \mathbf{b}=I_{1}.
 \end{equs}
 	Combining \eqref{eq:loop:g1}, \eqref{e:F41d}, \eqref{e:F11d}, the result \eqref{loop:g:2} follows.
 \end{proof}
 

 Similarly as above 
 we choose the fixed bond $\ee_3$ in the discrete master loop equation \eqref{eq:wl},  and  we have  
 \begin{equs}[eq:loop:g2]
 \E	W_l^\eps & =\frac{1}{2\eps^2} \!\!\! \sum_{l'\in \mathbb{D}^-_{\ee_3}(l)}\!\!\E W_{l'}^\eps
 -\frac{1}{2\eps^2} \!\!\! \sum_{l'\in \mathbb{D}^+_{\ee_3}(l)} \!\!\E W_{l'}^\eps.
 \end{equs}
 
 Hence, we can also obtain the following result. 
 
 \bl\label{lem:ms2} The discrete master loop equation \eqref{eq:loop:g2} converges to 
 \begin{align}\label{loop:g:3}
 2(\p_{t_2}-\p_{t_3})\E W_l+2I_3=\E W_l. 
 \end{align}
 \el 
 \begin{proof}
We consider the loop 
 \[
 l_{F_2,-}^{\ee_3}=e^\eps e_1^\eps A^\eps (e_4^\eps)^{-1}e^\eps e_2^\eps B^\eps  e_{3,\neg}'
 \]
  from the negative deformation of \eqref{e:genloop} at $\ee_3$ in the face $F_2$, where $e_{3,\neg}'$ is given by replacing $\ee_3$ in $e_3^\eps$ by $\ee_3'$, which consists $3$ bonds. 
   As in Remark \ref{re:ax}  we select the edge variables to be $I$ on the tree, determined by the path extending from $e^\eps$ to $\ee_3$, exclusive of the edge $\ee_3$. Additionally, for the sake of brevity, we overload the notation $a_3$ and denote $a_3=Q_{e_3\backslash \ee_3}$. We  have
 \begin{align*}
 	\E W_{l_{F_2,-}^{\ee_3}}^\eps=\int &\tr(a_1\alpha a_4^{-1}  a_2\beta a_3^{-1}Q_{\ee_3'}^{-1})\prod_{i=1,4}S^\eps_{\frac{t_i(\eps)}{\eps^2}}(a_i\alpha_i a_{i+1}^{-1})
 	\\&\times S^\eps_{\frac{t_2(\eps)-\eps^2}{\eps^2}}(a_2\alpha_2 a_3^{-1}Q_{\ee_3'}^{-1}) \\&\times S^\eps_{\frac{t_3(\eps)+\eps^2}{\eps^2}}(Q_{\ee_3'}a_3\alpha_3 a_4^{-1})
	\,\dif \mathbf{a}\,\nu^\eps(\mathbf{b})\,\dif \mathbf{b}\,\dif Q_{\ee_3'}.
 \end{align*}
We also have the loop
 \[
 l_{F_2,+}^{e_3}=e^\eps e_1^\eps A^\eps (e_4^\eps)^{-1}e^\eps e_2^\eps B^\eps e_{3,\pos}'
 \]
  from the positive deformation of \eqref{e:genloop} at $\ee_3$ in the face $F_2$ with 
  $e_{3,\pos}'$ given by replacing $\ee_3$ in $e_3^\eps$ by $\ee_3(\ee_3')^{-1}\ee_3$. We have
 \begin{align*}
 	\E W_{l_{F_2,+}^{\ee_3}}^\eps=\int &\tr(a_1\alpha a_4^{-1}  a_2\beta a_3^{-1}Q_{\ee_3}^{-1}Q_{\ee_3'}Q_{\ee_3}^{-1})\prod_{i=1,4}S^\eps_{\frac{t_i(\eps)}{\eps^2}}(a_i\alpha_i a_{i+1}^{-1})
 	\\&\times S^\eps_{\frac{t_2(\eps)-\eps^2}{\eps^2}}(a_2\alpha_2 a_3^{-1}Q_{\ee_3'}^{-1})
 	 S^\eps_{\frac{t_3(\eps)}{\eps^2}}(Q_{\ee_3}a_3\alpha_3 a_4^{-1})
	\\&\times S^\eps(Q_{\ee_3}^{-1}Q_{\ee_3'})\,\nu^\eps(\mathbf{b})\,\dif \mathbf{a}\,\dif \mathbf{b}.
 \end{align*}
 Using change of variable $Q_{\ee_3}Q_{\ee_3'}^{-1}\to Q$, $Q_{\ee_3}a_3\to a_3$ and recalling \eqref{EE3}, we find that
 \begin{align*}
 	\frac1{2\eps^2}\Big(\E W_{l_{F_2,-}^{\ee_3}}^\eps-\E W_{l_{F_2,+}^{\ee_3}}^\eps\Big)
=
 	\frac1{2\eps^2}\Big(\E W_{l_{F_2,-}^{\ee}}^\eps-\E W_{l_{F_2,+}^{\ee}}^\eps\Big).
 \end{align*}
 Hence, by \eqref{e:F2d} we have
 \begin{equ}
 	\lim_{\eps\to 0}\frac1{2\eps^2}\Big(\E W_{l_{F_2,-}^{\ee_3}}^\eps-\E W_{l_{F_2,+}^{\ee_3}}^\eps\Big) =  2(\p_{t_3}-\p_{t_2})\E W_l+I_3.
 \end{equ}
 
Consider the loop 
 \[
 l_{F_3,-}^{\ee_3}=e^\eps e_1^\eps A^\eps (e_4^\eps)^{-1}e^\eps e_2^\eps B^\eps e_{3,\neg}''
 \]
  from the negative deformation  of \eqref{e:genloop} at $\ee_3$ in the face $F_3$, where $e_{3,\neg}''$ is given by replacing $\ee_3$ in $e_3^\eps$ by $\ee_3''$, which consists $3$ bonds.

We also have the loop 
 \[
 l_{F_3,+}^{\ee_3}=e^\eps e_1^\eps A^\eps (e_4^\eps)^{-1}e^\eps e_2^\eps B^\eps  e_{3,\pos}''
 \]
  from the positive deformation  of \eqref{e:genloop} at $e_3^\eps$ in the face $F_3$, where $e_{3,\pos}''$ is given by replacing $\ee_3$ in $e_3^\eps$ by $\ee_3(\ee_3'')^{-1}\ee_3$, which consists $3$ bonds. 
 Using change of variable as before we find
  \begin{align*}
& \E W_{l_{F_3,-}^{\ee_3}}^\eps	-\E W_{l_{F_3,+}^{\ee_3}}^\eps
 	\\=\int &\tr\Big(a_1\alpha a_4^{-1}  a_2\beta a_3^{-1}(Q-Q^{-1})\Big)\prod_{i=1,2,4}S^\eps_{\frac{t_i(\eps)}{\eps^2}}(a_i\alpha_i a_{i+1}^{-1})
 	\\&\times S^\eps_{\frac{t_3(\eps)-\eps^2}{\eps^2}}(Q^{-1}a_3\alpha_3 a_4^{-1})S^\eps(Q)\,
	\nu^\eps(\mathbf{b})
	\,\dif \mathbf{a}\,\dif \mathbf{b}\,\dif Q.
 \end{align*}
 Similar as before and using Lemma \ref{lem:J1} we know 
 \begin{align*}
 	\lim_{\eps\to 0}\frac1{2\eps^2}\Big(\E W_{l_{F_3,-}^{\ee_3}}^\eps-\E W_{l_{F_3,+}^{\ee_3}}^\eps\Big) = I_3. 
 \end{align*}
 The result then follows. 
 \end{proof}

 \bt\label{th:g1} 
 The following linear combination of the discrete master loop equation
 \begin{align*}
 \eqref{eq:loop:g}-\frac12\eqref{eq:loop:g1}-\frac12 \eqref{eq:loop:g2}
 \end{align*}
 converges to 
 \begin{equ}[e:DHK1713]
 \Big(\frac{\p}{\p t_1}-\frac{\p}{\p t_2}+\frac{\p}{\p t_3}-\frac{\p}{\p t_4}\Big)\E W_l=\E W_{l_{1}}W_{l_{2}},
 \end{equ}
which coincides with \cite[(1.3)]{Driver17}. 
 \et 
 
\begin{proof}
This follows by combining Proposition~\ref{prop:g}, Lemma~\ref{lem:ms1} and Lemma~\ref{lem:ms2}.
\end{proof}
 
Once we have the above theorem,
we can also extend the convergence result to more general linear combinations of the discrete master loop equations. 
 Recalling the notation in \eqref{e:genloop-bar},
for $\underline{\ee}\in \underline{e}^\eps$ and $\ee_2\in e_2^\eps, \ee_4\in e_4^\eps$ we have 
  \begin{equs}[eq:loop:gn]
 	\E	W_l^\eps & =\frac{1}{2\eps^2} \!\!\! \sum_{l'\in \mathbb{D}^-_{\underline{\ee}}(l)}\!\!\E W_{l'}^\eps
 	-\frac{1}{2\eps^2} \!\!\! \sum_{l'\in \mathbb{D}^+_{\underline{\ee}}(l)} \!\!\E W_{l'}^\eps
 	-\E W_{l_1}^\eps W_{l_2}^\eps \;,
 \end{equs}
  \begin{equs}
 	\E	W_l^\eps & =\frac{1}{2\eps^2} \!\!\! \sum_{l'\in \mathbb{D}^-_{\ee_2}(l)}\!\!\E W_{l'}^\eps
 	-\frac{1}{2\eps^2} \!\!\! \sum_{l'\in \mathbb{D}^+_{\ee_2}(l)} \!\!\E W_{l'}^\eps\;,
	\label{eq:loop:g1n}
\\
 	\E	W_l^\eps & =\frac{1}{2\eps^2} \!\!\! \sum_{l'\in \mathbb{D}^-_{\ee_4}(l)}\!\!\E W_{l'}^\eps
 	-\frac{1}{2\eps^2} \!\!\! \sum_{l'\in \mathbb{D}^+_{\ee_4}(l)} \!\!\E W_{l'}^\eps\;.
	\label{eq:loop:g2n}
 \end{equs}
 
 \bc  \label{cor:lin-comb}
 For $a_i,b_j\in\R$, $i=1,\dots,4$, $j=1,2$, satisfying $\sum_{i=0}^4a_i=1$, $\sum_{j=1}^2b_j=1$, the following linear combination of the discrete loop equations
 \begin{align*}
 	b_1\eqref{eq:loop:g}+b_2\eqref{eq:loop:gn}-a_0\E W_l^\eps-a_1\eqref{eq:loop:g1}-a_2\eqref{eq:loop:g1n}-a_3 \eqref{eq:loop:g2}-a_4\eqref{eq:loop:g2n}
 \end{align*}
 converges to \eqref{e:DHK1713}. 
 \ec 
 \begin{proof} We set for $i=1,\dots,4$
 	\begin{align*}
 			\cD_i\eqdef\lim_{\eps\to0}\Big(\frac{1}{2\eps^2} \!\!\! \sum_{l'\in \mathbb{D}^-_{\ee_i}(l)}\!\!\E W_{l'}^\eps
 		-\frac{1}{2\eps^2} \!\!\! \sum_{l'\in \mathbb{D}^+_{\ee_i}(l)} \!\!\E W_{l'}^\eps\Big),
 	\end{align*}
 	and
 		\begin{align*}
 		\cD_{\ee} & \eqdef\lim_{\eps\to0}\Big(\frac{1}{2\eps^2} \!\!\! \sum_{l'\in \mathbb{D}^-_{\ee}(l)}\!\!\E W_{l'}^\eps
 		-\frac{1}{2\eps^2} \!\!\! \sum_{l'\in \mathbb{D}^+_{\ee}(l)} \!\!\E W_{l'}^\eps\Big),
\\
 		\cD_{\underline\ee} & \eqdef\lim_{\eps\to0}\Big(\frac{1}{2\eps^2} \!\!\! \sum_{l'\in \mathbb{D}^-_{\underline\ee}(l)}\!\!\E W_{l'}^\eps
 		-\frac{1}{2\eps^2} \!\!\! \sum_{l'\in \mathbb{D}^+_{\underline\ee}(l)} \!\!\E W_{l'}^\eps\Big).
 	\end{align*}
Recall \eqref{e:conv-loop} and \eqref{e:con:loops}.
Taking limits on both sides of \eqref{eq:loop:g1}, \eqref{eq:loop:g1n},  \eqref{eq:loop:g2} and \eqref{eq:loop:g2n}  we have 
 \begin{align*}
 	\E W_l=\cD_1=\cD_2=\cD_3=\cD_4.
 \end{align*}
 On the other hand, taking limits on both sides of \eqref{eq:loop:g}, \eqref{eq:loop:gn} we have
 \begin{align*}
 		\E W_l+\E W_{l_1}W_{l_2}=\cD_{\ee}=\cD_{\underline \ee}.
 \end{align*}
 Hence, 
  \begin{align*}
 	b_1\cD_{\ee}+b_2\cD_{\underline \ee}-a_0\E W_l-\sum_{i=1}^4a_i\cD_i=\cD_{\ee}-\frac12\cD_1-\frac12\cD_3.
 \end{align*}
By Theorem \ref{th:g1} the RHS is the desired combination to get the alternating sum of the area derivatives, so the result follows. 
 	\end{proof}
 
 \br\label{re:ee} 
 It is also possible to take the lattice approximation of the loop $l$ as 
 \begin{equ}
 	l^\eps=e^\eps e_1^\eps A^\eps (e_4^{\eps})^{-1}(e^\eps)^{-1} e_2^\eps B^\eps (e_3^{\eps})^{-1}
 \end{equ}
 where $e^\eps$ and $(e^\eps)^{-1}$ appear as shown in the following picture,
 instead of \eqref{e:genloop} where $e^\eps$ appears twice.
 This will affect the above proof, but will
eventually gives the same result as Theorem \ref{th:g1}. 
 \[
 \begin{tikzpicture}[scale=2]
\draw[thick,midarrow2] (-0.1,0) -- (0.1,0)   node[midway,below] {$e^\eps$};
\draw[thick,postaction={decorate,decoration={markings,mark=at position 0.85 with {\arrow{>}}}}] (0.1,0) -- (-0.1,0);
\draw[thick,->] (0.1,0)
\foreach \i in {1, 2, ..., 14} {-- ++({(0.5+0.5*(-1)^\i)*0.1}, {(0.5-0.5*(-1)^\i)*0.1})};
\node at (0.9,0.6) {$e_1^\eps$};
\draw[thick,->] (0.1,0)
\foreach \i in {1, 2, ..., 14} {-- ++({(0.5+0.5*(-1)^\i)*0.1}, {-(0.5-0.5*(-1)^\i)*0.1})};
\node at (0.9,-0.6) {$e_4^\eps$};
\draw[thick,->] (-0.1,0)
\foreach \i in {1, 2, ..., 14} {-- ++({-(0.5+0.5*(-1)^\i)*0.1}, {(0.5-0.5*(-1)^\i)*0.1})};
\node at (-0.9,0.6) {$e_2^\eps$};
\draw[thick,->] (-0.1,0)
\foreach \i in {1, 2, ..., 14} {-- ++({-(0.5+0.5*(-1)^\i)*0.1}, {-(0.5-0.5*(-1)^\i)*0.1})};
\node at (-0.9,-0.6) {$e_3^\eps$};
\node at (0, 0.7) {$F_1^\eps$};\node at (0, -0.7) {$F_3^\eps$};
\node at (0.7,0) {$F_4^\eps$};\node at (-0.7,0) {$F_2^\eps$};
\end{tikzpicture}
\]
Indeed, instead of the discrete master loop equation \eqref{eq:loop:g} we now have
  \begin{equs}[eq:loop:g1-n]
 	\E	W_l^\eps & =\frac{1}{2\eps^2} \!\!\! \sum_{l'\in \mathbb{D}^-_{\ee}(l)}\!\!\E W_{l'}^\eps
 	-\frac{1}{2\eps^2} \!\!\! \sum_{l'\in \mathbb{D}^+_{\ee}(l)} \!\!\E W_{l'}^\eps
 	+\E W_{l_1}^\eps W_{l_2}^\eps
 \end{equs}
where the sign in front of $\E W_{l_1}^\eps W_{l_2}^\eps$ becomes $+$,
since the positive splitting in \eqref{eq:loop:g} is changed to negative splitting. 
We then have the deformations w.r.t. the faces $F_1$, $F_3$ instead of the faces $F_2$, $F_4$. Exactly the same arguments as in Proposition~\ref{prop:g} show that
\begin{equs}[e:same5.5]
	\frac{1}{2\eps^2}(\E W_{l^{\ee}_{F_1,-}}^\eps
	-\E W_{l^{\ee}_{F_1,+}}^\eps) &\to I_1=2(\p_{t_4}-\p_{t_1})\E W_l+I_4,
\\
	\frac{1}{2\eps^2}(\E W_{l^{\ee}_{F_3,-}}^\eps
	-\E W_{l^{\ee}_{F_3,+}}^\eps) &\to I_3=2(\p_{t_2}-\p_{t_3})\E W_l+I_2,
\end{equs}
with  $I_m$ as in \eqref{def:Im}.
 As a result, \eqref{eq:loop:g1-n} converges to 
  \begin{align}\label{e:loop:1n}
 	&-2(\p_{t_1}-\p_{t_2}+\p_{t_3}-\p_{t_4}) \E W_l+I_2+I_4
 	=-\E W_{l_1}W_{l_2}+\E W_l.
 \end{align}
 Similarly  as Lemmas~\ref{lem:ms1}+\ref{lem:ms2},
again considering deformations at $\ee_1$ and $\ee_3$, the discrete master loop equations \eqref{eq:loop:g1} and \eqref{eq:loop:g2} converge to 
  \begin{equs}
 	2(\p_{t_4}-\p_{t_1})\E W_l+2I_4 &=\E W_l,
\\
 	2(\p_{t_2}-\p_{t_3})\E W_l+2I_2 &=\E W_l. 
 \end{equs}
 Hence,  
 $\eqref{eq:loop:g1-n}-\frac12\eqref{eq:loop:g1}-\frac12 \eqref{eq:loop:g2}$
 converges to  \eqref{e:DHK1713}, the same limiting loop equation as before.
 \er 
 
 \br\label{re:com}  
 From the above proof we have  that for $\ee $ in $e^\eps$  or $\ee_i$ in $e^\eps_i$
 \begin{equ}\label{con:defor}
 	\frac{1}{2\eps^2}(\E W_{l^{\ee}_{F_m,-}}^\eps
 	-\E W_{l^{\ee}_{F_m,+}}^\eps)\to I_m, 
 \end{equ}
 for $\ee\subset \p F_m$, with $I_m$ given in \eqref{def:Im}.
 Furthermore, if we track the proof in \cite{Driver17}, we find that  $I_m$ is just $A_{m}$ on Page 16 of \cite{Driver17}. $I_m$ also appears in \cite[(94)]{Levy11} and the proof of Proposition 6.4 there. 
 
 We also note that \eqref{loop:g:2} and \eqref{loop:g:3}   are also special master loop equations derived in (94) from \cite[Proposition 6.4]{Levy11} with $n=1$. In fact, we can use  \cite[Proposition 6.16, item 1.]{Levy11} to replace $\frac12\Delta^{e_1}f$ there by $-\frac12\E W_l$ and we get  \eqref{loop:g:2} and \eqref{loop:g:3}.
 \er

We also note that our proof of course never relies on 
the master loop equation in continuum.
In fact, if we had the continuum  loop equations 
at our disposal,
then by simply replacing $\E W_l$ and the splitting term in \eqref{eq:loop:g} by their limits 
using \eqref{e:conv-loop}\eqref{e:con:loops}, we would generally obtain that the sum of {\it all} deformations converges to some area derivatives and correction terms $I_m$.
However we  emphasize that our main contribution is to directly analyze the contribution of each deformation, which provides more informative limiting result, and also  as a corollary
we obtain a new proof of the continuum loop equation \eqref{e:DHK1713} from discrete approximations.
This approach also enables us to derive \eqref{con:defor}, which is a stronger result than simply proving the convergence of the sum of all deformations.

\subsection{Some degenerate cases}\label{sec:deg}

We conclude this section by showing how certain special and degenerate cases follow from our general results above. The first case is that the intersecting vertex $v$ above is  adjacent to less than $4$ faces: 
\[
\begin{tikzpicture}
    \draw[thick] (0,0) ellipse [x radius=1.6, y radius=1];
   \draw[thick,bend left =50] (-1.6,0) to (-0.3,0); 
    \draw[thick,bend right =50] (-1.6,0) to (-0.3,0); 
\draw[thick] (-0.3,0) .. controls (1.2,1) and (1.2,-1) .. (-0.3,0);
\node at (-0.3,-0.3) {$v$};
\end{tikzpicture}
\qquad\qquad
\begin{tikzpicture}
    \draw[thick] (0,0) ellipse [x radius=1.6, y radius=1];
   \draw[thick,bend left =50] (-1.6,0) to (-0.3,0); 
    \draw[thick,bend right =50] (-1.6,0) to (-0.3,0); 
\draw[thick] (-0.3,0) .. controls (1.2,1) and (1.2,-1) .. (-0.3,0);
\draw[thick,red] (1.6,0) -- (0.81,0); 
\node at (-0.3,-0.3) {$v$};
\node at (0,0.6) {$F_1$};\node at (0,-0.6) {$F_3$};
\node at (-0.9,0) {$F_2$};\node at (0.5,0) {$F_4$};
\end{tikzpicture}
\]
In this case
we can add a new  edge   (red  line in above picture) with an arbitrary orientation, so that $v$ is again adjacent to $4$ faces, to reduce it to the general case. 
Note that we only change the {\it graph} and do not change the {\it loop} $l$;
in particular $W_l$ remains the same.
Writing  $s=t_1+t_3 = |F_1|+|F_3|$,
 we have $p_s=p_{t_1}*p_{t_3}$,
where the convolution variable is the holonomy along the red edge,
so adding the edge simply amounts to  replacing the heat kernel  $p_s$ showing up in  Driver's formula for the Wilson loop  expectation by $p_{t_1}*p_{t_3}$.
We then apply Theorem \ref{th:g1} to get that suitable linear combination of discrete master loop equation converges to
\begin{align*}
	\Big(\frac{\p}{\p t_1}-\frac{\p}{\p t_2}+\frac{\p}{\p t_3}-\frac{\p}{\p t_4}\Big)\E W_l=\E W_{l_{1}}W_{l_{2}}.
\end{align*}
Since $\p_s \E W_l=\p_{t_1}\E W_l=\p_{t_3}\E W_l$, which follows from
 $\p_s p_s=\p_{t_1}p_{t_1}*p_{t_3} = p_{t_1}*\p_{t_3}p_{t_3}$
inside Driver's formula, we then get the following master loop equation:
\begin{align*}
	\Big(2\frac{\p}{\p s}-\frac{\p}{\p t_2}-\frac{\p}{\p t_4}\Big)\E W_l=\E W_{l_{1}}W_{l_{2}}.
\end{align*}

Another degenerate case is when some of the faces adjacent to $v$ is unbounded.
\[
\begin{tikzpicture}
\draw[thick,midarrow] (0,0) .. controls (-1,-1.2) and (1,-1.2) .. (0,0);
\draw[thick] (0,0) .. controls (1.5,0.5) and (1.5,-2)  .. (0,-2);
\draw[thick,midarrow] (0,0) .. controls (-1.5,0.5) and (-1.5,-2)  .. (0,-2);
\node at (0,0.2) {$v$};
\end{tikzpicture}
\qquad\qquad
\begin{tikzpicture}
\draw[thick] (0,0) .. controls (-1,-1.2) and (1,-1.2) .. (0,0);
\draw[thick] (0,0) .. controls (1.5,0.5) and (1.5,-2)  .. (0,-2);
\draw[thick] (0,0) .. controls (-1.5,0.5) and (-1.5,-2)  .. (0,-2);
\node at (0,-0.2) {$v$};
\draw[thick,blue] (0,-0.9) -- (0,-2);
\draw[thick,blue] (-0.64,0) .. controls (-0.64,0.7) and (0.64,0.7)  .. (0.64,0);
\node at (-0.6,-1) {$F_1$};\node at (0.6,-1) {$F_3$};
\node at (0,-0.6) {$F_2$};\node at (0,0.2) {$F_4$};
\end{tikzpicture}
\]
For this loop we can add two blue edges as in the right picture. 
By the same reason as above, the extra lines do not change the expectation of the Wilson loop. In fact, for the blue edge on the top, the Wilson loop does not depend on this edge and we can integrate it out in Driver's formula. 
For the blue edge in the bottom, for $s=t_1+t_3$ we can write $p_s=p_{t_1}*p_{t_3}$.
We then apply Theorem \ref{th:g1} to get that suitable linear combination of discrete master loop equation converges to
\begin{align*}
	\Big(\frac{\p}{\p t_1}-\frac{\p}{\p t_2}+\frac{\p}{\p t_3}-\frac{\p}{\p t_4}\Big)\E W_l=\E W_{l_{1}}W_{l_{2}}\;.
\end{align*}
As argued above $\p_s \E W_l=\p_{t_1}\E W_l=\p_{t_3}\E W_l$ and we also have $\p_{t_4}\E W_l=0$, we then get the following master loop equation:
\begin{align*}
	\Big(2\frac{\p}{\p s}-\frac{\p}{\p t_2}\Big)\E W_l=\E W_{l_{1}}W_{l_{2}}\;.
\end{align*} 
This coincides with \cite[Example 6.13]{Levy11}. In fact, it is 
$$
\text{\cite[(106)]{Levy11}}
-2\times \text{ the first equation in \cite[Example 6.13]{Levy11}}.
$$

 \section{Some extensions}\label{sec:ext}
 \subsection{Extension to strings (collections of loops)}\label{sec:ext1}

We extend the  result above on a single loop to a {\it string}  $s=(l_{1},\dots,l_{m})$, which means a collection of loops in $\R^2$, which only has simple crossings
(more precisely, each crossing is either only a simple  self-crossing of $l_i$ for some $i$, or only a simple  crossing between $l_i$ and $l_j$ for some $i\neq j$).
Define 
\begin{equ}[e:phi]
	W_s\eqdef W_{l_1}W_{l_2}\cdots W_{l_m}\;.
\end{equ}
Remark that $s$ can be viewed as a {\it graph} in the sense of Section~\ref{sec:Dri}.
Similarly as the discussion above Remark~\ref{rem:finer},
since in a graph the edges are required to meet only at their endpoints,
for each simple crossing $v$ of $s$, if $v$ belongs to an edge $e$ of this graph then $v$ must be an endpoint of $e$.
The function $W_s$ is gauge invariant for which Driver's formula applies.

Recalling the sets of mergers $\mathbb{M}^{\pm}_{U,e}$ defined below \eqref{e:merger2},
the lattice loop equation for strings with $G=U(N)$ is given by (see  \cite[Theorem~5.7]{CPS2023})
\begin{equs}[eq:ws]
	\E W_s^\eps 
	& =\frac{1}{2\eps^2} \!\!\! \sum_{s'\in \mathbb{D}^-_{\ee}(s)}\!\!\E W_{s'}^\eps
	-\frac{1}{2\eps^2} \!\!\! \sum_{s'\in \mathbb{D}^+_{\ee}(s)} \!\!\E W_{s'}^\eps
	+\!\!\!\sum_{s'\in \mS^{-}_{\ee}(s)}\!\!\E W_{s'}^\eps
	-\!\!\! \sum_{s'\in \mS^{+}_{\ee}(s)}\!\!\E W_{s'}^\eps
	\\&\qquad\qquad\qquad\qquad
	+\frac1{N^2} \!\!\!\sum_{s'\in \mM_{U,\ee}^{-}(s)}\!\!\E W_{s'}^\eps
	-\frac1{N^2}\!\!\! \sum_{s'\in \mM_{U,\ee}^{+}(s)}\!\!\E W_{s'}^\eps\,,
\end{equs}
where $\ee$ is a fixed bond in the string $s$. 
Note that the location of the bond $\ee$ is also fixed, but 
as in \eqref{e:ignore-loc} we hide these locations in our notation.

\medskip

We first consider the case of a simple self-crossing for one of the loops in $s$.

Suppose that $l_1$ in $s$  is given as $l$ in Section~\ref{sec:gen}, i.e. $l_1=e_1Ae_4^{-1}e_2Be_3^{-1}$, see Fig.~\ref{fig:DHKcase}. 
As explained above (around \eqref{e:phi}),
 the other loops $l_j$ for  $j\neq 1$ in $s$ do not cross $v$. 
 We can also assume that the edges $e_1, e_3$ do not intersect with any $l_j$ for $j\geq 2$ by choosing $e_1, e_3$ short enough as explained in Remark~\ref{rem:finer}.
The same argument as in Section \ref{sec:gen} implies that suitable linear combination of discrete master loop equation  converges to 
\begin{align}\label{ms:s}
	\Big(\frac{\p}{\p t_1}-\frac{\p}{\p t_2}+\frac{\p}{\p t_3}-\frac{\p}{\p t_4}\Big)\E W_s
	=\E \Big(W_{l_{11}}W_{l_{12}}\prod_{j=2}^m  W_{l_j}\Big)\;,
\end{align}
with $l_{11}=e_1Ae_4^{-1}$ and $l_{12}=e_2Be_3^{-1}$. 

Indeed, as before we can  find suitable lattice approximation to $s$, and an edge $e^\eps$ with $|e^\eps|\sim \eps$ to approximate the crossing point $v$, with $e_1^\eps, e_3^\eps$ as the discrete counterparts of $e_1,e_3$.  
We then apply the master loop equation \eqref{eq:ws} choosing $\ee$
to be a bond in $e^\eps$ and then bonds in $e_1^\eps, e_3^\eps$, as in Section~\ref{sec:gen}. 

Now the formula \eqref{eqc:WLd} for $W_l^\eps$ 
obviously extends to  $W_s^\eps$,
and all the formulas for $W_{l'}^\eps$ with  
$l' \in \mathbb{D}^\pm_{\ee}(l) \cup  \mathbb{D}^\pm_{\ee_1}(l) \cup   \mathbb{D}^\pm_{\ee_3}(l) $ 
in Section~\ref{sec:gen}
obviously extend to  $W_{s'}^\eps$ for 
$s' \in \mathbb{D}^\pm_{\ee}(s) \cup  \mathbb{D}^\pm_{\ee_1}(s) \cup   \mathbb{D}^\pm_{\ee_3}(s)$. 
The only differences  are:
\begin{itemize}
\item
one has an additional factor $\prod_{j\ge 2} W^\eps_{l_i}$ in the integrands,
\item
 the edge variables 
 $\alpha_i$ therein now may depend on the edges in $l_j$ for $j\ge 2$,
\item
there are additional $\mathbf{b}$ variables arising from $l_j$ for $j\ge 2$,
\item
there are additional Wilson actions in $\nu^\eps(\mathbf{b})$.
\end{itemize}
Note that whether the edge variables  $\alpha_i$ depend on the other loops
is irrelevant for all the arguments in Section~\ref{sec:gen};
in particular all the derivatives, gradients, Laplacian and integration by parts are with respect to the variables $a_i$, not $\alpha_i$.
Therefore, the other Wilson loops $W_{l_i}$ for  $i\ge 2$ do not affect the calculation in Section \ref{sec:gen}, and \eqref{ms:s} follows as before
in the limit as $\eps\to 0$ of  suitable linear combinations of discrete loop equations
as in Theorem~\ref{th:g1} or Corollary~\ref{cor:lin-comb}.

\medskip

Next, we consider the case of a simple crossing between two loops in $s$,
for which we will see a merger of these two loops.

We assume that $l_1$ and $l_2$ in $s$ have generic forms
 \[
 l_1=e_3^{-1}e_1A \;,\qquad
 l_2=e_4^{-1}e_2B
 \]
as illustrated in Fig.~\ref{fig:merger}. Here we are using similar notation as in Section~\ref{sec:gen}   where $A$ and $B$ are sequences of edges 
which do not involve 
$e_1,e_2,e_3,e_4$. 
Also  $e_i$ for $i=1,\dots,4$ do not belong to the other loops $l_j$ for $j\geq 3$ in $s$. 
The crossing  vertex $v$ belongs to every $e_i$ for $i=1,\dots,4$. 
 \begin{figure}[h] 
  \centering
 \begin{tikzpicture}[scale=1.3]
\filldraw [black] (0,0) circle (1pt); \node at (0, -0.2) {$v$};
\draw[thick,->] (0,0) -- (1,1) node[midway,right] {$e_1$};
\draw[thick,->] (0,0) -- (1,-1) node[midway,right] {$e_4$};
\draw[thick,->] (0,0) -- (-1,1) node[midway,left] {$e_2$};
\draw[thick,->] (0,0) -- (-1,-1) node[midway,left] {$e_3$};
\draw[dashed,->] (1.02,1.02) arc[start angle=35, end angle=235, radius=1.47];
\draw[dashed,->] (-1.02,1.02) arc[start angle=145, end angle=-55, radius=1.47];
\node at (-1.6,1.3) {$A$}; \node at (1.6,1.3) {$B$};
\node at (0, 0.7) {$F_1$};\node at (0, -0.7) {$F_3$};
\node at (0.7,0) {$F_4$};\node at (-0.7,0) {$F_2$};
\end{tikzpicture}
\qquad
 \begin{tikzpicture}[scale=1.3]
\filldraw [gray] (0,0) circle (1pt); 
\draw[thick,->] (0.1,0.1) -- (1,1) node[midway,right] {$e_1$};
\draw[thick,->] (0.1,-0.1) -- (1,-1) node[midway,right] {$e_4$};
\draw[thick,->] (-0.1,0.1) -- (-1,1) node[midway,left] {$e_2$};
\draw[thick,->] (-0.1,-0.1) -- (-1,-1) node[midway,left] {$e_3$};
\draw[thick] (0.1,0.1) arc[start angle=150, end angle=210, radius=0.2];
\draw[thick] (-0.1,0.1) arc[start angle=30, end angle=-30, radius=0.2];
\draw[dashed,->] (1.02,1.02) arc[start angle=35, end angle=235, radius=1.47];
\draw[dashed,->] (-1.02,1.02) arc[start angle=145, end angle=-55, radius=1.47];
\node at (-1.6,1.3) {$A$}; \node at (1.6,1.3) {$B$};
\node at (0, 0.7) {$F_1$};\node at (0, -0.7) {$F_3$};
\node at (0.7,0) {$F_4$};\node at (-0.7,0) {$F_2$};
\end{tikzpicture}
\qquad
\begin{tikzpicture}[scale=1.8]
\draw[thick] (0,-0.1) -- (0,0.1)   node[midway,right] {$e^\eps$};
\draw[thick,->] (0,0.1)
\foreach \i in {1, 2, ..., 15} {-- ++({(0.5-0.5*(-1)^\i)*0.1}, {(0.5+0.5*(-1)^\i)*0.1})};
\node at (0.7,0.5) {$e_1^\eps$};
\draw[thick,->] (0,-0.1)
\foreach \i in {1, 2, ..., 15} {-- ++({(0.5-0.5*(-1)^\i)*0.1}, {-(0.5+0.5*(-1)^\i)*0.1})};
\node at (0.7,-0.5) {$e_4^\eps$};
\draw[thick,->] (0,0.1)
\foreach \i in {1, 2, ..., 15} {-- ++({-(0.5-0.5*(-1)^\i)*0.1}, {(0.5+0.5*(-1)^\i)*0.1})};
\node at (-0.7,0.5) {$e_2^\eps$};
\draw[thick,->] (0,-0.1)
\foreach \i in {1, 2, ..., 15} {-- ++({-(0.5-0.5*(-1)^\i)*0.1}, {-(0.5+0.5*(-1)^\i)*0.1})};
\node at (-0.7,-0.5) {$e_3^\eps$};
\draw[dashed,->] (0.83,0.83) arc[start angle=35, end angle=235, radius=1.2];
\draw[dashed,->] (-0.83,0.83) arc[start angle=145, end angle=-55, radius=1.2];
\node at (-1.3,1) {$A^\eps$}; \node at (1.3,1) {$B^\eps$};
\node at (0, 0.7) {$F_1^\eps$};\node at (0, -0.7) {$F_3^\eps$};
\node at (0.7,0) {$F_4^\eps$};\node at (-0.7,0) {$F_2^\eps$};
\end{tikzpicture}
\caption{Illustrations for $\{l_1,l_2\}$, their 
merger $l_{12}$ in Theorem~\ref{th:g2}, and the  lattice approximations. 
Here the dashed lines $A,B$ only mean generic sequences of edges and can allow arbitrary shapes.} 
\label{fig:merger}
\end{figure}
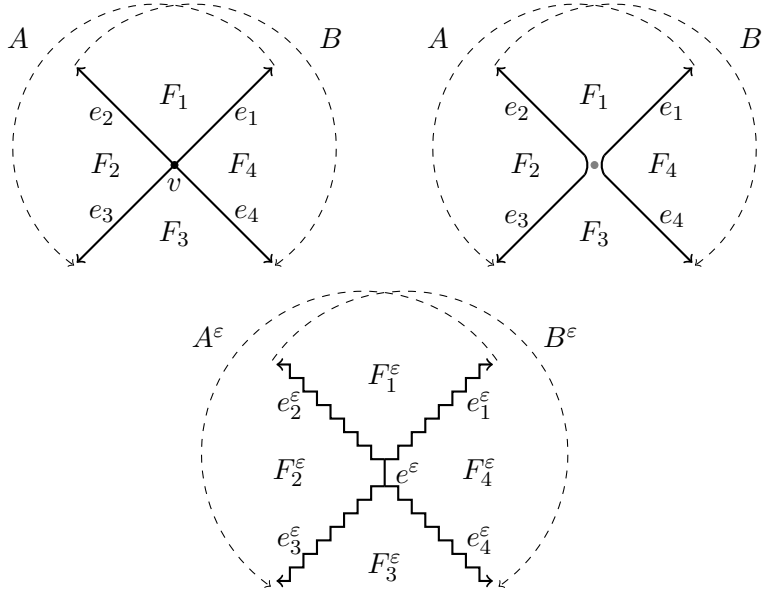

We now consider a class of  lattice approximation $\{l^\eps_1\}$ and $\{l^\eps_2\}$ to $l_1$ and $l_2$ ``with respect to $v$'', meaning that 
\begin{equ}[e:before-merge]
l^\eps_1=(e_3^{\eps})^{-1}e^\eps e_1^\eps A^\eps ,\qquad l_2^\eps=(e_4^{\eps})^{-1}e^\eps e_2^\eps B^\eps .
\end{equ}
The map $i_\eps$ as in \eqref{e:ij} takes
$(e_1,e_2,e_3,e_4,A,B)$ to $(e_1^\eps,e_2^\eps,e_3^\eps,e_4^\eps,A^\eps,B^\eps)$,
with  $e^\eps \notin \mbox{Image}(i_\eps)$.
Moreover, for $i\ge 3$, let  $\{l^\eps_i\}$
be lattice approximations of $l_i$.
We also write 
$$W_s^\eps=\prod_{i=1}^m W_{l_i}^\eps=\prod_{i=1}^m W_{l_i^\eps}.$$ 

As before,
we choose three bonds  $\ee\in e^\eps$, $\ee_1 \in e^\eps_1$,  $\ee_3\in (e^\eps_3)^{-1}$.
Applying  \eqref{eq:ws} for the bond $\ee \in e^\eps$ gives  the following master loop equation:
\begin{equs}[eq:ws-m]
	\E W_s^\eps 
	& =\frac{1}{2\eps^2} \!\!\! \sum_{s'\in \mathbb{D}^-_{\ee}(s)}\!\!\E W_{s'}^\eps
	-\frac{1}{2\eps^2} \!\!\! \sum_{s'\in \mathbb{D}^+_{\ee}(s)} \!\!\E W_{s'}^\eps
	-\frac1{N^2}\E \Big(W_{l_{12}}^\eps\prod_{i=3}^n W_{l_i}^\eps\Big)\;,
\end{equs}
with  $l_{12}^\eps$ being the merger
\[
l_{12}^\eps=(e_3^{\eps})^{-1} e^\eps e_2^\eps B^\eps (e_4^{\eps})^{-1} e^\eps e_1^\eps A^\eps\;.
\]
Moreover, applying \eqref{eq:ws} with $\ee_1$ and $\ee_3$ (both from $l_1^\eps$), gives the following master loop equations respectively:
\begin{equs}
	\E W_s^\eps 
	& =\frac{1}{2\eps^2} \!\!\! \sum_{s'\in \mathbb{D}^-_{\ee_1}(s)}\!\!\E W_{s'}^\eps
	-\frac{1}{2\eps^2} \!\!\! \sum_{s'\in \mathbb{D}^+_{\ee_1}(s)} \!\!\E W_{s'}^\eps\;,\label{eq:ws-m1}
\\
	\E W_s^\eps 
	& =\frac{1}{2\eps^2} \!\!\! \sum_{s'\in \mathbb{D}^-_{\ee_3}(s)}\!\!\E W_{s'}^\eps
	-\frac{1}{2\eps^2} \!\!\! \sum_{s'\in \mathbb{D}^+_{\ee_3}(s)} \!\!\E W_{s'}^\eps\;.\label{eq:ws-m2}
\end{equs}
Similar as in Section \ref{sec:gen} we can derive the following result. 

 \bt\label{th:g2} 
 The following linear combination of discrete master loop equation
\begin{align*}
	\eqref{eq:ws-m}-\frac12\eqref{eq:ws-m1}-\frac12 \eqref{eq:ws-m2}
\end{align*}
converges to 
\begin{equ}[e:LevyThm31]
	\Big(\frac{\p}{\p t_1}-\frac{\p}{\p t_2}+\frac{\p}{\p t_3}-\frac{\p}{\p t_4}\Big)\E W_s=\frac1{N^2}\E \Big(W_{l_{12}}\prod_{i=3}^m W_{l_{i}}\Big)\;,
\end{equ}
with $l_{12}=e_{3}^{-1}e_2Be_4^{-1}e_1A$.
\et 
The above result \eqref{e:LevyThm31} recovers \cite[Theorem 3.1]{LevyNotes}. 

\begin{proof} Using \eqref{e:conv-loop} and \eqref{e:con:loops} and Remark \ref{re:ax} we know  that $\E	W_s^\eps$ on the LHS  and $\E W_{l_{12}}^\eps \prod_{i=3}^m W_{l_i}^\eps$ on the RHS of  \eqref{eq:ws} converge to $\E W_s$ and $\E( W_{l_{12}} \prod_{i=3}^m W_{l_i})$, respectively. For the terms involving deformations in \eqref{eq:ws-m}, \eqref{eq:ws-m1} and \eqref{eq:ws-m2}, it gives the same contribution as the case in Section \ref{sec:gen}. More precisely, by Driver's formula \eqref{dri:dis} 
\begin{equs}[eqc:WLd-m]
	\E W_{s}^\eps
	=\int &\tr(Q_ea_1\alpha a_3^{-1})  
	\tr(Q_ea_2\beta a_4^{-1})
	\prod_{i=1}^4S^\eps_{\frac{t_i(\eps)}{\eps^2}}(a_i\alpha_ia_{i+1}^{-1})
	\\&\times\Big(\prod_{i=3}^m W^\eps_{l_i}\Big)
	 \nu^\eps(\mathbf{b})\, \dif \mathbf{a}\,\dif \mathbf{b},
	\end{equs} 
	 where $\mathbf{a}=(a_1,a_2,a_3,a_4)\in G^4$, $\mathbf{b}$ denote the edge variables not involving $\{a_1,a_2,a_3,a_4\}$ and $\nu^\eps(\mathbf{b})$ is a product of Wilson actions in the $\mathbf{b}$ variables.
	 Since the boundaries of the faces adjacent to the edge $e^\eps$ do not change, 
	 namely they are still $a_i$ as before,
	 based on the same discussion after \eqref{ms:s}, 
 the deformation for the bond $\ee$ in $e^\eps$ in the loop $l_1^\eps$ gives the same contribution as that in Proposition \ref{prop:g}.	 
 Hence, the discrete master loop equation \eqref{eq:ws-m} converges to  
	  \begin{align*}
	 	&2(\p_{t_1}-\p_{t_2}+\p_{t_3}-\p_{t_4}) \E W_s+I_1^s+I_3^s
	 	=\frac1{N^2}\E \Big(W_{l_{12}}\prod_{i=3}^m W_{l_{i}}\Big)+\E W_s,
	 \end{align*}
 where 
	 \begin{align*}
	 	I_m^s=\int &\<\nabla_{b_m}\tr(a_1\alpha a_3^{-1})  ,\nabla_{b_m}p_{t_m}(a_m\alpha_m a_{m+1}^{-1})\>\tr(a_2\beta a_4^{-1})
	 	\\&\times \prod_{i\in\{1,\dots,4\}\backslash\{m\}} p_{t_i}(a_i\alpha_i a_{i+1}^{-1})
		\Big(\prod_{i\ge 3} W_{l_i}\Big)
		\,\nu(\mathbf{b})\,\dif \mathbf{a}\,\dif \mathbf{b},
	 \end{align*}
	for $m=1,\dots,4$,  with $(b_1,b_2,b_3,b_4)\eqdef (a_1,a_3,a_3,a_1)$. 
	 Using exactly the same calculation as in Lemma \ref{lem:ms1} and Lemma \ref{lem:ms2} we obtain that 
	 the discrete master loop equations \eqref{eq:ws-m1} and  \eqref{eq:ws-m2} converge to   
	  \begin{align*}
	 	&2(\p_{t_1}-\p_{t_4}) \E W_s+2I_1^s
	 	=\E W_s\;,
\\
	  	&2(\p_{t_2}-\p_{t_3}) \E W_s+2I_3^s
	  	=\E W_s\;.
	  \end{align*} 
	  The result then follows.
	\end{proof}

Recall $l^\eps_1,l^\eps_2$ in \eqref{e:before-merge} but now write 
$l_2^\eps=(e_4^{\eps})^{-1}\underline{e}^\eps e_2^\eps B^\eps $
where $\underline{e}^\eps= e^\eps$.
Let us also use a shorthand notation 
$\mathrm{\bf MM}^\eps(\ee)$ to refer to the lattice Makeenko--Migdal equation  \eqref{eq:ws} 
with a particular bond $\ee$.

Once we have Theorem~\ref{th:g2}  we can again extend it to general linear combinations. The same argument as in Corollary~\ref{cor:lin-comb} together with Theorem~\ref{th:g2} gives:
 \bc \label{cor:merger}
Assuming $\sum_{i=0}^4a_i=1$, $\sum_{j=1}^2b_j=1$, then
 \begin{align*}
 	b_1\, \mathrm{\bf MM}^\eps(\ee)+b_2\, \mathrm{\bf MM}^\eps(\underline{\ee})
	-a_0\E W_s^\eps
	-\sum_{i=1}^4 a_i \,\mathrm{\bf MM}^\eps (\ee_i)
	\;\;\stackrel{\eps\to 0}{\longrightarrow}\;\; \eqref{e:LevyThm31}.
 \end{align*}
 \ec 

 \subsection{Extension to $SU(N)$ and $SO(N)$}\label{sec:e2}

We now turn to a more general formulation which includes the group $G=U(N)$ together with $SU(N)$, and  $SO(N)$.  The argument is essentially the same; the main difference being the multiplicative constants arising in the master loop equation together with two additional operations at the discrete level: twisting and extension.  To this end, 
we introduce \footnote{We borrow this notation from L\'evy in \cite{Levy11} but also caution the reader that we are deviating here from some references on lattice Yang-Mills, where $\beta$ denotes an arbitrary couping constant.}
\[
\beta=
\begin{cases}
    1, & \text{if } G=SO(N) \\
    2, & \text{if } G\in \{SU(N),U(N)\}
\end{cases}
\qquad
\gamma=
\begin{cases}
    1, & \text{if } G=SU(N) \\
    0, & \text{if } G\in \{SO(N),U(N)\}\;.
\end{cases}
\]
   To write the Makeenko--Migdal equations in a unified way, c.f. \cite[Proposition 7.3]{Dah2022II}, it is convenient to incorporate this parameter into the action and consider instead 
\begin{equation}
S^{\eps}(Q)=e^{- \frac{\beta N}{2 \eps^{2}} \text{Tr}(I-Q)} \label{EE41},
\end{equation}  
together with the corresponding Yang Mills measure.  The associated inner product is 
\begin{equation}\label{def:inn1}
\langle X,Y  \rangle= \frac{\beta N }{2}\text{Tr}(X^{*}Y).
\end{equation}

As in Remark~\ref{rem}, with the above choice  \eqref{def:inn1} of inner product,
for the standard representation  $\tau$ of  $SU(N)$ on $\C^N$, $d_\tau=N$ and $c_\tau = -1+\frac{1}{N^2}$. 
Regarding the standard representation  $\tau$ of  $SO(N)$ on $\R^N$, 
one can extend it to a representation of $SO(N)$ on $\C^N$ viewed as the complexification of $\R^N$, namely we apply $SO(N)$ matrices to the real and imaginary parts of $\C^N$ vectors separately. This is a unitary representation since it preserves the standard Hermitian inner product on $\C^N$.
One has 
$d_\tau=N$ and $c_\tau = -1+\frac{1}{N}$.

We recall the  two additional operations.

{\it Twisting.} Given a loop $l=aebec$ (where $e$ is a bond appearing twice at locations $x$ and $y$), the {\it negative twisting} of $l$  is a loop
\[
\propto_{x,y} l \eqdef ab^{-1}c\;.
\]
For  $l=aebe^{-1}c$ (where $e$ and $e^{-1}$ appear at locations $x$ and $y$ respectively),
the {\it positive twisting} of $l$ is a loop
\[
\propto_{x,y} l \eqdef aeb^{-1}e^{-1}c \;.
\]
We write  $\mT^{+}_e((x,y);l)$ (resp. $\mT^{-}_e((x,y);l)$)
for the set of  loops 
obtained from positive (resp. negative) twisting of $l$ with respect to the bond $e$ at locations $(x,y)$.
In fact, once we fix $e$ and the locations $x,y$,  then each of
$\mT^{+}_e((x,y);l)$  and $\mT^{-}_e((x,y);l)$ only contains one possible loop.

{\it Expansion.}
Finally, a positive expansion of $l$ at location $x$ by a plaquette  $p$ passing through the bond $e^{-1}$ replaces $l$ with the pair of loops $(l,p)$.  
A negative expansion of $l$ at location $x$ by a plaquette  $p$ passing through the bond $e$ replaces $l$ with the pair of loops $(l,p)$.  
The sets $\mE^{+}_e(x;l)$ and $\mE^{-}_e(x;l)$ consist of all loops obtained from positive or negative expansions of $l$ with respect to $e$ at location $x$, respectively.

As in \eqref{e:ignore-loc} we write 
$ \mT^{\pm}_e(s) = \mT^{\pm}_e((x,y);s)$, 
$\mE^{\pm}_e(s) = \mE^{\pm}_e(x;s)$ 
for simplicity.

\medskip

We consider the same loops $l$, $l_1$, $l_2$ in \eqref{e:genl}\eqref{e:genl12}
and their lattice approximations  \eqref{e:genloop}\eqref{e:genloop12}
 as in Section \ref{sec:gen}. \footnote{If instead of \eqref{e:genloop} we consider lattice loop in which  $e$ and $e^{-1}$  appear, then as in Remark~\ref{re:ee} we will find the same limiting equation.}
We will apply the lattice loop equations proved in \cite[Theorem 6.104]{CPS2023} 
with choices of bonds similarly as in Section \ref{sec:gen}.

We start by applying the single location  master loop equation at the bond $\ee$ in the edge $e^\eps$, which leads to
\begin{align}
 	  &\frac{1}{2 \eps^2} \!\!\! \sum_{l'\in \mathbb{D}^-_{\ee}(l)}\!\!\E W_{l'}^\eps-\frac{1}{2\eps^2} \!\!\! \sum_{l'\in \mathbb{D}^+_{\ee}(l)} \!\!\E W_{l'}^\eps \nonumber\\
	  &=\big (1-\frac{(2-\beta ) }{\beta N}-\frac{2\gamma}{N^{2}} \big ) \E	W_l^\eps+\E W_{l_1}^\eps W_{l_2}^\eps
 	-\frac{(2-\beta )}{\beta N} 
	\!\!\!\sum_{l'\in \mathbb{T}^-_{\ee}(l)}\!\!\E W_{l'}^\eps
	\nonumber \\
	& \qquad + \frac{\gamma}{2\eps^2} \!\!\! \sum_{l'\in \mathbb{E}^+_{\ee}(l)}\!\!\E W_{l'}^\eps
 	-\frac{\gamma}{2\eps^2} \!\!\! \sum_{l'\in \mathbb{E}^-_{\ee}(l)} \!\!\E W_{l'}^\eps \;.\label{EE30}
 \end{align} 
Note that based on the above choice of action, the coefficient of deformation is $\frac{\beta}{2 \eps^{2}}$ for $SO(N)$ and $\frac{\beta}{4 \eps^{2}}$ for $U(N)$ or $SU(N)$, and hence by definition of $\beta$ this gives $\frac{1}{2 \eps^{2}}$ in both cases. 
By definition of twisting, we simply have 
$$
\sum_{l'\in \mathbb{T}^-_{\ee}(l)}\!\!\E W_{l'}^\eps=\E W_{l_{1}l_{2}^{-1} }^{\eps}\;,
$$
where the loop $l_{1}l_{2}^{-1}$ is the negative  twisting of $l$.
Taking the continuum limit $\eps \to 0$, we find that the above equation converges to 
\begin{equs}
 	{}&2(\p_{t_1}-\p_{t_2}+\p_{t_3}-\p_{t_4}) \E W_l+(I_1+I_3)  \\
 	&=\big (1-\frac{(2-\beta ) }{\beta N}-\frac{2\gamma}{N^{2}} \big ) \E	W_l+\E W_{l_1} W_{l_2} 	-\frac{(2-\beta ) }{\beta N}\E W_{l_{1}l_{2}^{-1}} \;,
\end{equs}
where $I_1,I_3$ are defined in \eqref{def:Im}.
Indeed, the above claim follows as in the arguments in Proposition \ref{prop:g} 
except that now the metric is given by \eqref{def:inn1} (in particular  the gradient, Laplace--Beltrami operator, the heat kernel and inner product in $I_1$, $I_3$ are defined with respect to the new metric \eqref{def:inn1}). 
Moreover, when we apply Driver's formula to expansion terms as in Section \ref{sec:gen}, we can reduce the contribution from expansion terms 
$$\frac{\gamma}{2\eps^2} \!\!\! \sum_{l'\in \mathbb{E}^+_{\ee}(l)}\!\!\E W_{l'}^\eps
-\frac{\gamma}{2\eps^2} \!\!\! \sum_{l'\in \mathbb{E}^-_{\ee}(l)} \!\!\E W_{l'}^\eps $$
 to the following integral: 
 \begin{equs}[e:expansions]
	&\frac{\gamma}{2\eps^2} \int \tr((Q-Q^{-1}))\tr(Q_l) 
	S^\eps_{\frac{t_i(\eps)-\eps^2}{\eps^2}}(Qa_i\alpha_ia_{i+1}^{-1})S^\eps(Q)\dif Q
	\\&\times\prod_{k\in\{1,2,3,4\}\backslash\{i\}}
	S^\eps_{\frac{t_k(\eps)}{\eps^2}}(a_k\alpha_ka_{k+1}^{-1})\nu^\eps(\mathbf{b})\dif \mathbf{a}\dif \mathbf{b} ,
\end{equs} 
where  $Q_l$ is a product of variables $\mathbf{a}, \mathbf{b}$. 
Here $Q$ and $Q^{-1}$ are the holonomies around the new plaquettes in the positive and negative expansions, and $i=2,4$ since we can expand a plaquette into the face $F_2$ or $F_4$, see Fig.~\ref{fig:expansions}.
 As $\eps \to 0$ we apply Lemma \ref{lem:J1} to find that the first line of \eqref{e:expansions} goes to
 \begin{equ}
 \sum_{j=1}^{d(\mfg)}\tr(L_j )\tr(Q_l)\cL_jp_{t_i}( a_i\alpha_ia_{i+1}^{-1})=0,
 \end{equ}
thanks to the fact that $\tr(L_j)=0$  since $L_j$ is in the Lie algebra of $SO(N)$ and $SU(N)$.
Remark that $\tr(L_j)=0$ does not hold for $U(N)$, but the lattice loop equations in the $U(N)$ case studied in the previous sections do not have expansion terms.
Since the integral over the other variables in \eqref{e:expansions} is obviously bounded, 
\eqref{e:expansions} vanishes as $\eps\to 0$.

 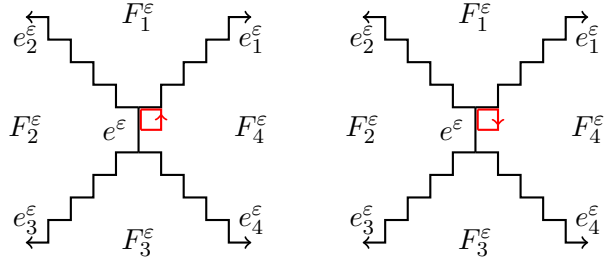
\begin{figure}[h] 
   \centering
\[
\begin{tikzpicture}[scale=3,baseline=5]
\draw[thick] (0,-0.1) -- (0,0.1)   node[midway,left] {$e^\eps$};
\draw[thick,->] (0,0.1)
\foreach \i in {1, 2, ..., 9} {-- ++({(0.5-0.5*(-1)^\i)*0.1}, {(0.5+0.5*(-1)^\i)*0.1})};
\node at (0.5,0.4) {$e_1^\eps$};
\draw[thick,->] (0,-0.1)
\foreach \i in {1, 2, ..., 9} {-- ++({(0.5-0.5*(-1)^\i)*0.1}, {-(0.5+0.5*(-1)^\i)*0.1})};
\node at (0.5,-0.4) {$e_4^\eps$};
\draw[thick,->] (0,0.1)
\foreach \i in {1, 2, ..., 9} {-- ++({-(0.5-0.5*(-1)^\i)*0.1}, {(0.5+0.5*(-1)^\i)*0.1})};
\node at (-0.5,0.4) {$e_2^\eps$};
\draw[thick,->] (0,-0.1)
\foreach \i in {1, 2, ..., 9} {-- ++({-(0.5-0.5*(-1)^\i)*0.1}, {-(0.5+0.5*(-1)^\i)*0.1})};
\node at (-0.5,-0.4) {$e_3^\eps$};
\draw[thick,red,midarrow] (0.01,0) -- (0.1,0) -- (0.1,0.09) -- (0.01,0.09);
\draw[thick,red]  (0.01,0.09) -- (0.01,0);
\node at (0, 0.5) {$F_1^\eps$};\node at (0, -0.5) {$F_3^\eps$};
\node at (0.5,0) {$F_4^\eps$};\node at (-0.5,0) {$F_2^\eps$};
\end{tikzpicture}
\qquad
\begin{tikzpicture}[scale=3,baseline=5]
\draw[thick] (0,-0.1) -- (0,0.1)   node[midway,left] {$e^\eps$};
\draw[thick,->] (0,0.1)
\foreach \i in {1, 2, ..., 9} {-- ++({(0.5-0.5*(-1)^\i)*0.1}, {(0.5+0.5*(-1)^\i)*0.1})};
\node at (0.5,0.4) {$e_1^\eps$};
\draw[thick,->] (0,-0.1)
\foreach \i in {1, 2, ..., 9} {-- ++({(0.5-0.5*(-1)^\i)*0.1}, {-(0.5+0.5*(-1)^\i)*0.1})};
\node at (0.5,-0.4) {$e_4^\eps$};
\draw[thick,->] (0,0.1)
\foreach \i in {1, 2, ..., 9} {-- ++({-(0.5-0.5*(-1)^\i)*0.1}, {(0.5+0.5*(-1)^\i)*0.1})};
\node at (-0.5,0.4) {$e_2^\eps$};
\draw[thick,->] (0,-0.1)
\foreach \i in {1, 2, ..., 9} {-- ++({-(0.5-0.5*(-1)^\i)*0.1}, {-(0.5+0.5*(-1)^\i)*0.1})};
\node at (-0.5,-0.4) {$e_3^\eps$};
\draw[thick,red,midarrow]  (0.01,0.09) -- (0.1,0.09) -- (0.1,0) -- (0.01,0);
\draw[thick,red]  (0.01,0.09) -- (0.01,0);
\node at (0, 0.5) {$F_1^\eps$};\node at (0, -0.5) {$F_3^\eps$};
\node at (0.5,0) {$F_4^\eps$};\node at (-0.5,0) {$F_2^\eps$};
\end{tikzpicture}
\]
\caption{Two of the possible expansions: one positive  and one negative.}\label{fig:expansions}
\end{figure}

\medskip

Next we apply the single location 
master loop equation at the bonds $\ee_1,\ee_3$  to find
  \begin{align}
 	  &\frac{1}{2\eps^2} \!\!\! \sum_{l'\in \mathbb{D}^-_{\ee_{1}}(l)}\!\!\E W_{l'}^\eps-\frac{1}{2\eps^2} \!\!\! \sum_{l'\in \mathbb{D}^+_{\ee_{1}}(l)} \!\!\E W_{l'}^\eps \nonumber\\
	  &=\big (1-\frac{(2-\beta ) }{\beta N}-\frac{\gamma}{N^{2}} \big ) \E	W_l^\eps	 + \frac{\gamma}{2\eps^2} \!\!\! \sum_{l'\in \mathbb{E}^+_{\ee_{1}}(l)}\!\!\E W_{l'}^\eps-\frac{\gamma}{2\eps^2} \!\!\! \sum_{l'\in \mathbb{E}^-_{\ee_{1}}(l)} \!\!\E W_{l'}^\eps\;, \label{EE31}\\
	   &\frac{1}{2\eps^2} \!\!\! \sum_{l'\in \mathbb{D}^-_{\ee_{3}}(l)}\!\!\E W_{l'}^\eps-\frac{1}{2\eps^2} \!\!\! \sum_{l'\in \mathbb{D}^+_{\ee_{3}}(l)} \!\!\E W_{l'}^\eps \nonumber\\
	  &=\big (1-\frac{(2-\beta ) }{\beta N}-\frac{\gamma}{N^{2}} \big ) \E	W_l^\eps	 + \frac{\gamma}{2\eps^2} \!\!\! \sum_{l'\in \mathbb{E}^+_{\ee_{3}}(l)}\!\!\E W_{l'}^\eps-\frac{\gamma}{2\eps^2} \!\!\! \sum_{l'\in \mathbb{E}^-_{\ee_{3}}(l)} \!\!\E W_{l'}^\eps\;.\label{EE32}
 \end{align}
In the continuum limit, following the arguments in Lemmas \ref {lem:ms1} and  \ref {lem:ms2}, these converge towards
\begin{equs}
 2(\p_{t_1}-\p_{t_4})\E W_l+2I_1&=\big (1-\frac{(2-\beta ) }{\beta N}-\frac{\gamma}{N^{2}} \big ) \E	W_l \;,  \\
 2(\p_{t_3}-\p_{t_2})\E W_l+2I_3&=\big (1-\frac{(2-\beta ) }{\beta N}-\frac{\gamma}{N^{2}} \big ) \E	W_l \;. 
\end{equs}
In particular, the expansion terms vanish, by the same argument below \eqref{e:expansions}.

\medskip
 
 Combining these observations, we obtain the following result
 \bt \label{thm:SUSO}
 For the Yang-Mills measure associated to the action \eqref{EE41}, the following linear combination of the discrete master loop equation
 \begin{align*}
\eqref{EE30}-\frac{1}{2}\eqref{EE31}-\frac{1}{2}\eqref{EE32} 
 \end{align*}
 converges to 
\begin{align}\label{eq:Wlg}
 	(\p_{t_1}-\p_{t_2}+\p_{t_3}-\p_{t_4}) \E W_l=\E W_{l_1} W_{l_2} -\frac{(2-\beta ) }{\beta N}\E W_{l_{1}l_{2}^{-1} }-\frac{\gamma}{N^{2}}\E	W_l .
\end{align}
 \et 
Notice that \eqref{eq:Wlg} corresponds to \cite[Propostion 7.3]{Dah2022II} \footnote{Note that the signs before the last two terms on the RHS of \eqref{eq:Wlg} do not align  with the master loop equation in \cite[Propostion 7.3]{Dah2022II}. This is merely a minor typo there: concerning the proof of master loop equation in \cite{Driver17}, by  substituting \cite[(67)]{Dah2022II} into \cite[(2.13)]{Driver17} we get \eqref{eq:Wlg}.}.

Using the same shorthand notation
$\mathrm{\bf MM}^\eps(\ee)$ as in Corollary~\ref{cor:merger} to refer to the lattice Makeenko--Migdal equation  
with a particular bond $\ee$,
the same argument as in Corollary~\ref{cor:lin-comb} together with Theorem~\ref{thm:SUSO} gives:
 \bc 
Assuming $\sum_{i=0}^4a_i=1$, $\sum_{j=1}^2b_j=1$, then
 \begin{align*}
 	b_1\, \mathrm{\bf MM}^\eps(\ee)+b_2\, \mathrm{\bf MM}^\eps(\underline{\ee})
	-a_0 \big (1-\frac{(2-\beta ) }{\beta N}-\frac{\gamma}{N^{2}} \big ) \E W_l^\eps
	-\sum_{i=1}^4 a_i \,\mathrm{\bf MM}^\eps (\ee_i)
 \end{align*}
 converges to \eqref{eq:Wlg}.
 \ec

\appendix
\renewcommand{\appendixname}{Appendix~\Alph{section}}
\renewcommand{\theequation}{A.\arabic{equation}}

\section{Proof of \eqref{eq:A2}}
The idea is to decompose the $\sum_{\tau\notin\cF}$ to the sum over three regions:
$$|c_\tau|\leq  \eps^{-2},\qquad \eps^{-2}<|c_\tau|\leq  (\eps^{-2})^{1+\delta},\qquad |c_\tau|>  (\eps^{-2})^{1+\delta},$$
for some fixed $\delta\in (0,1)$. We use $\cA_{2i}, i=1, 2, 3$ to denote these sums, respectively. 
Using 
\cite[Corollary A.4]{BS83} we have that for $|c_\tau|\leq \eps^{-2}$, 
\begin{align}\label{bd:atau}
	a_\tau(\eps)\leq e^{\frac{3c_\tau \eps^2}{8}+b\eps^2}\;,
\end{align}
for some constant $b>0$. 
Hence, we use \eqref{bd:atau} to obtain
\begin{align*}
	\cA_{21} &\lesssim		\sum_{\tau\notin \cF, |c_\tau|\leq \eps^{-2} }d_\tau^2a_{\tau}(\eps)^{\frac{t(\eps)}{\eps^2}}(|c_\tau|^{1/2}+(1+|c_\tau|^{2})\eps^2)
	\\& \lesssim  e^{b t(\eps)}\sum_{\tau\notin \cF,  }d_\tau^2e^{\frac{3c_\tau t(\eps)}{8}}(1+|c_\tau|)\lesssim e^{b t}\sum_{\tau\notin \cF }d_\tau^2e^{\frac{3c_\tau ( t+1)}{8}}(1+|c_\tau|)\;.
\end{align*}
This can be made arbitrarily small by choosing $\cF$ large enough since 
\begin{align*}
	\sum_{\tau  }d_\tau^2e^{\frac{3c_\tau ( t+1)}{8}}(1+|c_\tau|)=(1-\Delta)e^{\frac{3(t+1)}{8}\Delta }I<\infty\;.
\end{align*}
For $|c_\tau|\leq (\eps^{-2})^{1+\delta}$ we use \cite[(A.41)]{BS83} to  choose $\eps$ small enough such that
\begin{align}\label{bd:atau1}
	a_\tau(\eps)\leq e^{-c}<1\;,
\end{align}
for some $c>0$. By \cite[(A.42)]{BS83} we also have
\begin{align}\label{bd:dq}
	\sum_{\eps^{-2}<|c_\tau|\leq  (\eps^{-2})^{1+\delta}}d_\tau^2\lesssim (\eps^{-2})^{\frac{p}2(1+\delta)},
\end{align}
with $p=d(\mfg)$. 
Hence, we use \eqref{bd:atau1} and \eqref{bd:dq} to obtain that
\begin{align*}
	\cA_{22}&\lesssim  \sum_{\eps^{-2}<|c_\tau|\leq  (\eps^{-2})^{1+\delta} }d_\tau^2a_{\tau}(\eps)^{\frac{t(\eps)}{\eps^2}}(|c_\tau|^{1/2}+(1+|c_\tau|^{2})\eps^2)
	\\&\lesssim e^{-c \frac{t(\eps)}{\eps^2}}(\eps^{-2})^{1+2\delta}\sum_{\eps^{-2}<|c_\tau|\leq  (\eps^{-2})^{1+\delta}}d_\tau^2
	\\&\lesssim e^{-c \frac{t(\eps)}{\eps^2}}(\eps^{-2})^{(\frac{p}2+1)(1+2\delta)},
\end{align*}
which goes to zero by choosing $\eps$ small enough. 

For $|c_\tau|>(\eps^{-2})^{1+\delta}$ we use \cite[(A.54)]{BS83} to obtain
\begin{align*}
	d_\tau^2a_\tau(\eps)^2\leq e^{-c'\sqrt{|c_\tau |\eps^2}}.
\end{align*}
Hence, we obtain
\begin{align*}
	\cA_{23}\lesssim&  \sum_{|c_\tau|> (\eps^{-2})^{1+\delta} }d_\tau^2a_{\tau}(\eps)^{\frac{t(\eps)}{\eps^2}}(|c_\tau|^{1/2}+(1+|c_\tau|^{2})\eps^2)
	\\\lesssim& \sum_{|c_\tau|> (\eps^{-2})^{1+\delta} }e^{-c't(\eps)\sqrt{|c_\tau|/(\eps^2)} }(|c_\tau|^{1/2}+(1+|c_\tau|^{2})\eps^2)\;.
\end{align*}
Using \eqref{bd:dq} we know that the number of representations $\tau$ obeying $K-1<|c_\tau|\leq K$ is bounded by $C K^{p/2}$ for some universal constant $C>0$. 
Hence, 
\begin{align*}
	\cA_{23}\to0, \qquad \eps\to0.
\end{align*}
Hence, \eqref{eq:A2} follows.

 
\bibliographystyle{./Martin}
\bibliography{./refs}

\end{document}